\documentclass[11pt,letterpaper]{article}

\usepackage{subcaption}

\usepackage[margin=1in]{geometry}
\usepackage[utf8]{inputenc}

\usepackage{amsmath,amssymb,amsthm,mathrsfs,bm}
\usepackage{physics}
\usepackage{complexity}
\usepackage{caption}
\usepackage{algorithm}
\usepackage{algpseudocode}
\makeatletter
\providecommand{\theHALG@line}{}
\renewcommand{\theHALG@line}{\thealgorithm.\arabic{ALG@line}}
\makeatother
\usepackage{setspace}

\usepackage{graphicx}
\usepackage[colorlinks=true]{hyperref}
\usepackage[dvipsnames]{xcolor}
\definecolor{darkred}  {rgb}{0.5,0,0}
\definecolor{darkblue} {rgb}{0,0,0.5}
\definecolor{darkgreen}{rgb}{0,0.5,0}
\hypersetup{
urlcolor      = blue,         
linkcolor     = darkblue,     
citecolor     = darkgreen,    
filecolor     = darkred,
linkcolor     = darkblue
}
\usepackage{aliascnt}

\usepackage{tcolorbox}
\tcbuselibrary{breakable}
\usepackage{tikz-cd}
\usetikzlibrary{decorations.pathreplacing, positioning,calc,arrows.meta,shapes.geometric}
\pgfmathsetmacro\MathAxis{height("$\vcenter{}$")}

\usepackage[toc,page]{appendix}
\usepackage{tikz,pgfplots}\pgfplotsset{compat=1.18}
\usetikzlibrary{positioning,fit,calc}

\definecolor{darkred}{rgb}{0.5,0,0}
\definecolor{darkblue}{rgb}{0,0,0.5}
\definecolor{darkgreen}{rgb}{0,0.5,0}
\definecolor{DarkRed}{RGB}{170,0,0}

\makeatletter
\@ifpackageloaded{hyperref}{}{
  \usepackage[colorlinks=true]{hyperref}
}
\makeatother

\hypersetup{
  urlcolor  = blue,
  linkcolor = darkblue,
  citecolor = darkgreen,
  filecolor = darkred
}

\theoremstyle{plain}

\newtheorem{theorem}{Theorem}[section]
\newtheorem*{theorem*}{Theorem}

\newaliascnt{proposition}{theorem}

\aliascntresetthe{proposition}

\newaliascnt{conjecture}{theorem}

\aliascntresetthe{conjecture}

\newaliascnt{lemma}{theorem}
\newtheorem{lemma}[lemma]{Lemma}
\aliascntresetthe{lemma}

\newaliascnt{claim}{theorem}
\newtheorem{claim}[claim]{Claim}
\aliascntresetthe{claim}

\newaliascnt{fact}{theorem}
\newtheorem{fact}[fact]{Fact}
\aliascntresetthe{fact}

\newaliascnt{corollary}{theorem}
\newtheorem{corollary}[corollary]{Corollary}
\aliascntresetthe{corollary}

\theoremstyle{definition}

\newaliascnt{definition}{theorem}
\newtheorem{definition}[definition]{Definition}
\aliascntresetthe{definition}

\newaliascnt{remark}{theorem}
\newtheorem{remark}[remark]{Remark}
\aliascntresetthe{remark}

\newaliascnt{note}{theorem}

\aliascntresetthe{note}

\newaliascnt{properties}{theorem}

\aliascntresetthe{properties}

\newaliascnt{assumption}{theorem}

\aliascntresetthe{assumption}

\newaliascnt{observation}{theorem}

\aliascntresetthe{observation}

\newaliascnt{example}{theorem}

\aliascntresetthe{example}

\newaliascnt{question}{theorem}
\newtheorem{question}[question]{Question}
\aliascntresetthe{question}

\usepackage[nameinlink, capitalize]{cleveref}

\crefname{theorem}{Theorem}{Theorems}
\Crefname{theorem}{Theorem}{Theorems}
\crefname{proposition}{Proposition}{Propositions}
\Crefname{proposition}{Proposition}{Propositions}
\crefname{conjecture}{Conjecture}{Conjectures}
\Crefname{conjecture}{Conjecture}{Conjectures}
\crefname{lemma}{Lemma}{Lemmas}
\Crefname{lemma}{Lemma}{Lemmas}
\crefname{claim}{Claim}{Claims}
\Crefname{claim}{Claim}{Claims}
\crefname{fact}{Fact}{Facts}
\Crefname{fact}{Fact}{Facts}
\crefname{corollary}{Corollary}{Corollaries}
\Crefname{corollary}{Corollary}{Corollaries}
\crefname{definition}{Definition}{Definitions}
\Crefname{definition}{Definition}{Definitions}
\crefname{remark}{Remark}{Remarks}
\Crefname{remark}{Remark}{Remarks}
\crefname{properties}{Properties}{Properties}
\Crefname{properties}{Properties}{Properties}
\crefname{assumption}{Assumption}{Assumptions}
\Crefname{assumption}{Assumption}{Assumptions}
\crefname{observation}{Observation}{Observations}
\Crefname{observation}{Observation}{Observations}
\crefname{example}{Example}{Examples}
\Crefname{example}{Example}{Examples}
\crefname{protocol}{Protocol}{Protocols}
\Crefname{protocol}{Protocol}{Protocols}
\crefname{equation}{Equation}{Equations}
\Crefname{equation}{Equation}{Equations}
\crefname{section}{Section}{Sections}
\Crefname{section}{Section}{Sections}
\crefname{question}{Question}{Question}

\usepackage{enumerate}

\renewcommand{\polylog}{\mathrm{polylog}}
\renewcommand{\poly}{\mathrm{poly}}
\renewcommand{\epsilon}{\varepsilon}

\providecommand{\tr}{\operatorname{tr}}
\providecommand{\Tr}{\operatorname{Tr}}
\providecommand{\poly}{\operatorname{poly}}
\providecommand{\polylog}{\operatorname{polylog}}

\newcommand{\labs}[1]{\left\lvert #1\right\rvert}

\newcommand{\ceil}[1]{\left\lceil #1\right\rceil}

\newcommand{\under}[2]{\underbrace{#1}_{#2}}
\newcommand{\undersetbrace}[2]{\underbrace{#2}_{#1}}

\providecommand{\MathAxis}{2.5}

\providecommand{\BC}{\mathbb{C}}
\providecommand{\BE}{\mathbb{E}}

\providecommand{\BR}{\mathbb{R}}
\providecommand{\BZ}{\mathbb{Z}}

\providecommand{\CA}{\mathcal{A}}
\providecommand{\CB}{\mathcal{B}}

\providecommand{\calH}{\mathcal{H}}

\providecommand{\CL}{\mathcal{L}}

\providecommand{\CN}{\mathcal{N}}
\providecommand{\CO}{\mathcal{O}}

\providecommand{\CR}{\mathcal{R}}
\providecommand{\CS}{\mathcal{S}}

\providecommand{\CU}{\mathcal{U}}

\providecommand{\vA}{\mathbf{A}}
\providecommand{\vB}{\mathbf{B}}
\providecommand{\vC}{\mathbf{C}}

\providecommand{\vE}{\mathbf{E}}

\providecommand{\vG}{\mathbf{G}}
\providecommand{\vH}{\mathbf{H}}
\providecommand{\vI}{\mathbf{I}}

\providecommand{\vK}{\mathbf{K}}

\providecommand{\vO}{\mathbf{O}}
\providecommand{\vP}{\mathbf{P}}

\providecommand{\vQ}{\mathbf{Q}}

\providecommand{\vU}{\mathbf{U}}
\providecommand{\vV}{\mathbf{V}}

\providecommand{\vX}{\mathbf{X}}
\providecommand{\vY}{\mathbf{Y}}
\providecommand{\vZ}{\mathbf{Z}}

\providecommand{\vrho}{\boldsymbol{\rho}}
\providecommand{\vsigma}{\boldsymbol{\sigma}}
\providecommand{\vtau}{\boldsymbol{\tau}}
\providecommand{\valpha}{\boldsymbol{\alpha}}

\providecommand{\sA}{\mathsf{A}}
\providecommand{\sB}{\mathsf{B}}

\providecommand{\sL}{\mathsf{L}}

\providecommand{\sR}{\mathsf{R}}

\providecommand{\sd}{\mathsf{d}}
\providecommand{\sk}{\mathsf{k}}
\providecommand{\sD}{\mathsf{D}}
\providecommand{\sgr}{\mathsf{gr}}
\providecommand{\scgr}{\mathsf{c}}

\providecommand{\sq}{\mathsf{q}}

\newcommand{\rd}{\mathrm{d}}

\newif\ifshowrevisions
\showrevisionstrue

\usepackage[backend=biber,style=alphabetic,url=false,isbn=false,maxnames=10,minnames=3,maxalphanames=4,minalphanames=3]{biblatex}
\newcommand{\sig}{\boldsymbol{\sigma}}
\newcommand{\rhoH}{\boldsymbol{\rho}_{\vH}}

\newcommand{\ADB}{\mathsf{ADB}}
\newcommand{\FI}{\mathsf{FI}}

\newcommand{\supp}{\operatorname{supp}}
\newcommand{\dist}{\operatorname{dist}}

\newcommand{\QF}{\mathcal{Q}}

\newcommand{\Grid}{\mathsf{Grid}}

\begin{document}

\title{\Large The stationarity test: a framework for learning \\ quantum many-body systems from their thermal  states}
\author{Thiago Bergamaschi\thanks{Department of EECS, UC Berkeley, CA, USA. \href{mailto:thiagob@berkeley.edu}{thiagob@berkeley.edu}}
\thanks{Google Quantum AI.}}

\maketitle

\begin{abstract}

The task of learning the Hamiltonian interactions governing a quantum system, given samples of its thermal (or `Gibbs') states, is a foundational question at the intersection of quantum learning theory and many-body physics. In this paper, we 
draw connections to the quantum Gibbs sampling literature to introduce a natural learning algorithm we call \textit{the stationarity test}: which simply ``guesses'' the Hamiltonian, and measures the rate-of-change of local observables, under the associated detailed-balanced quantum Markov chain \cite{chen2023efficient}.

We leverage the stationarity test to address the following applications:
\begin{enumerate}
    \item To give the first learning algorithm for the underlying interaction graph, i.e. \textit{structure learning}, of lattice Hamiltonians at all temperatures, given copies of their Gibbs states.

    \item To give the first learning algorithm for the coefficients of a lattice Hamiltonian, given only copies of its \textit{thermal metastable states}, modeled as the ``local minima'' of the free energy. 
    
    \item In addition, we present a refinement to the recent area law for thermal metastable states \cite{BCV25}, which holds in the thermodynamic limit. 
\end{enumerate}

Our learning algorithms are rigorous, time-efficient, and nearly sample-optimal in system size and accuracy. At a technical level, our arguments are based on new approximate locality and convexity properties for the quantum Fisher information of these Gibbs sampling algorithms.

\end{abstract}

\newpage

\setcounter{tocdepth}{2}
\tableofcontents
\newpage

\section{Introduction}

The Hamiltonian learning problem is a fundamental question at the intersection of learning theory, condensed matter physics, quantum metrology, and device calibration \cite{da_Silva_2011, Shabani11, Wiebe2_2014, Wiebe1_2014, Wang_2017, Bairey_2019}. One important version of this problem is inspired by statistical physics, where nature, or some experimentally motivated black-box, has prepared the equilibrium thermal (or `Gibbs') state of an unknown many-body Hamiltonian. The task is to infer and characterize the underlying Hamiltonian interactions, from measurements on its thermal state.

The classical counterpart to this problem, learning Ising or graphical models, has been studied extensively since the 80's, motivated in part by close connections to machine learning \cite{Ackley1985ALA, Hinton1986LearningAR, bresler2015efficiently, klivans2017learning, hamilton2017information, vuffray2016interaction}. In contrast, only recently has substantial progress been made in the quantum setting \cite{Bairey_2019, anshu2020sample, Haah_2024, RF2024, Fawzi_2024, artymowicz2024efficienthamiltonianlearninggibbs, artymowicz2024certifiedalgorithmsquantumhamiltonian}, culminating in breakthrough time-and-sample efficient algorithms for learning the \textit{parameters} (coefficients) of non-commuting local Hamiltonians, from their Gibbs states at any temperature \cite{bakshi2023learningquantumhamiltonianstemperature, chen2025learning}. However, these results rely on rather sophisticated constructions, such as semi-definite relaxations or complex uniqueness equations for the Gibbs state, and left intriguing directions unstudied   \cite{anshu2023surveycomplexitylearningquantum}. 

In this paper, we present a simple framework for Hamiltonian learning problems by drawing a connection to recent dissipative quantum Gibbs sampling algorithms \cite{Rall2023thermalstate,chen2023quantum,chen2023efficient, ding2024efficient, jiang2024quantum}. At the center of our framework is a primitive we coin the \textit{stationarity test}: which simply guesses the Hamiltonian, and locally checks if the sampled thermal state is (approximately) a fixed point of the Gibbs sampling algorithm \cite{chen2023efficient} associated to the guess. We leverage the stationarity test to recover the existing optimal parameter learning guarantees \cite{chen2025learning}, and to develop the first learning algorithms for two open problems in this literature.

Our first application is to the \textit{structure learning} problem, in which the task is to learn both the connectivity (i.e. interaction graph) and the coefficients of an unknown Hamiltonian.
For example, suppose we receive samples from the Gibbs distribution of an Ising model on a lattice, but without the information of how the sampled bits map to the lattice sites. Can we infer the underlying lattice? Naturally, this is a well-studied problem in the statistical inference literature \cite{bresler2008reconstruction, bento2009graphicalmodelsdifficultlearn, santhanam2012information,
abbeel2012learningfactorgraphspolynomial,
bresler2014structure}, however, has remained a long-standing open question for non-commuting Hamiltonians, outside of very recent work in the high-temperature regime \cite{lewis2026learningstructureopenquantum}. Conceptually, the obstacle is that quantum Gibbs states fundamentally do not satisfy the same exact conditional independence properties that classical ones do \cite{brown2012quantum, 
kuwahara2024clustering,chen2025GibbsMarkov, bakshi2025dobrushinconditionquantummarkov}.


Next we turn our attention to a fundamental issue with the assumptions on the input to Hamiltonian learning problems. Preparing Gibbs states at low-temperatures is computationally intractable in general (at least $\mathsf{NP}$ hard \cite{sly2010computational}), and thus it seems unrealistic to assume that nature, or an empirical black-box, is always preparing a genuine equilibrium state. An alternative, inspired by recent work on slow-mixing classical and quantum Markov chains, is that nature  has instead prepared an approximately stationary, \textit{thermal metastable state} \cite{BCV25, Gamarnik2024SlowMO, rakovszky2024bottlenecksquantumchannelsfinite, liu2024locally, Jayakumar_2026}. 
By arguments akin to gradient descent, these ``local minima'' of free energy are always polynomial-time preparable, but may be statistically far from the Gibbs state, even locally (\cref{rmk:efficiently-preparable}). 
Do samples from these metastable states suffice to recover the Hamiltonian?




Here we give rigorous, time-efficient, and nearly sample-optimal learning algorithms for both of these questions on lattice Hamiltonians, building on the stationarity test. At a technical level, we draw new connections between these learning-theoretic problems and Gibbs sampling algorithms by drawing on modern tools from quantum optimal transport theory \cite{carlen2012analog2wassersteinmetricnoncommutative, carlen2017gradientflowentropyinequalities, Datta2017RelatingRE, Carlen_2019, gao2021riccicurvaturequantumchannels, carlen2024dynamics, BCV25}. Looking forward, we hope this exchange brings forth interesting new interdisciplinary directions in quantum learning theory, as well as new tools for the study of quantum Gibbs samplers; see \cref{section:discussion} for an outlook and for further related work.

\subsection{A framework for Hamiltonian learning problems}
\label{section:main-results}

We dedicate this section to an account of our framework and applications. After some brief setup, we begin in \cref{section:intro-stat-test} with a description of the stationarity test. Then, in \cref{section:structure-learning-intro} we discuss our structure learning algorithm, and in \cref{section:metastab-intro} applications to metastable states.

\paragraph{Setup.} Let $\vH=\sum_{\gamma\in\Gamma}h_\gamma\vP_\gamma$ denote an $n$ qubit $\sk-$local Hamiltonian, on a set of non-identity Pauli interactions $\Gamma$ with coefficients $|h_\gamma|\leq 1$. We assume the \textit{interaction graph} defined by $\Gamma$ is of degree $\leq \sd$, and can be embedded on a $\mathsf{D}-$dimensional lattice.  We treat $\sk, \sd, \mathsf{D}$ as fixed constants independent of $n$. At a fixed inverse-temperature $\beta>0$, we denote the Gibbs state of $\vH$ as:
\begin{equation}
    \vrho_{\vH}:=\frac{e^{-\beta\vH}}{\tr[e^{-\beta\vH}]}.
\end{equation}

In the modern quantum Gibbs sampling paradigm, one models the time-evolution of a quantum state $\vsigma$ during the algorithm as a quantum Markov semigroup, generated by a Lindbladian\footnote{The dynamics depends on $\vH, \beta$ and a choice of updates (here, single-site Pauli operators); see \cref{section:prelim}.} \cite{chen2023efficient} which satisfies some form of quasi-locality, and (KMS) detailed-balance, thus fixing $\vrho_{\vH}$:
\begin{equation}
    \frac{\rd}{\rd t}\vsigma=\CL_{\vH}[\vsigma], \quad \text{where} \quad \CL_{\vH}[\vsigma]=\sum_{i\in [n]}\CL_{\vH, i} [\vsigma] \quad \text{ and }\quad \CL_{\vH, i}[\vrho_{\vH}]=0.\label{eq:lind-eq-intro}
\end{equation}
\subsubsection{The stationarity test}
\label{section:intro-stat-test}
The central primitive in our learning algorithms is the following test, which we refer to as the \textit{stationarity test}. Operationally, \cref{def:stat-test} simply ``guesses'' a Hamiltonian $\vG$, and computes the rate-of-change of a local observable $\vO$ under $\CL_{\vG}$ when evaluated on a sampled state $\vsigma$.

\begin{tcolorbox}[
    breakable,
    colback=white,
    colframe=black,
    boxrule=0.5pt,
    arc=2mm,          
    grow sidewards by=2mm, 
    left=8pt,
    right=8pt,
    top=4pt,
    bottom=12pt,      
    before skip=10pt,
    after skip=14pt   
]

\begin{definition}[The stationarity test]\label{def:stat-test}
Let $\vG$ be a candidate Hamiltonian, and $\CL_{\vG}=\sum_i\CL_{\vG, i}$ the Lindbladian dynamics \eqref{eq:lind-eq-intro}. Given a sample \(\vsigma\), an observable $\vO$, and a site $i\in [n]$, we denote:
\begin{equation}
    \mathsf{Test}_{i}[\vO,\vG;\vsigma]
    :=
    \tr\big[\vsigma\,\CL_{\vG,i}^{\dagger}[\vO]\big].
    \label{eq:intro-stationarity-test}
\end{equation}

\end{definition}
\end{tcolorbox}

When $\vO$ is a local observable, the quasi-locality of $\CL_{\vG, i}$ entails the stationarity test can be evaluated locally, and thus in parallel, which will ensure sample efficiency. We note this is not the first work to suggest (some form of) stationarity as a means to identify thermal states \cite{Bairey_2019, evans2019scalablebayesianhamiltonianlearning, artymowicz2024certifiedalgorithmsquantumhamiltonian}, including practical schemes; we defer a comparison to prior work to \cref{section:discussion}.

\begin{remark} [Completeness]\label{rmk:completeness-intro}
   The main virtue of the stationarity test is a simple \textit{completeness} condition. Namely, if given a sample of the Gibbs state $\vsigma:=\vrho_{\vH}$ and we ``guess'' correctly $\vG=\vH$, then detailed-balance implies $\mathsf{Test}_{i}[\vO,\vH;\vrho_{\vH}]=0$ for every observable $\vO$ and $i\in [n]$.
\end{remark}

The challenge lies in the converse: if the outcome of the test $\mathsf{Test}_{i}[\vO,\vG;\vrho_{\vH}]\approx 0$ (for every local observable $\vO$), then the guess Hamiltonian is close to the ground truth $\vG\approx \vH$; and thus we have correctly learned the Hamiltonian. To prove this \textit{soundness} condition, we make use of the following quantitative relationship between the stationarity test and the \textit{quantum Fisher information}.\footnote{The quantum Fisher information of $\vsigma$ relative to $\vrho_{\vG}$ measures the rate of change of the free energy (and thus the relative entropy) of $\vsigma$ when evolved under $\CL_{\vG}$, see \cref{def:relative-fisher-information} for formal definitions.}

\begin{remark}
    [The stationarity test and the Fisher information]\label{rmk:stat-fisher-intro} When input the observable $\vO=\vH-\vG$, the stationarity test measures the \textit{quantum Fisher information} of $\vrho_{\vH}$ relative to $\vrho_{\vG}$:
    \begin{align}
         -\frac{\rd}{\rd t} \mathsf{D}\bigg(e^{t\CL_{\vG, i}}[\vrho_{\vH}]\bigg|\bigg|\vrho_{\vG}\bigg)\bigg|_{t=0} := \FI_{\vG,i}\big(\vrho_{\vH}\big|\big| \vrho_{\vG}\big)= \beta\cdot \mathsf{Test}_{i}[\vH-\vG,\vG;\vrho_{\vH}] 
    \end{align}
\end{remark}

The crux of the framework will be that the Fisher information can ultimately
be bounded \textit{even without knowledge of} $\vH$, simply from quasi-local measurements on samples of $\vrho_{\vH}$, see \cref{section:overview} for further explanation. Next, we discuss the learning algorithms we build out of said test.

\subsubsection{Structure learning of lattice Hamiltonians}
\label{section:structure-learning-intro}

Our first application is to give the first structure learning algorithm for quantum lattice Hamiltonians, given samples of their Gibbs states, at any temperature.

\begin{figure}[!b]
    \centering
    \includegraphics[width=0.47\linewidth]{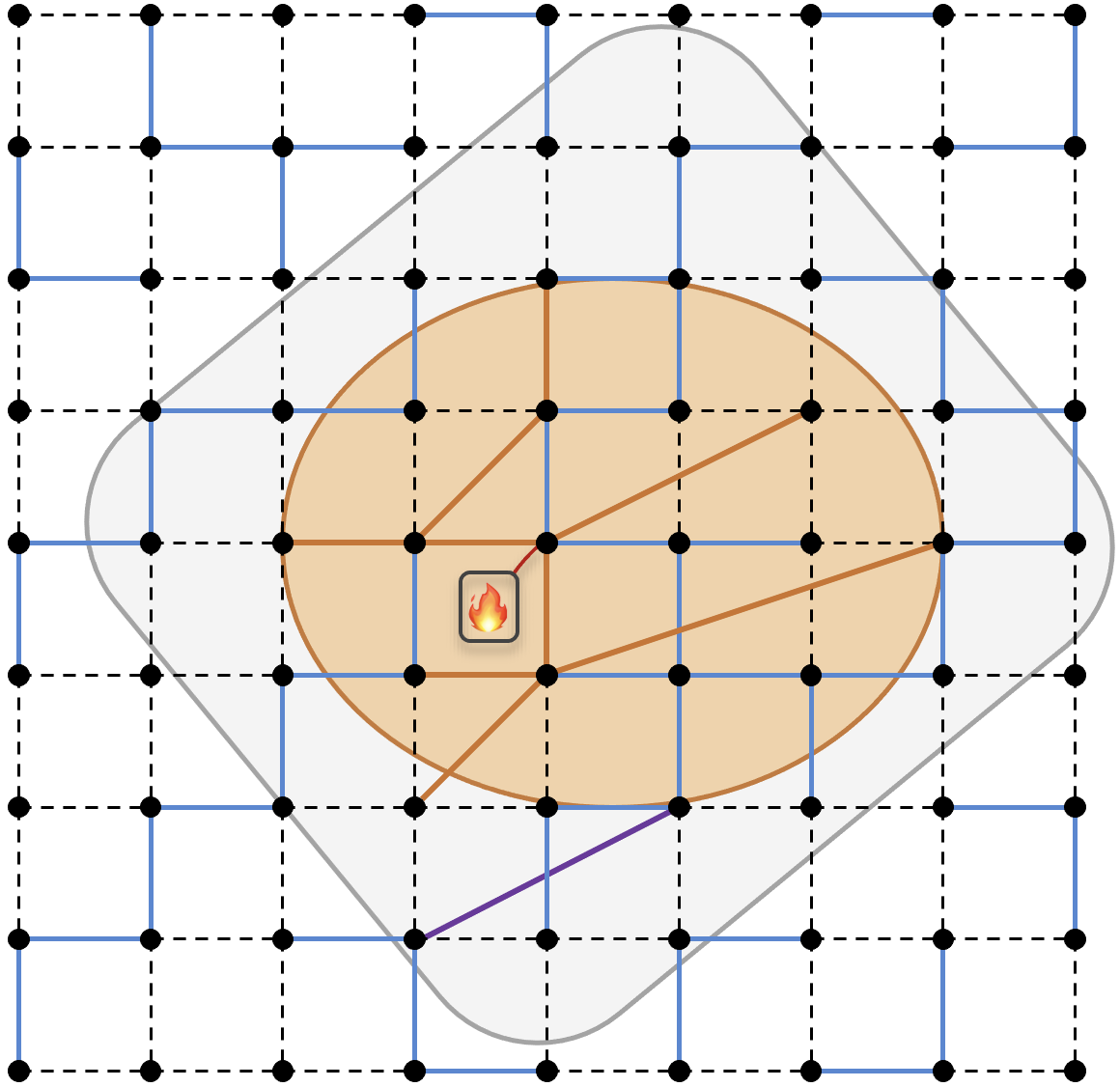}
    \caption{A cartoon of the ``neighborhood guessing'' structure learning algorithm. A ground truth 2-local Hamiltonian $\vH$ lies on the edges of an unknown square lattice. The current guess Hamiltonian $\vG$ (edges in blue) is a strict subgraph of $\vH$ (missing edges in dashed black). In each iterative step, we ``guess'' the neighborhood of a qubit, by perturbing $\vG+\vU$ with a small cluster $\vU$ of 2-local terms (orange). The stationarity test then evaluates the rate of change of 2-local observables (purple), evolved under $\CL_{\vG+\vU}$, truncated to a finite radius (gray) around the origin.}
    \label{fig:st-learn}
\end{figure}

\begin{tcolorbox}[
    breakable,
    colback=white,
    colframe=black,
    boxrule=0.5pt,
    arc=2mm,
    grow sidewards by=2mm,  
    left=8pt,
    right=8pt,
    top=5pt,
    bottom=5pt,
    before skip=8pt,
    after skip=5pt
]

\begin{theorem}[Structure learning from quantum Gibbs states, informal]
\label{theorem:structure-learning-intro}
Fix $\eta\in(0,1)$, $\beta>0$. Let $\vH=\sum_{\gamma\in \Gamma_{\vH}}h_\gamma\vP_\gamma$ be an $n$ qubit, lattice Hamiltonian on an unknown set of interactions $\Gamma_{\vH}$. Then, there exists an algorithm which on input at most
\begin{equation}
    \mathsf{N_{samp}}:=
  \frac{e^{\poly(\beta)}}{\beta^2}
  \cdot O\bigg(\frac{\polylog \, \eta^{-1}}{\eta^2} \cdot \log n\bigg),
  \label{eq:structure-complexity-intro}
\end{equation}
samples of $\vrho_{\vH}$, with probability $\geq 0.99$, outputs a Hamiltonian $\vG=\sum_{\gamma\in \Gamma_{\vG}}g_\gamma\vP_\gamma$ satisfying:
\begin{align}
    \textsf{(structure learning)}\quad  \Gamma_{\vG} \subseteq  \Gamma_{\vH}, \qquad  \textsf{(parameter learning)} \quad \max_{\gamma} |g_\gamma-h_\gamma|\leq \eta.
\end{align}
The time complexity of the algorithm is at most $n^{\poly(\beta)} \cdot \poly(\mathsf{N_{samp}})$.
\end{theorem}

\end{tcolorbox}

\begin{remark} [Optimality]
    The sample complexity of the algorithm above is optimal in the system size $n$ and, up to logarithmic factors, in the accuracy $\eta$, matching the lower bound of \cite{Haah_2024}. However, it need not be optimal in $\beta$, Hamiltonian locality parameters, nor time-complexity. 
\end{remark}

We refer the reader to \cref{fig:st-learn} for an description of (a single iterative step of) the algorithm. Roughly speaking, our approach to \cref{theorem:structure-learning-intro} is largely inspired by the parameter learning algorithm of \cite{chen2025learning}, which proceeds by iteratively refining a guess Hamiltonian $\vG^{(1)}, \vG^{(2)}\cdots, \vG^{(r)}$ to increasing accuracy, by performing local perturbations to the Hamiltonian. Here, we replace their iterative step by careful application of the stationarity test; and crucially, maintain the invariant that in each iteration ``more of the structure of $\vH$ is revealed'':
\begin{equation}
    \Gamma_{\vG^{(i)}} \subseteq \Gamma_{\vG^{(i+1)}}\subseteq \cdots \subseteq  \Gamma_{\vH}.
\end{equation}

As a warm-up to our proof approach, we dedicate \cref{section:gibbs-alg} to a introductory derivation of prior parameter learning statements in our framework, recovering the results of \cite{bakshi2023learningquantumhamiltonianstemperature, chen2025learning} where the underlying graph is known. See \cref{section:overview} for further discussion on the proof techniques.

\subsubsection{Applications to thermal metastable states}
\label{section:metastab-intro}

\textbf{Definitions.} We next turn our attention to thermal metastable states. We operate under the definition introduced in \cite{BCV25}, by casting metastable states as the \textit{approximately stationary states} of the dissipative Gibbs sampling algorithm of \cite{chen2023efficient}:

\begin{definition}[Metastability as approximate stationarity {\cite{BCV25}}]
\label{def:intro-local-metastability}
A state \(\vsigma\) is said to be \(\epsilon\)-locally metastable\footnote{\cite{BCV25}  proposed a ``global'' definition, wherein $\vsigma$ is approximately stationary under the sum $\CL_{\vH}=\sum_i\CL_{\vH, i}$. However, these definitions were shown to be equivalent up to polynomial factors, and to imply a local minima condition for the free energy \cite[Theorem C.2, equation (2.2)]{BCV25}. See \cref{section:discussion} for further discussion.} with respect to a Hamiltonian
\(\vH\) if, for every site \(i\in[n]\),
\begin{equation}
    \left\|\CL_{\vH,i}[\vsigma]\right\|_1\leq\epsilon.
    \label{eq:intro-local-metastability}
\end{equation}
\end{definition}

The Gibbs state $\vrho_{\vH}$ is the unique exactly stationary state of $\CL_{\vH}$, however, a metastable state $\vsigma$ may be arbitrarily far from $\vrho_{\vH}$ in statistical distance (e.g. \cref{fig:dw-ising}).  
The virtue in the definition above is that although $\vrho_{\vH}$ may be worst-case hard to prepare \cite{sly2010computational}, a simple ``gradient descent'' argument implies metastable states are always efficiently preparable, even at low temperatures. 

\begin{lemma}
    [Metastable states are efficiently preparable, {\cite[Lemma II.1 and Lemma E.1]{BCV25}}]
    \label{rmk:efficiently-preparable}
   Let $\vsigma_t=e^{t\CL_{\vH}}[\vsigma_0]$ denote the time-evolution of an arbitrary initial state $\vsigma_0$. For time $t$ chosen uniformly at random up to some finite cutoff $\mathsf{T}>0:$
    \begin{align}
        \big\|\BE_{t\sim [0, \mathsf{T}]} \CL_{\vH, i}[\vsigma_t]\big\|_1\leq \poly(n, \mathsf{T}^{-1}), \quad \text{and}\quad \BE_{t\sim [0, \mathsf{T}]} \mathsf{FI}_{\vH, i}(\vsigma_t||\vrho_{\vH})\leq \poly(n, \mathsf{T}^{-1}).
    \end{align}
\end{lemma}

Consequently, by time-averaging, evolving \textit{any} state for $\poly(n)$ time always results in an inverse $\poly(n)$ metastable state \cite{liu2024locally, BH16, Bairey_2019, evans2019scalablebayesianhamiltonianlearning}. We add that the convergence to metastable is expected to be much faster in specific systems \cite{GS23Ising, Bresler2024Metastable, GSS25}.


\paragraph{Parameter learning algorithms for metastable states.} In the language above, we can give a concise re-interpretation to our framework for Hamiltonian learning.

\begin{remark}
    [Metastability and the stationarity test] By Schatten duality, the stationarity test can be interpreted as an algorithmic, local check of the metastability condition:
    \begin{align}
        \|\CL_{\vG,i}[\vsigma]\|_1=\sup_{\|\vO\|=1} |\mathsf{Test}_i[\vO, \vG;\vsigma]|.
    \end{align}
\end{remark}

 \noindent In other words, if a Gibbs state $\vrho_{\vH}$ is sufficiently metastable under a guess $\vG$, then $\vH\approx \vG$! In this context, the question of whether Hamiltonian learning can be performed from samples of thermal metastable states, is informally captured by the following  ``triangle inequality'':

\begin{question}\label{question:triangle}
    If a state $\vsigma$ is $\epsilon_{\vH}$ locally metastable w.r.t. $\vH$, and $\epsilon_{\vG}$ locally metastable w.r.t. $\vG$, then is $\vrho_{\vH}$ locally metastable w.r.t. $\vG$? If so, then $\vH\approx \vG$. See \cref{fig:dw-ising} for an illustration.
\end{question}

 We make this discussion quantitative in \cref{section:metastable-learning}, where we argue that such metastable states suffice to replace exact Gibbs samples used in our parameter learning algorithm of \cref{section:gibbs-alg}.

\begin{figure}[t]
    \centering
    \includegraphics[width=0.6\linewidth]{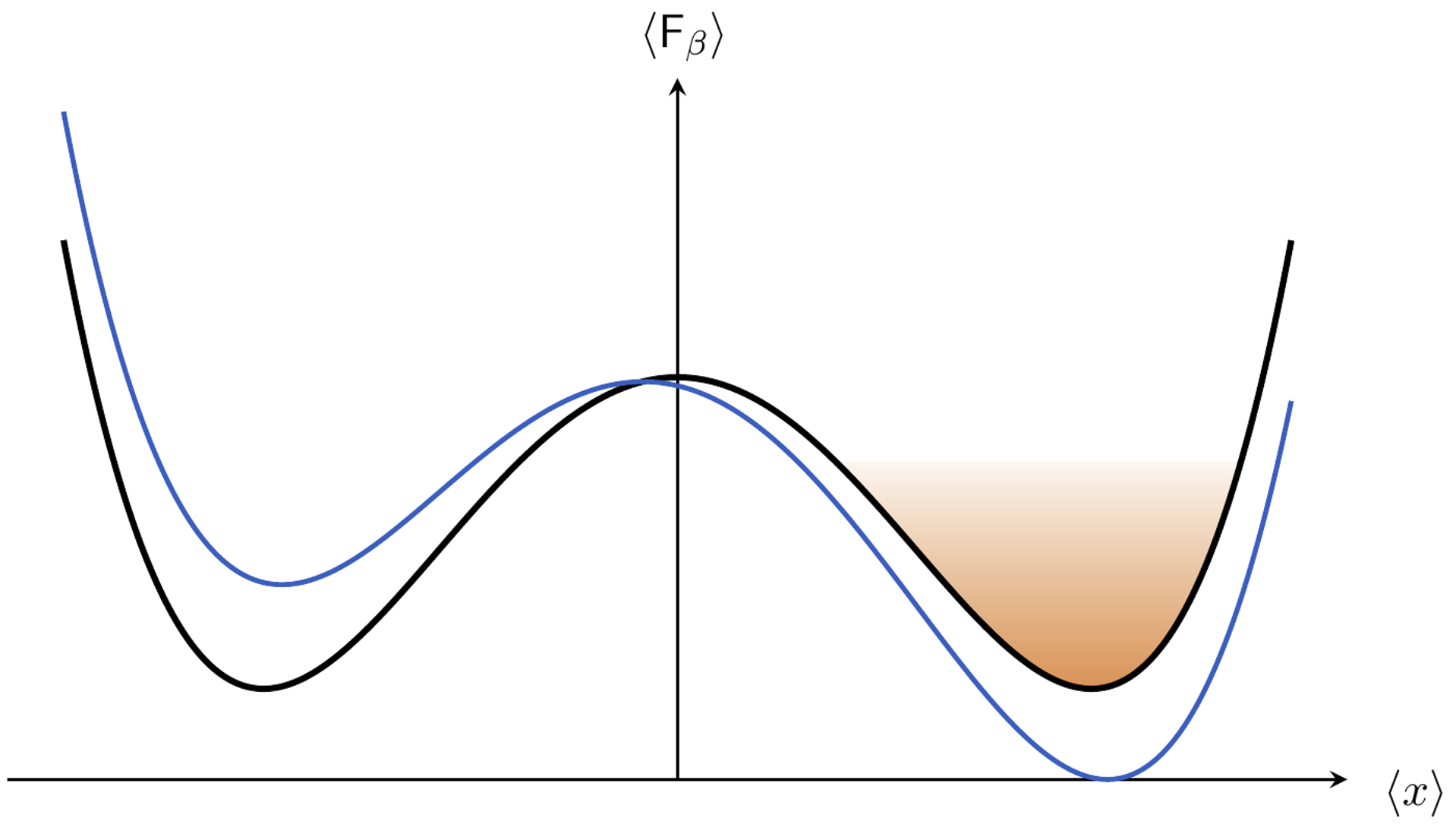}
    \caption{A cartoon of the Hamiltonian learning problem from metastable states. A metastable state (orange gradient) lies in a double well, depicting a toy model of the free energy landscape of low temperature, 2D Ising model configurations of average magnetization $x$ (black curve). The learning algorithm guesses a Hamiltonian, and attempts to ``fit'' to the free-energy landscape (blue curve), by checking if the sampled state is also metastable w.r.t. the guess. See \cref{question:triangle}.}
    \label{fig:dw-ising}
\end{figure}

\begin{tcolorbox}[
    breakable,
    colback=white,
    colframe=black,
    boxrule=0.5pt,
    arc=2mm,
    grow sidewards by=2mm,  
    left=8pt,
    right=8pt,
    top=5pt,
    bottom=6pt,
    before skip=9pt,
    after skip=14pt
]

\begin{theorem}[Parameter learning from metastable states, informal]
\label{thm:intro-metastable-learning}
Fix $\beta>0, \epsilon\geq 0$ and $\eta>\eta_{\mathsf{thr}}:=\mathsf{z}_\beta\cdot \epsilon$ for a $\beta$ dependent constant $\mathsf{z}_\beta\geq0$. Let $\vH=\sum_{\gamma\in \Gamma}h_\gamma\vP_\gamma$ be an $n$ qubit lattice Hamiltonian on a known set of interactions $\Gamma$, and $\vsigma$ be an $\epsilon-$locally metastable state with respect to $\vH$ (\cref{def:intro-local-metastability}). There exists an algorithm which on input 
\begin{equation}
    \mathsf{N_{samp}}:=
  \frac{e^{\poly(\beta)}}{\beta^2}
  \cdot O\bigg(\frac{\polylog \, \eta^{-1}}{\eta^2} \cdot \log n\bigg),
  \label{eq:metastab-complexity}
\end{equation}
samples of $\vsigma$, with probability $\geq 0.99$, outputs a Hamiltonian $\vG=\sum_{\gamma\in \Gamma}g_\gamma\vP_\gamma$ satisfying $\max_{\gamma\in \Gamma} |g_\gamma-h_\gamma|\leq \eta.$
The time complexity of the algorithm is at most $n\cdot \poly(\mathsf{N_{samp}})$.
\end{theorem}

\end{tcolorbox}

Combining \cref{thm:intro-metastable-learning} and \cref{rmk:efficiently-preparable}, one concludes that metastable states produced by the polynomial-time evolution of the Lindbladian \cite{chen2023efficient}, \textit{on any input state}, suffice to learn the parameters of the underlying Hamiltonian up to inverse-polynomial accuracy. In a sense, this presents a concrete answer to the dilemma posed in the introduction, on whether thermal samples produced efficiently ``by nature'' (modeled as a Gibbs sampler) suffice to learn the Hamiltonian.

\begin{remark}
    [The accuracy plateau $\eta_{\mathsf{thr}}(\epsilon)$] \cref{thm:intro-metastable-learning} holds up to an accuracy threshold. At first glance it seems this is unavoidable (since completeness itself is approximate, \cref{rmk:metastab-complete}), and similarly appears in related classical inquiries \cite{Jayakumar_2026}, but we are unaware of any lower bounds.
\end{remark}

\paragraph{An area law for metastable states in the thermodynamic limit.} The work of \cite{BCV25} proved that such metastable states approximately satisfy an array of structural properties akin to that of the Gibbs state, including a local Markov property \cite{brown2012quantum, kuwahara2024clustering, chen2025GibbsMarkov, bakshi2025dobrushinconditionquantummarkov} and a thermal area law of mutual information\footnote{Defined by $\mathsf{I}(\sA:\overline{\sA})_{\vsigma}:=\mathsf{S}(\vsigma_{\sA})+\mathsf{S}(\vsigma_{\overline{\sA}})-\mathsf{S}(\vsigma)$. The thermal area law for Gibbs states \cite{wolf2008area} stipulates that $\mathsf{I}(\sA:\overline{\sA})_{\vrho}$ scales with the area of the interface between $\sA$ and $\overline{\sA}=[n]\setminus \sA$.} \cite{wolf2008area, SDHS16, KAA21, LS23, DKS25}. The additive error to their approximate area law scaled with the metastability parameter $\epsilon$, exponentially in the subregion size $\sA\subseteq [n]$, and polynomially in the full system size $n$. 

As a simple corollary to the derived quasi-locality properties for the stationarity test, we are able to refine the statement in \cite{BCV25} to remove the system size dependence, thereby providing a meaningful statement in the thermodynamic limit as $n\rightarrow \infty$.

\begin{tcolorbox}[
    breakable,
    colback=white,
    colframe=black,
    boxrule=0.5pt,
    arc=2mm,
    grow sidewards by=2mm,  
    left=8pt,
    right=8pt,
    top=5pt,
    bottom=6pt,
    before skip=9pt,
    after skip=14pt
]

\begin{theorem}[An area law for metastable states]
\label{thm:intro-metastability-area-law}
Fix $\beta\geq 0$. Let \(\sA\subseteq[n]\) be a patch of qubits, $\vH$ a lattice Hamiltonian, and \(\vsigma\) a state satisfying $ \|\CL_{\vH,i}[\vsigma]\|_1\leq\epsilon$ for each $i\in \sA$. Then, the quantum mutual information between the region $\sA$ and its complement in $\vsigma$ satisfies an approximate thermal area law:
\begin{equation}
    \mathsf{I}(\sA:\overline{\sA})_{\vsigma}
    \leq
    2\beta\cdot  \|\partial_{\sA}\vH\|
    +e^{\mu|\sA|}\epsilon^\lambda,
    \label{eq:intro-metastability-area-law}
\end{equation}
where \(\mu,\lambda^{-1}\leq\poly(\beta)\), and $\partial_{\sA}\vH$ are the Hamiltonian terms incident on $\sA$ and $\overline{\sA}$.
\end{theorem}

\end{tcolorbox}

We remark that prior related work has largely focused on deriving area laws for the \textit{exact} steady states of Lindbladian evolution, and generally assumes the dynamics is fast-mixing \cite{BCLSP15, Kastoryano_2013, Firanko_2024}. From these stronger assumptions, they generally are able to derive stronger results.

\begin{remark}
    Following the counter-example in \cite[Figure 3]{BCV25}, we do not expect the dependence in the additive correction in \eqref{eq:intro-metastability-area-law} to be improved beyond sub-exponential in $|\sA|$. 
\end{remark}

\subsection{Discussion and related work}
\label{section:discussion}


\paragraph{Learning graphical models.} There is a rich literature to the question of learning Ising or graphical models; and we do not attempt to discuss its full history here. Some of the more well-known approaches include neighborhood enumeration or ``Markov blanket'' algorithms \cite{abbeel2012learningfactorgraphspolynomial,bresler2008reconstruction} (which guesses a superset of the neighborhood, and certifies conditional independence using local statistics), greedy ``influence'' algorithms \cite{bresler2015efficiently,hamilton2017information}, the pseudolikelihood method and other relaxations of conditional log-likelihood \cite{Besag75,Ravikumar2010, wu2019sparselogisticregressionlearns}, etc. Notably \cite{bresler2015efficiently, klivans2017learning} devised algorithms which run in quadratic time $\Tilde{O}(n^2)$ (with prefactors dependent on degree $\sd$, for 2-local systems), beating the ``neighborhood guessing'' cost of $O(n^\sd)$. 

We add that our approach based on the stationarity test bears close resemblance to the pseudo-likelihood method when instantiated with classical heat-bath dynamics. In fact, in recent work \cite{Jayakumar_2026} gave learning algorithms based on this method for Ising parameters and structure from (classical) metastable distributions (and without imposing lattice assumptions!); and serves as some conceptual inspiration. However, the relationship at technical level to our approach based on the Fisher information is less clear.


\paragraph{Quantum Hamiltonian learning from thermal states.} \cite{anshu2020sample} gave sample-efficient parameter learning algorithms by proving a strong convexity guarantee for the free energy, albeit the resulting algorithms were not time-efficient. 
\cite{Haah_2024} gave time-efficient and sample-optimal algorithms at high temperatures, based on an algorithmic cluster expansion. \cite{bakshi2023learningquantumhamiltonianstemperature, narayanan2024improvedalgorithmslearningquantum} gave time-and-sample efficient algorithms \textit{at all temperatures} by combining an ingenious polynomial approximation scheme for a certain detailed-balance condition, with an SDP hierarchy. As mentioned, closely related to our work is that of \cite{chen2025learning}, who gave sample-optimal (on lattices), time-efficient algorithms at all temperatures based on the ``accuracy bootstrapping'' idea  (\cref{fig:st-learn}), and a quasi-local identifiability equation resembling a detailed-balance condition. 

The stationarity test proposed here is closely related to proposals by \cite{Bairey_2019, evans2019scalablebayesianhamiltonianlearning} (based on unitary evolution) and  \cite{artymowicz2024efficienthamiltonianlearninggibbs, artymowicz2024certifiedalgorithmsquantumhamiltonian} (Energy-Entropy Balance inequalities \cite{Fawzi_2024}), with impressive numerical performance in certain cases. However, these works do not establish sample-complexity guarantees for non-commuting Hamiltonians.  A series of works, see e.g. \cite{onorati2023efficient, bluhm2026certifyinglearninglocalquantum, Baumgratz_2013, Cramer_2010}, consider learning Gibbs states to weaker error guarantees.

At the broader intersection of learning theory and many-body physics, there is also significant practical and theoretical interest in learning quantum systems directly from \textit{dynamics}, either closed-system (unitary) evolution, see e.g. \cite{Dutkiewicz_2024,  zubida2021optimalshorttimemeasurementshamiltonian, Bakshi_2024, huang2022learning, li2023heisenberglimitedhamiltonianlearninginteracting, Caro_2024}, or open-system (Lindbladian) evolution, see e.g. \cite{möbus2026robuststructurelearningklocal, heightman2026lindbladianlearningneuraldifferential, 
ivashkov2026ansatzfreelearninglindbladiandynamics, romanov2026learningarbitrarylindbladiansquantum, lewis2026learningstructureopenquantum} and references within. We further refer the reader to \cite{anshu2023surveycomplexitylearningquantum} for an excellent survey. 


\paragraph{Quantum Gibbs sampling and metastable Markov chains.} The study of metastable or slow-mixing Markov chains is a centerpiece to the statistical physics literature (see e.g. \cite{Bovier2000MetastabilityAL, Bovier2004MetastabilityIR, Bovier2005MetastabilityIR, BH16}) with theoretical and empirical applications to optimization and inference \cite{SKC94, Baldassi_2016, liu2024locally}.  Inspired by modern developments in quantum Gibbs sampling algorithms \cite{Rall2023thermalstate,chen2023quantum,chen2023efficient, ding2024efficient, jiang2024quantum}, recent work has begun to characterize the structure and applications of quantum metastable states, including bottleneck conditions \cite{Gamarnik2024SlowMO, rakovszky2024bottlenecksquantumchannelsfinite}, connections to error-correction \cite{BGL25}, 
tomography \cite{chen2026efficientshadowtomographythermal, chen2026notestrongquantummarkov, li2026unifiedframeworklocallystable}, etc. 

As mentioned, closely related to our results is a line of work on using Lindbladian evolution to prove new structural properties about Gibbs states, 
including correlation-decay and approximate Markov properties \cite{chen2025GibbsMarkov, bakshi2025dobrushinconditionquantummarkov, BC26, rosaruiz2026staticfeaturesmixingshort} and area laws \cite{BCV25}. In particular, \cite{BCV25} proved an approximate Markov property for metastable states, and reasoned said property rigorously captured a local minima condition for the free energy. Finally, we add that there is a rich literature on the study of fast-mixing quantum Gibbs sampling algorithms, see e.g. \cite{kastoryano2016commuting, capel2021modified, rouze2024efficient, rouze2024optimal, Bardet2023Rapid, DLLZ24, bakshi2025dobrushinconditionquantummarkov, BC26, B26all, basso_quantum_2025} and references within. 


\paragraph{Open questions.} We conclude by highlighting three intriguing directions. First, a natural question is whether one can achieve ``the best of both worlds'' and learn Hamiltonian structure from metastable samples, akin to the classical work of \cite{Jayakumar_2026}. While this task should fit within the framework of this paper, we do not attempt it here for conciseness. More broadly, can we find other applications of quantum metastable states to inference and optimization? See e.g. \cite{liu2024locally} for a discussion on the independent set and community detection problems.

Next, can we decrease the processing runtime of the structure learning algorithm for 2-local systems to $\Tilde{O}(n^2)$, while maintaining sample-optimality, as in  \cite{bresler2015efficiently, klivans2017learning}? It is conceivable that replacing our naive ``neighborhood guessing'' approach (of runtime $n^{\poly(\beta)}$, omitting factors of $\sd$ in the exponent) by convex optimization methods as in \cite{lewis2026learningstructureopenquantum, klivans2017learning} could provide such a route.

Finally, it remains a remarkable open problem to prove rapid-mixing in 1D non-commuting quantum systems at all temperatures; see e.g. \cite{kastoryano2016commuting, BC26, Bardet_2021} for related efforts. In part, this is since modified log-Sobolev inequalities (which capture a scalar relationship between the Fisher information, and the quantum relative entropy) \cite{temme2010chi, kastoryano2013quantum, Bardet_2021, capel2021modified, Bardet2024, BCL24} are notoriously hard to prove for quantum spin systems. The derivation of the ``approximate convexity property'' for the quantum Fisher information (see equation \eqref{eq:approx-convexity-overview} and \cref{section:fisher-identifiability}) was originally an attempt at this question, and we hope it could inspire further attempts.

\subsection{Technical overview}
\label{section:overview}

We highlight four key new ideas behind our proof approach, which roughly justify the correctness, time-complexity, the sample-complexity (respectively) of our structure learning algorithm, and then the application to metastable states. We refer the reader to \cref{section:prelim} for relevant preliminaries on weighted inner products, operator Fourier transforms (OFTs), the \cite{chen2023efficient} Lindbladian, and Lieb-Robinson bounds.  

\paragraph{The approximate convexity of the Fisher information.} The starting point to our arguments is an explicit formula for the quantum Fisher information of the \cite{chen2023efficient} Lindbladian derived in \cite{BCV25}, which captures the rate-of-decay of the relative entropy in terms of a carefully weighted norm\footnote{Here, $\rd\mu_{\omega, t, s}$ is a measure decaying in frequency and time, and $s$ is an interpolation parameter, see \cref{def:relative-fisher-information}.} of certain non-commutative spatial gradients of the log-likelihood ratio. Fixing our attention to a single-site Pauli operator $\vA^a$:
\begin{align}
    \FI_{\vG,a}\big(\vrho_{\vH}\big|\big| \vrho_{\vG}\big)= \under{\beta\cdot  \tr[\CL_{\vG, a}[\vrho_{\vH}]\cdot (\vH-\vG)]  }{\mathsf{time-derivative}} = \beta^2\cdot \iiint \under{\|
    [\hat{\vA}_{\vG}^{a}(\omega,t), \vH-\vG]\|_{\vrho_{\vH},s}^{2} }{\textsf{spatial-derivative}} \rd\mu_{\omega, t, s}.\label{eq:overview-fisher-eq}
\end{align}
Indeed, the ``moves'' or updates of the Lindbladian $\CL_{\vG, a}$ are the operator Fourier transforms $\hat{\vA}_{\vG}^{a}(\omega,t)$ of $\vA^a$ (\cref{def:oft}), and thus commutators with $\hat{\vA}_{\vG}^{a}(\omega,t)$ play the role of spatial ``gradients'' (see \cite{carlen2024dynamics} for an excellent survey of this theory of quantum optimal transport).

Ultimately, we would like to relate this Fisher information between Gibbs states (which we can efficiently measure and bound, see \cref{section:fisher-identifiability}) to the distance between $\vH$ and $\vG$. A natural, local way to quantify this distance is captured by commutators $[\vA^a, \vG-\vH]$ with the proper (unfiltered) local Pauli operator $\vA^a$ \cite{anshu2020sample, chen2025learning}; and thus the task reduces to relating the norms of the commutators with $\hat{\vA}_{\vG}^{a}(\omega,t)$ \eqref{eq:overview-fisher-eq}, to that with $\vA^a$. We prove such an \textit{approximate convexity} condition in \cref{lem:uniform-truncated-bare-reconstruction}, of the form:
\begin{align}
  \beta^2\cdot  \| [\vA, \vH-\vG]\|_{\vrho_{\vH}}^{2} \lesssim  \FI_{\vG,a}\big(\vrho_{\vH}\big|\big| \vrho_{\vG}\big) \quad + \quad \text{error},
   \label{eq:approx-convexity-overview}
\end{align}
by building on certain Dirichlet form comparison techniques developed in the quantum fast-mixing literature \cite{BC26, B26all, basso2026spectralgapdaviesgenerator} and methods to regularize OFTs \cite{chen2025GibbsMarkov, chen2025learning}.\footnote{To the expert reader, we remark that this step is a comparison between the Fisher information of the \cite{chen2023efficient} Lindbladian, and the Dirichlet form of the log-likelihood under the `pseudo-Lindbladian' generator of \cite{BC26}.} The above discussion justifies the intuition that if the stationarity test $\approx 0$, then $\vG\approx \vH$.

\paragraph{Lipschitz continuity properties for the \cite{chen2023efficient} Lindbladian.} We next sketch a means to iteratively refine the guess Hamiltonian to increasing accuracy $\eta_1, \eta_2, \cdots, \eta_r=2^{-r}$ (refer back to \cref{fig:st-learn}). 
Following \cite{chen2025learning}, in each iteration the guess $\vG^{(1)}, \vG^{(2)}, \cdots, \vG^{(r)}$ is exhaustively locally perturbed $\vG^{(r)}+ \eta_r \vU$, and (roughly speaking) we pick the candidate that minimizes the stationarity test on all local observables.  For simplicity, in the parameter learning setting of \cref{section:gibbs-alg}, assume for now that $\vU$ lies in some finite radius $\ell$ around a qubit $i\in [n]$ in the lattice. 

By design, one such guess $\vG^{(r)}+ \eta_r  \vU^*$ will be closest to $\vH$ within the radius $\ell$, and we require a means to analyze the performance of the stationarity test under said perturbation. The following Lipschitz continuity property for the \cite{chen2023efficient} Lindbladian answers this question (\cref{lemma:lipschitz}).
 \begin{align}
        \|\CL^\dagger_{\vH, a}-\CL^\dagger_{\vG, a}\|_{\infty-\infty} \lesssim  \sum_{\gamma\in \Gamma} |g_\gamma-h_\gamma|\cdot e^{-\Omega(\mathsf{dist}(a, \gamma))} \label{eq:lipschitz-overview}
    \end{align}
Crucially, if $\vH$ and $\vG^{(r)}+ \eta_r  \vU^*$ agree within the radius $\ell$ (i.e. we guessed correctly, locally) and $\vH\approx \vG^{(r)}$ outside of $\ell$ to some coarse error $\eta_r$, then the error\footnote{The convergence of \eqref{eq:lipschitz-overview} only holds for lattice Hamiltonians, which is why we only enumerate over perturbations $\vG^{(r)}+ \eta_r  \vU$ which preserve the property that the interaction graph admits polynomial growth. Furthermore, a thresholding step implies we only accept guesses $\vG^{(r)}$ which are a subgraph of $\vH$.} of replacing $\vG$ with $\vH$ in the stationarity test decays with both $\eta_r$ and $\ell$! This will enable us to choose the enumeration radius $\ell\sim \poly(\beta)$ to be independent of the target accuracy, which is the key step to ensure time-efficiency. 

\begin{remark}
    By combining \eqref{eq:lipschitz-overview} and \eqref{eq:approx-convexity-overview}, we observe (informally) that we have learned the Hamiltonian if and only if $\vrho_{\vH}$ is metastable w.r.t $\vG$: 
    \begin{equation}
     \vG\approx \vH \quad \iff \quad \CL_{\vG}[\vrho_{\vH}]\approx \CL_{\vH}[\vrho_{\vH}] = 0
    \end{equation}
\end{remark}


\paragraph{A sample-optimal neighborhood guessing algorithm.} It follows from Lieb-Robinson bounds that to evaluate the stationarity test on some local observable $\vO$, to error $\eta$, it suffices to truncate the guess Hamiltonian $\vG$ to some radius $\sR\sim \log \eta^{-1}$ around the site $i\in [n]$. Naively, if one were to perform state tomography to compute such quasi-local expectations $\CL_{\vG, i}^\dagger[\vO]$, we would require a number of samples exponential in the number of qubits in the ball \cite{kuwahara2020clustering, huang2020predicting, aaronson2018shadowtomographyquantumstates}, resulting in a quasi-polynomial $\sim 2^{\polylog \eta^{-1}}$ sample-complexity (factors of success probability omitted). 

At least in the parameter learning setting of \cref{section:gibbs-alg}, there is a simple solution to this issue. We can simply run short-time Lindbladian dynamics (defined by the truncation $\vG_{\sR}$) in parallel, across non-overlapping subsets of radius $\sR$, to get optimal sample complexity $\sim \eta^{-2} \polylog \eta^{-1}$ using a median-of-means estimator (see \cref{section:H-complexity}). Unfortunately, by itself this strategy does not give optimal results in the structure learning setting. This is since we need to enumerate over $n^{\poly(\ell, \sd, \sk)}$ possible choices of each neighborhood guess, resulting in a total of $\poly(n, \eta^{-1})$ samples. 

To get optimal sample complexity for the structure learning problem, we instead combine the two ideas above by revisiting the continuity property for the \cite{chen2023efficient} Lindbladian. Let $\mathsf{L}\gg \ell$ be a (constant) length-scale, and for any perturbation $\eta\vU$ let $\vG_\mathsf{L}, (\vG+\eta\vU)_\mathsf{L}$ denote truncations\footnote{As a technical point, this latter truncation is relative to the interaction graph of the perturbed Hamiltonian.} to radius $\mathsf{L}$. We approximate the stationarity test under the perturbation $\vG+\eta\vU$ with two terms:
\begin{align}
    \mathsf{Test}_{i}[\vO,\vG+\eta\vU;\vrho_{\vH}] \approx \under{\mathsf{Test}_{i}[\vO,\vG_{\sR};\vrho_{\vH}]}{\textsf{a ``baseline'' estimate}} +  \under{\mathsf{Test}_{i}[\vO,(\vG+\eta\vU)_\mathsf{L};\vrho_{\vH}] - \mathsf{Test}_{i}[\vO,\vG_{\mathsf{L}};\vrho_{\vH}]}{\textsf{a correction at radius } \mathsf{L}, \textsf{ to error decaying with $\eta$ and $\mathsf{L}$}}\label{eq:intro-baseline+corr}
\end{align}
\begin{enumerate}
    \item a ``baseline'' estimate under the current guess $\vG$, by truncating the Lindbladian to radius $\sR\sim \log \eta^{-1}$, and measuring the rate-of-change of observables by short-time evolution (similar to the parameter learning setting).
    \item a perturbation-dependent correction, computed using Pauli shadow tomography \cite{huang2020predicting} on all constant-sized RDMs (determined by the radius $\mathsf{L}\sim \poly(\beta)\cdot \ell$) simultaneously.
\end{enumerate}

We refer the reader to \cref{fig:st-learn,fig:optimal-sample} for informal illustrations, and \cref{lemma:parallel-stat-meas} for further details.

\begin{figure}[t]
\centering
\captionsetup[subfigure]{
    justification=centering,
    singlelinecheck=false
}
\begin{subfigure}[t]{0.46\textwidth}
  \centering
  \includegraphics[width=\linewidth]{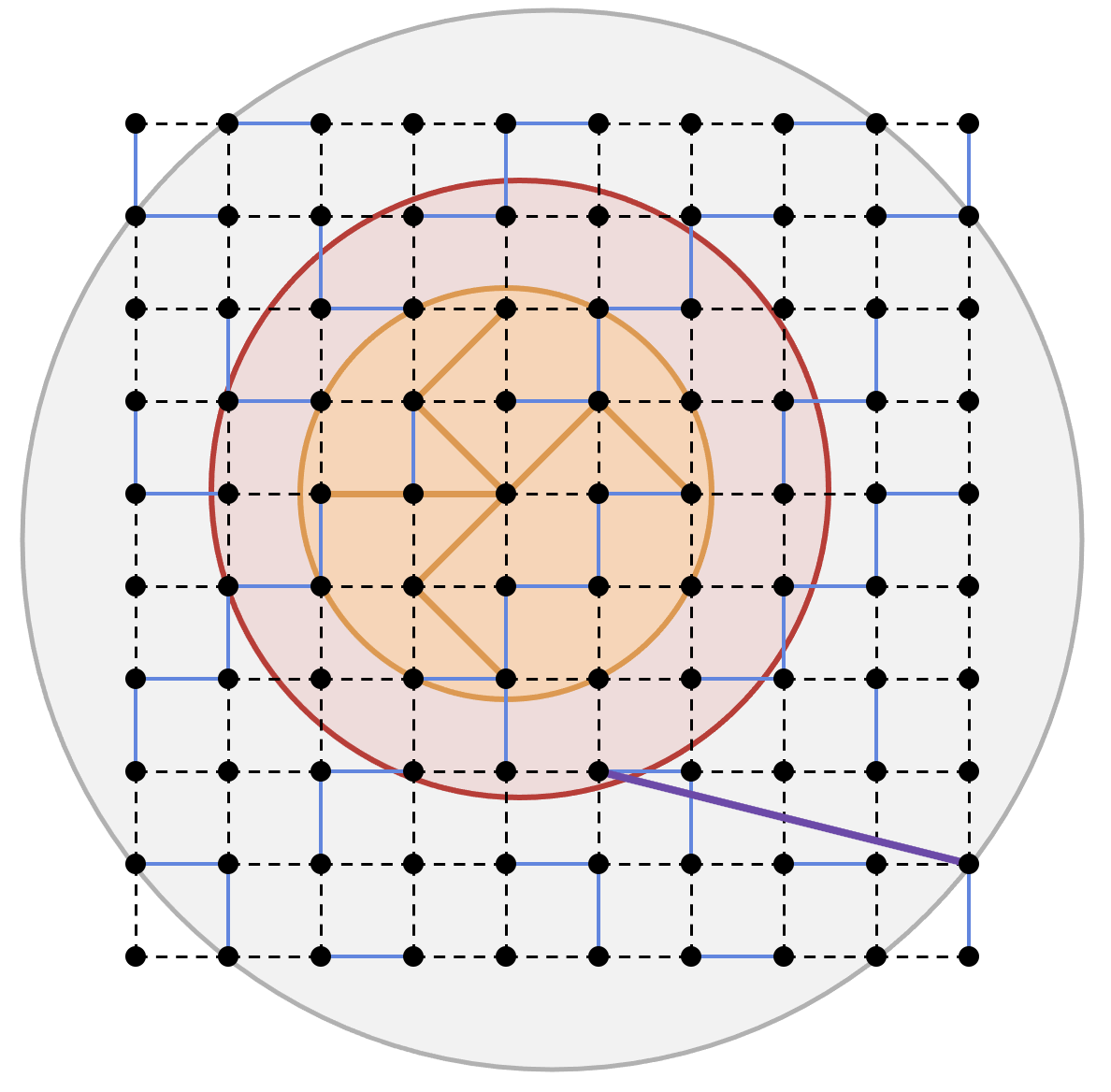}
  \caption{Each $\mathsf{Test}$ is approximated by a quasi-local  baseline and a strictly local correction.}
  \label{fig:optimal-sample-1}
\end{subfigure}
\begin{subfigure}[t]{0.42\textwidth}
  \centering
  \includegraphics[width=1.0\linewidth]{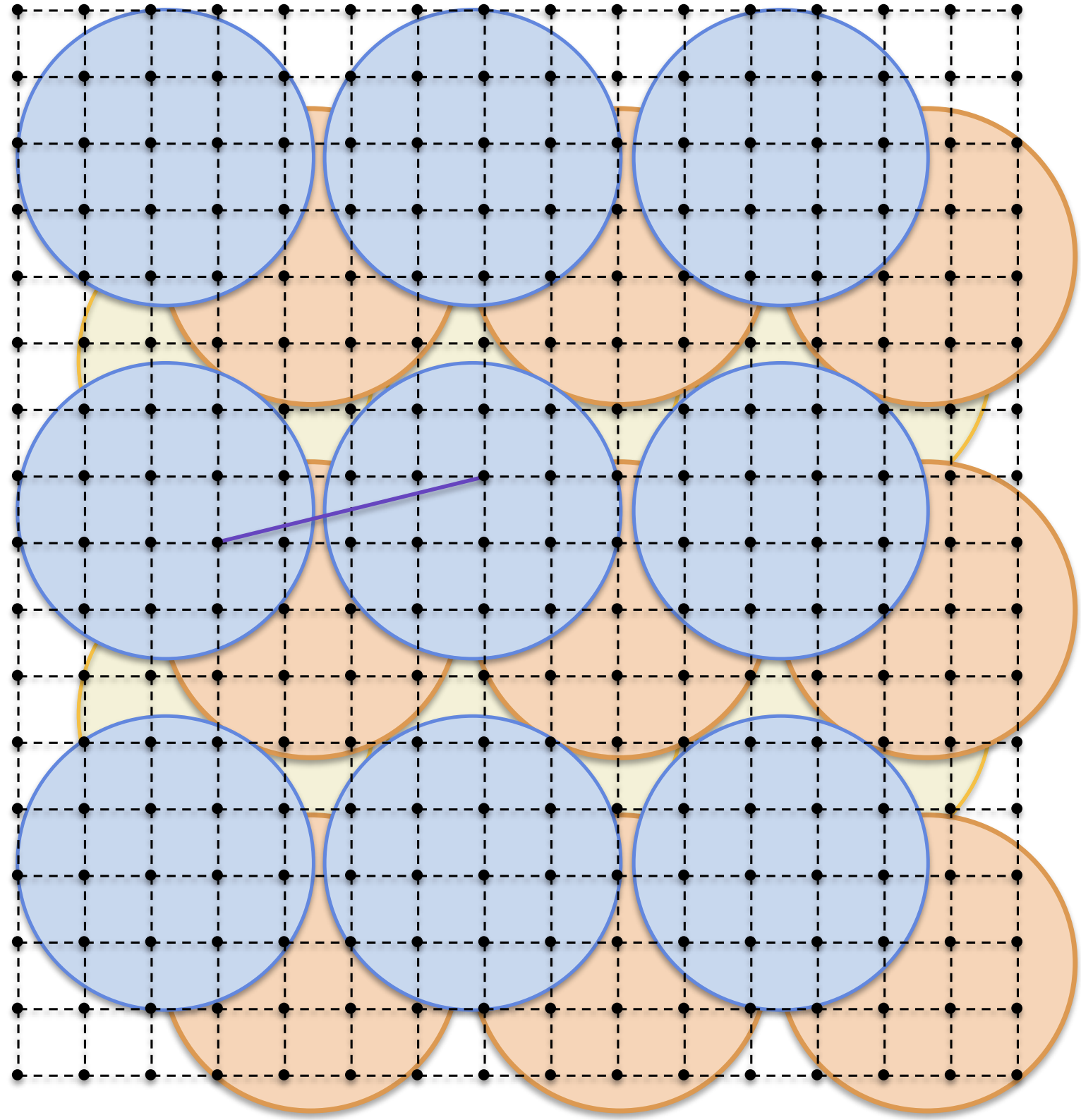}
  \caption{A coloring schedule measures the baseline estimates in parallel.}
  \label{fig:optimal-sample-2}
\end{subfigure}

\caption{The sample-optimal neighborhood guessing algorithm. \\(a) The $\mathsf{Test}$ associated to a locally perturbed Hamiltonian $\vG+\eta\vU$ is computed by first estimating the test under the current Hamiltonian $\vG$, truncated to radius $\sR\sim \log\eta^{-1}$ (gray). This `baseline' is refined by local measurements at radius $\sL$ (red) around the perturbation \eqref{eq:intro-baseline+corr}. \\ (b) The sites are partitioned into disjoint `colors' (only 3 pictured), each with non-overlapping radius $\sR$ balls, and performed in parallel. A randomization scheme is used to address observables (purple) crossing balls of the same color (see \cref{alg:struct-meas-baseline}).} 
\label{fig:optimal-sample}
\end{figure}

\paragraph{Approximate detailed-balance and generalizations to metastable states.} Conceptually, there are two central challenges to generalizing our approach to parameter learning to samples from a metastable state $\vsigma$ (w.r.t. $\vH$). First, $\vsigma$ is not necessarily a Gibbs state, and thereby the log-likelihood $\log \vsigma$ need not be local -- in particular, we cannot efficiently measure the associated Fisher information $\FI_{\vG,a}\big(\vsigma\big|\big| \vrho_{\vG}\big)$. What is more, even if we could compute it, metastable states only satisfy a certain rather rigid form of \textit{approximate detailed-balance} (see \cref{def:metastable-adb}, from \cite{BCV25}); and thus it is not clear that the convexity arguments for the Fisher information apply. 

Accordingly, we are forced to slightly re-route the proof of correctness, and instead directly prove convexity of the stationarity test $\mathsf{Test}_i[\vH-\vG, \vG;\vsigma]$ when measured on metastable states. Crucially, we leverage said approximate detailed balance condition to prove an analog of \eqref{eq:approx-convexity-overview}, in that the time derivative of $\vH-\vG$ when $\vsigma$ is evolved under $\CL_{\vG, a}$ is approximately a gradient-square:
\begin{align}
   \mathsf{Test}_i[\vH-\vG, \vG;\vsigma] = \tr[\CL_{\vG, a}[\vsigma]\cdot (\vH-\vG)]  \approx \beta\cdot \iiint \|
    [\hat{\vA}_{\vG}^{a}(\omega,t), \vH-\vG]\|_{\vsigma,s}^{2} \rd\mu_{\omega, t, s},
\end{align}
now weighted under $\vsigma$. We note that this step (\cref{lem:metastable-differentiated-entropy-gradient}) is quite technically involved, but ultimately relies on combining tools from the study of Markov properties in metastable states \cite[Section D]{BCV25} with our new Lipschitz continuity properties  \eqref{eq:lipschitz-overview} and \eqref{eq:intro-baseline+corr} for the \cite{chen2023efficient} Lindbladian. We refer the reader to the correctness argument in \cref{section:metastab-correct} for further details.

\section*{Organization}

In \cref{section:prelim}, we present the relevant preliminaries. In \cref{section:fisher-identifiability}, we present the connection between the stationarity test and the quantum Fisher information. As a warm-up, in \cref{section:gibbs-alg} we give a parameter learning algorithm from Gibbs states. In \cref{section:structure-learning} we give our structure learning algorithm from Gibbs states (\cref{theorem:structure-learning-intro}), and \cref{section:metastable-learning} our parameter learning algorithm from metastable states (\cref{thm:intro-metastable-learning}). In \cref{section:locality} we present the relevant locality estimates for \cite{chen2023efficient}, and in \cref{section:area-law} the refinements to the area law (\cref{thm:intro-metastability-area-law}).

\section*{Acknowledgements}

I thank Chi-Fang Chen, Umesh Vazirani, Yunchao Liu, Reza Gheissari, Tim M\"{o}bus, Daniel Stilck Fran\c{c}a and Cambyse Rouz\'e for collaborations on related work which inspired this inquiry \cite{BC26, BCV25, BGL25, möbus2026robuststructurelearningklocal}. Anurag Anshu, Sitan Chen, Sabee Grewal, Laura Lewis, Quynh Nguyen, Rolando Somma, for conversations on the literature \cite{anshu2020sample,
lewis2026learningstructureopenquantum} and on different approaches to this problem. Ryan Babbush, Tomas Jochym-O'Connor, Robin Kothari, Thomas Schuster for valuable related discussions. Finally, Sidhanth Mohanty, Amit Rajaraman, and David Wu for answering numerous questions on their related work \cite{ liu2024locally}.

I thank ChatGPT 5.5 - 6 and Gemini Argon for assistance in proof writing, verification, and rendering diagrams. In particular, ChatGPT 5.6 shaved a $\log n$ factor in the classical shadows argument of \cref{theorem:pauli-shadows-after-channels} and simplified the proofs of the Lipschitz continuity arguments of \cref{section:locality,lem:metastable-convexity}.\\

\noindent \textbf{Concurrent work.} Concurrently and independently, Bingrun Wang, Qi Ye, and Chi-Fang Chen derived similar results on quantum metastable states to our \cref{section:metastab-intro}. I thank the authors for discussing their work and coordinating an arxiv posting. Shukla, Jayakumar and Lokhov \cite{shukla2026operatorscorematchinglearning} and Dong and Leng \cite{dong2026quantumscorematchingapplications} introduced quantum analogs of ``score matching'' and used them to develop algorithms for Hamiltonian parameter learning. They report numerical success on specific instances at low-temperatures, and \cite{dong2026quantumscorematchingapplications} provides rigorous sample complexity guarantees in the high temperature regime.

\section{Preliminaries}
\label{section:prelim}

\noindent \textbf{Notation.} We consider quantum spin systems of finite size $n$, each of constant local dimension $2^\sq$. Each such spin $i\in [n]:=\{1, 2, \cdots, n\}$ will be referred to as a \textit{site}. The associated Hilbert space is denoted as $\calH = \otimes_{i=1}^{n} (\BC^{2^\sq})$. The space of linear operators on $\calH$ is referred to as $\CB(\calH)$; we use boldface letters to refer to operators. A superoperator $\CN:\CB(\calH)\rightarrow \CB(\calH)$ is a map between linear operators; we use in curly font~$\CL$ with matrix arguments in square brackets to refer to $\CL[\vrho]$. We refer to $[\cdot, \cdot]$ and $\{\cdot, \cdot\}$ as the commutator and anti-commutator respectively.

We denote the set of single-site $\sq$-weight Pauli strings, supported on the $i$th $2^\sq$ dimensional qudit as:
\begin{align}
    \CS_i &:= ( \{\vI_2,\vX,\vY,\vZ\}^{\otimes \sq}\} )\otimes \vI_{2^\sq}^{n-1},
\end{align}
and we will never care how the Pauli strings are actually embedded within a site. 

Sans serif is used to refer to subsets of sites $\sA\subseteq [n]$. For any subset $\sA\subseteq [n]$, we define the set of single-site Pauli strings, and multi-site Pauli strings:
\begin{align}
    \CS^1_\sA &:= \cup_{i\in \sA} \CS_i,\qquad \CS_\sA := ( \{\vI_2,\vX,\vY,\vZ\}^{\otimes \sq\labs{\sA}})\otimes \vI_{2^\sq}^{n-\labs{\sA}}.
\end{align}

We use $\CO(\cdot),\Omega (\cdot)$ to denote asymptotic upper and lower bounds; and in an abuse of notation $ a \lesssim b$ to mean $a \le c b$ for a universal constant $c$.  

\subsection{Hamiltonians and Lieb-Robinson bounds}

We fix the locality $\sk$ and local dimension $2^{\sq}$. On a set of $n$ qubits, a Hamiltonian is specified on a family $\Gamma\subseteq \CS_{[n]}^{\sk}$ of distinct, nonidentity Pauli strings, each supported on at most $\sk$ sites:
\begin{equation}
 \vH=\sum_{\gamma\in\Gamma}h_\gamma\vP_\gamma \in \CB(\calH),
 \qquad h_\gamma\in[-1,1].
\end{equation}
In a slight abuse of notation we write $\gamma\subseteq\sA$ when $\supp(\vP_\gamma)\subseteq\sA$, and $\gamma\ni i$ when $i\in\supp(\vP_\gamma)$.

\begin{definition}
    Given such a Hamiltonian $\vH=\sum_{\gamma\in\Gamma}h_\gamma\vP_\gamma$, we define the \textit{interaction graph} $G_{\vH}=(\Gamma, E)$ with vertices corresponding to the set $\Gamma$, and edges between $\gamma_1$ and $\gamma_2$ iff they have overlapping supports.
\end{definition}

When $\vH$ is otherwise implicit, $\dist(\gamma_1, \gamma_2)$ refers to the distance between clauses on said graph, and for $i\in [n]$,
$\dist_\Gamma(i,\gamma):=\min_{\gamma_0\ni i}\dist_\Gamma(\gamma_0,\gamma)$. For $r\geq0$ define the \emph{clause balls} and \emph{site balls}
\begin{align}
x\in[n]\cup\Gamma:\quad \CB_{\vH}[x,r]&:=\{\gamma\in\Gamma:\dist_\Gamma(x,\gamma)\leq r\},\quad 
 \sB_{\vH}[x,r]:=\supp(\CB_{\vH}[x,r])\subseteq [n].
 \label{eq:site-balls}
\end{align}
We similarly extend the definitions above to subsets $\sR\subseteq [n]\cup \Gamma$. When the Hamiltonian clauses $\Gamma$ are implicit, the subscript is omitted. We are now in a position to define the families of lattice Hamiltonians we consider. 

\begin{definition}[$(\sd,\sD,\scgr_\sgr)$ Hamiltonians]\label{def:polynomial-growth-hamiltonian}
For $\sd\geq1$, $\sD\geq1$, and $\scgr_\sgr \geq1$, a Hamiltonian $\vH$ is a \emph{$(\sd,\sD,\scgr_\sgr)$ Hamiltonian} if its interaction graph $G_{\vH}$ has maximum degree at most $\sd$ (including self-edges) and admits polynomial volume growth:
\begin{equation}
 |\CB_{\vH}[x,r]|\leq \scgr_\sgr\cdot  r^\sD
 \quad\text{for every }x\in[n]\cup\Gamma\text{ and every }r\geq1.
 \label{eq:polynomial-growth}
\end{equation}
\end{definition}

Throughout this paper, we will treat the Hamiltonian parameters $\sd, \sD, \mathsf{c_{gr}}, \sk, \sq$ as fixed constants and absorb them into ``big $O$'' notation. 

For any subset $\sR\subseteq [n]\cup \Gamma$ and distance $r\geq 1$ we denote the restriction of the Hamiltonian to the clauses around $\sR$ as $\vH_{\CB[\sR, r]}:=\sum_{\gamma\in\CB[\sR,r]}h_\gamma\vP_\gamma$.

\subsubsection{Lieb-Robinson bounds}

Lieb-Robinson bounds \cite{Lieb1972, hastings2006spectral, nachtergaele2006lieb, nachtergaele2018lieb} quantify the speed in which information propagates throughout a quantum spin system under time evolution. The intention is to understand how the locality of a local observable $\vX$ grows as it is time-evolved under a local Hamiltonian $\vH$. We require the following formulation, which expresses $\vX(t)$ as a sum over ``shells'':

\begin{lemma}[Lieb-Robinson bound]
\label{lem:lieb-robinson-commutator}
Let $\vH$ be a Hamiltonian of degree $\sd$. Then, for any observable $\vX$ supported on $\sA\subset[n]$, and time $t\in \BR$, the time-evolution $\vX(t)$ admits a decomposition into a telescoping sum:
\begin{align}
    \vX(t)
    &=
    \sum_{\ell\geq 0}
        \vX^{(\ell)}(t)-\vX^{(\ell-1)}(t), \quad \text{where}\quad \vX^{(\ell)}(t):=
    e^{it \vH_{\CB[\sA, \ell]} }\vX e^{-i t\vH_{\CB[\sA, \ell]}}
\end{align}
is the time-evolution under the truncated Hamiltonian $\vH_{\CB[\sA, \ell]}$. Furthermore, there exist constants $\mathsf{c_{LR}},\mu,v>0$ as a function of $\sq, \sd, \sk$, such that
\begin{equation}
  \forall t\in \BR, \quad \ell\geq 1,\quad  \theta\in (0, 1):\quad   \big\|
        \vX^{(\ell)}(t)-\vX^{(\ell-1)}(t)
    \big\| \leq \mathsf{c_{LR}}\,|\sA|\,\|\vX\| \cdot e^{\theta(v|t|-\mu \ell)}.
    \label{eq:LR-annulus-decomposition}
\end{equation}
\end{lemma}

\subsection{Weighted inner products and detailed-balance}

We dedicate this section to the relevant non-commutative weighted inner products. We assume $\vrho$ is full rank throughout this section. 
\begin{definition}
    [$(\vrho, s)$ Weighted Inner Product]\label{defn:s_inner} Given a full-rank state $\vrho$ and $s\in [-\frac{1}{2}, \frac{1}{2}]$, we define for any pair of operators $\vX, \vY$
    \begin{equation}
        \langle \vX,\vY\rangle_{\vrho, s}:=\tr[\vX^\dagger \vrho^{\frac{1}{2}+s}\vY\vrho^{\frac{1}{2}-s}]\,.
    \end{equation}
    We denote by $\|\vX\|_{\vrho, s} = \sqrt{\langle \vX,\vX\rangle_{\vrho, s}}$ the $(\vrho, s)$-weighted $2$-norm.
\end{definition}

\begin{remark}
    We recognize this notation is distinct from the convention in the literature, where $s\in [0, 1]$. We adopt this choice to simplify formulas that interpolate over different values of $s$.
\end{remark}

Of particular attention is the KMS inner product, the case where $s=0$. This choice plays such a central role that we refer to this case as the ``$\vrho$'' inner product:
\begin{equation}
     \langle \vX,\vY\rangle_{\vrho}:= \langle \vX,\vY\rangle_{\vrho, 0}
\end{equation}
Associated to a weighted inner product, we can define a notion of detailed-balance:
\begin{definition}
    [Detailed-balance] A superoperator $\CN$ is said to be $\vrho\mathsf{-detail-balanced}$ if it is self-adjoint w.r.t. the $\vrho$ inner product:
    \begin{equation}
        \forall \vX, \vY:\quad \langle \vX, \CN[\vY]\rangle_{\vrho}= \langle \CN[\vX], \vY\rangle_{\vrho} \label{eq:detailed-balance}
    \end{equation}
\end{definition}

We rely on the following unconditional bounds on these weighted inner products:
\begin{lemma}\label{lemma:kms-to-s}
    For any operator $\vX$ state $\vrho$ and $s\in [-\frac{1}{2}, \frac{1}{2}]$, $\|\vX\|_{\vrho, s}\leq \|\vX\|$ and $\|\vX\|_{\vrho, s}^2+\|\vX\|_{\vrho, -s}^2\geq 2 \|\vX\|_{\vrho}^2$.
\end{lemma}

We make extensive use of the following version of H\"older's inequality for the $\vrho$ norm.

\begin{lemma}[H\"older in KMS Norm, e.g., {~\cite[Lemma IX.4]{chen2025GibbsMarkov}}]\label{lem:kms-holder}
For any pair of operators $\vX, \vY$, and full rank state $\vrho,$
\begin{align}
\|\vX\vY\|_{\vrho} &\le \|\vrho^{1/4}\vX\vrho^{-1/4}\| \cdot \|\vY\|_{\vrho}
\end{align}     
\end{lemma}

Of course, the efficacy of the above is contingent on the convergence of complex-time evolution of $\vX$. We also consider the following induced norm on superoperators:

\begin{definition}
    [KMS-Induced Superoperator Norm]\label{defn:induced_norm} For any superoperator $\CN:\CB(\calH)\rightarrow \CB(\calH)$ and full-rank state $\vrho$,
    \begin{align}
    \norm{\CN}_{\vrho}:= \sup_{\vO}\frac{\norm{\CN[\vO]}_{\vrho}}{\norm{\vO}_{\vrho}}.
\end{align}
\end{definition}

\subsection{The \texorpdfstring{\cite{chen2023efficient}}{[CKG23]} Gibbs sampler}

We dedicate this section to an introduction on the construction of \cite{chen2023efficient}. 

\subsubsection{Operator Fourier Transform}

\begin{definition}\label{def:oft}
    The operator FT of an operator $\vA$ at energy width $\sigma>0$, associated to a Hamiltonian $\vG$, is written as: 
    \begin{align}
\hat{\vA}_{\vG, \sigma}(\omega)=  \frac{1}{\sqrt{2\pi}}\int_{-\infty}^{\infty} e^{i \vG t} \vA e^{-i \vG t} e^{-i \omega t} f_{\sigma}(t)\rd t.\label{eq:OFT}
    \end{align}
\end{definition}
\noindent In the above, the function $f_\sigma(t)$ above is a Gaussian filter:
    \begin{align}
       f_{\sigma}(t) = e^{-\sigma^2t^2}\sqrt{\sigma\sqrt{2/\pi}}\quad \text{and}\quad  \hat{f}_{\sigma}(\omega)=\frac{e^{- \omega^2/4\sigma^2}}{\sqrt{\sigma\sqrt{2\pi}}}  = \frac{1}{\sqrt{2\pi}}\int_{-\infty}^{\infty}e^{-i\omega t} f_\sigma(t)\mathrm{d}t.\label{eq:fwft}
    \end{align}

\noindent Whenever implicit, we omit the subscripts ${\hat{\vA}_{\vG, \sigma}}(\omega)\equiv {\hat{\vA}_{\vG}}(\omega)$, $f(t)=f_\sigma(t)$.

\begin{lemma}[Decomposing an Operator by the Energy Change]\label{lem:sumoverenergies}
For any (not necessarily Hermitian) operator $\vA$, we have that
\begin{align}
\vA =  \frac{1}{\sqrt{2\sigma\sqrt{2\pi}}} \int_{-\infty}^{\infty} \hat{\vA}_{\vG}(\omega)\rd \omega.
\end{align}
\end{lemma}

\begin{lemma}[Operator Parseval's identity {\cite[Proposition A.2]{chen2023efficient}}]\label{lem:OParseval}
    For any choice of Hamiltonian $\vG$, and operator $\vA$, and real-valued function $f,$
	\begin{align}
		\bigg\|\int_{-\infty}^{\infty} f(\omega)\hat{\vA}_{\vG, \sigma}(\omega)^{\dagger}\hat{\vA}_{\vG, \sigma}(\omega) \mathrm{d}\omega\bigg\|
		&\le \big\|f(\omega)\big\|_{\infty}\bigg\| \vA^{\dagger}\vA\bigg\|.\label{eq:contFourierIdentity}
	\end{align}
\end{lemma}

We further require the following bounds on the operator norm of operator Fourier transforms. In some sense, they quantify that the operator FT decays in norm in the large frequency $\omega$ regime. 

\begin{lemma}[A priori Norm bounds on Operator Fourier Transforms, Corollary IX.2 \cite{chen2025GibbsMarkov}]\label{lem:bounds_oft}
For any $\beta ,\omega\in \BR$ and operator $\vA$ with norm $\norm{\vA}\le 1$, the operator Fourier transform $\hat{\vA}_{\vG, \sigma}(\omega)$ with uncertainty $\sigma$ satisfies 
\begin{align}
        \|\hat{\vA}_{\vG, \sigma}(\omega)\| \leq \frac{e^{-\beta \omega + \sigma^2 \beta^2}}{\sqrt{{\sigma}\sqrt{2\pi}}} \|e^{\beta \vG}\vA e^{-\beta \vG}\|.
\end{align}
\end{lemma}

\begin{lemma}
    [Gaussian Convolution of Operator Fourier Transforms]
    \label{lem:OFT_Gaussian_expansion} For any operator $\vA$ and uncertainties $\sigma_1, \sigma_2, \sigma_3\in \BR^+$ satisfying $1/\sigma_1^2 = 1/\sigma_2^2+1/\sigma_3^2$, we have:
    \begin{equation}
        \hat{\vA}_{\vG, \sigma_1}(\omega) =  \sqrt{\frac{\sigma_2\sigma_3}{\sigma_1\sqrt{2\pi}}} \int_{-\infty}^{\infty} \hat{\vA}_{\vG, \sigma_3}(\omega, t) e^{-i \omega t}f_{\sigma_2}(t) \rd t.
    \end{equation}
    
\end{lemma}

\begin{lemma}
    [Imaginary Time Conjugation of Paulis {\cite[Corollary IX.1]{chen2025GibbsMarkov}}]\label{lem:norm_bound_pauli} Let $\vH$ denote a bounded degree $\sd$ Hamiltonian. 
    Then, for any $\beta\in \mathbb{R}$ s.t. $|\beta|<(2\sd)^{-1}$ and $w$-qubit Pauli $\vA$, 
    \begin{equation}
        \|e^{-\beta\vH} \vA e^{\beta\vH}\|\leq \bigg(\frac{1}{1-2\sd|\beta|}\bigg)^w
    \end{equation}
\end{lemma}

Finally, we treat the norms of the imaginary time evolution of OFTs.

\begin{lemma}[Norm Bounds on Imaginary Time Conjugation~{\cite[Lemma IX.2]{chen2025GibbsMarkov}}]\label{lem:bounds_imaginary_conjugation}
For any $\beta ,\omega\in \BR$, uncertainty width $\sigma\in \BR^+$, and operator $\vA$ with norm $\norm{\vA}\le 1$, the operator Fourier transform \eqref{eq:OFT} satisfies 
\begin{align}
        e^{\beta \vG} \hat{\vA}_{\vG, \sigma}(\omega) e^{-\beta \vG}
& = e^{\beta\omega} \cdot \hat{\vA}_{\vG, \sigma}(\omega+2\sigma^2\beta) e^{\sigma^2\beta^2}.
\end{align}
Thus, 
\begin{align}
    \norm{e^{\beta \vG} \hat{\vA}_{\vG, \sigma}(\omega) e^{-\beta \vG}} \le \frac{e^{\sigma^2\beta^2}}{\sqrt{{\sigma}\sqrt{2\pi}}} e^{\beta\omega} \norm{\vA}.
\end{align}
\end{lemma}

\subsubsection{The Lindbladian}

We are now in a position to define the family of Lindbladians of \cite{chen2023efficient}.

\begin{definition}
    [The \cite{chen2023efficient} Lindbladian]\label{def:ckg-lind} Fixed $\vG$, $\beta>0$, $0<\sigma\leq \beta^{-1}$, and a single self-adjoint jump $\vA^a = \vA^{a\dagger}$. The Lindbladian~\cite{chen2023efficient} is written as
\begin{align}
		\CL_{\vG, a}[\cdot] = \underset{\text{``coherent''}}{\underbrace{ -i [\vC^a, \cdot]}} + 
		\int_{-\infty}^{\infty} \gamma(\omega) \bigg(\underset{\text{``transition''}}{\underbrace{\hat{\vA}_{\vG, \sigma}^a(\omega)(\cdot)\hat{\vA}_{\vG, \sigma}^{a}(\omega)^\dagger}} - \underset{\text{``decay''}}{\underbrace{\frac{1}{2}\{\hat{\vA}_{\vG, \sigma}^{a}(\omega)^\dagger\hat{\vA}_{\vG, \sigma}^a(\omega),\cdot\}}}\bigg)\rd\omega,\label{eq:exact_DB_L}
\end{align}
\noindent the ``coherent part'' $\vC^{a}$ is a Hermitian operator:
\begin{align}
    &\vC^a = \iint_{-\infty}^{\infty} \gamma(\omega) \cdot c(t) \cdot \hat{\vA}_{\vG, \sigma}^a(\omega,t)^{\dagger}\hat{\vA}_{\vG, \sigma}^a(\omega,t)  \rd t\rd \omega, \quad  \\\text{with}\quad  &c(t) := \frac{1}{\beta\sinh(2\pi t/\beta)} \quad \text{and}\quad  \hat{\vA}_{\vG, \sigma}^a(\omega,t):=e^{i\vG t}\hat{\vA}_\sigma^a(\omega)e^{-i\vG t}.
\end{align}
\end{definition}
\noindent The choice of $\gamma$ ensures detailed-balance; here we  use the shifted-Metropolis weight function:
\begin{align}
    \gamma_\mathsf{M}(\omega) := \exp\bigg(-\beta\max\left(\omega +\frac{\beta \sigma^2}{2},0\right)\bigg).\label{eq:Metropolis}
\end{align} 

Henceforth we restrict our attention to $\sigma:=\beta^{-1}$, and often consider the generator defined by all single-site Pauli jumps on $i\in [n]:\CL_{\vG, i}:=\sum_{a\in \CS_{i}}\CL_{\vG, a}.$

\begin{lemma}
    [{\cite[Corollary A.1]{chen2025GibbsMarkov}}]
    When $\norm{\vA^a}=1$, we have system-size independent norm bounds $\norm{\CL_{\vG,a}^\dagger}_{\infty-\infty}\le O(1+\beta\sd)$ under $\gamma_\mathsf{M}$.
\end{lemma}

The defining feature of this family of Lindbladians \cite{chen2023efficient} is that it simultaneously satisfies KMS detailed-balance w.r.t to the Gibbs state $\vrho_{\vG}$, as defined in \eqref{eq:detailed-balance}, while still satisfying some form of (quasi-)locality; in a sense inheriting the geometry of the underlying Hamiltonian. This can be made precise by the following lemma:

\begin{lemma}
    [Quasi-locality of the Lindbladian dynamics, {\cite[Lemma VII.3]{chen2025GibbsMarkov}}]\label{lemma:lindblad-LR} Let $\vA^a$ be a single-site operator on site $i\in [n]$ s.t. $\|\vA^a\|\leq 1$. In the context of \cref{def:ckg-lind}, there exists constants $c_1, c_2>0$ such that any $\ell>1$, the super-operator $\CL_{\vG_{\CB[i,\ell]}, a}$ satisfies:
    \begin{equation}
        \|\CL_{\vG,a}^\dagger - \CL_{\vG_{\CB[i,\ell]}, a}^\dagger\|_{\infty-\infty} \leq c_1\cdot e^{-c_2\ell}.
    \end{equation}
   
\end{lemma}

The lemma above is a simple consequence of standard Lieb-Robinson bounds (\cref{lem:lieb-robinson-commutator}); we dedicate \cref{section:locality} to further discussion on quasi-locality properties.

Instead of thinking in the Schrodinger picture (i.e., the evolution of the state under the dynamics), it will be instructive to consider the evolution of operators. \cite[Lemma C.2]{rouze2024efficient} derived a clean and explicit expression for the Dirichlet form of the Lindbladians \eqref{eq:exact_DB_L} we study:
\begin{lemma}[The Dirichlet form {\cite[Lemma C.2]{rouze2024efficient},\cite[Lemma X.3]{chen2025GibbsMarkov}}]
\label{lem:Dicirchlet}
    The Dirichlet form of the Lindbladian  $\CL_{\vG, a}$~\eqref{eq:exact_DB_L}, is a $\vrho_{\vG}$-weighted inner product of commutators:
\begin{align}
   -\langle \vX,\CL_{\vG, a}^{\dagger}[\vY]\rangle_{\vrho_{\vG}} = \iint_{-\infty}^{\infty} \langle [\hat{\vA}_{\vG}^a(\omega,t),\vX ], [\hat{\vA}_{\vG}^a(\omega,t),\vY ]\rangle_{\vrho_{\vG}}\cdot g(t) h(\omega) \cdot \rd t \rd \omega,
\end{align}
where the specific functions $g, h$ are time and frequency filter functions:
\begin{align}
 h(\omega)&= e^{-\sigma^2\beta^2/8} e^{-\labs{\omega}\beta/2}\ge 0, \quad g(t) = \frac{1}{\beta\cosh(2\pi t/ \beta)}\ge 0.
\end{align}
\end{lemma}

Our analysis will hinge crucially on the structure of the quantum Fisher information. 

\begin{definition}[The Fisher information, {\cite[Theorem B.1]{BCV25}}]
\label{def:relative-fisher-information} Given a jump operator $\vA^a$, we define the Fisher information of a full rank state $\vsigma$ relative to $\vrho_\vG$ as 
\begin{equation}
    \FI_{\vG,a}(\sig|| \vrho_{\vG})
  :=
  -\Tr\left[
    \CL_{\vG,a}[\sig]( \log\sig-\log\vrho_{\vG})
  \right] =  \QF_{a}(\log\sig-\log\vrho_{\vG}, \vG;\vsigma),
\end{equation}
where we define the functional:
\begin{equation}
 \QF_{a}(\vX, \vG;\vsigma)
  :=
  \int_{-1/2}^{1/2}
  \int_{\BR}\int_{\BR}
  \left\|
    [\hat{\vA}_{\vG}^{a}(\omega,t),\vX]
  \right\|_{\vsigma,s}^{2}
  h_s(\omega)g_s(t)
  \,\rd t\,\rd\omega\,\rd s \geq 0.
  \label{eq:fisher-quadratic}
\end{equation}
In terms of the positive kernels 
\begin{align}
  h_s(\omega)
  &:=
  \exp\!\left(
    \frac{s\beta(2\omega-s\beta\sigma^{2})}{2}
  \right)
  h(\omega-s\beta\sigma^{2}), \quad 
  g_s(t)
  :=
  \frac{2}{\beta}\,
  \frac{
    \cos(\pi s)\cosh(2\pi t/\beta)
  }{
    \cosh(4\pi t/\beta)+\cos(2\pi s)
  },
  \label{eq:gs-explicit}
\end{align}
\end{definition}

\subsection{Pauli shadow tomography}

In our structure learning algorithm, we will require the following formulation of Pauli Shadow tomography \cite{huang2020predicting, Huangderandomized2021, aaronson2018shadowtomographyquantumstates, evans2019scalablebayesianhamiltonianlearning, Huang_2021-information, Cotler_2020}, which enables sample and time efficient learning of collections of local Pauli expectation values. 

\begin{theorem}
    [Pauli shadow tomography, {\cite[Theorem 1, Proposition 3]{huang2020predicting}}]\label{theorem:pauli-shadows}
    Fix $\eta, \delta>0$. Let $\vsigma$ be an $n$ qubit quantum state and $\vP_1, \cdots, \vP_{\mathsf{m}}\in \CS_{[n]}^{\mathsf{w}}$ be a collection of $\mathsf{m}$ distinct $\mathsf{w}-$local Pauli observables. Then, there exists a learning algorithm using 
    \begin{equation}
        \mathsf{N_{samp}}:=O\bigg( 2^{2\mathsf{w}}\cdot \eta^{-2}\cdot \log \frac{\mathsf{m}}{\delta}\bigg)
    \end{equation}
    copies of $\vsigma$ which, with probability $\geq 1-\delta$ returns estimates for $\tr[\vsigma \vP_i]$ up to additive error $\eta$ for each $i\in [\mathsf{m}]$. The algorithm runs in $\poly(n, \mathsf{m},\mathsf{N_{samp}} )$ classical and quantum time.
\end{theorem}

We require a simple corollary of this statement, concerning the computation of the same Pauli expectation values when the unknown state $\vsigma$ is subjected to a set of disjoint channels.

\begin{lemma}
    [Pauli shadows after disjoint channels]
    \label{theorem:pauli-shadows-after-channels}
    In the same context as \cref{theorem:pauli-shadows}, suppose we partition the $n$ qubits into disjoint subsets $\sA_1, \sA_2\cdots , \sA_\mathsf{s} \subseteq [n]$, and define a set of efficient CPTP maps $\Phi_1[\cdot ],\cdots, \Phi_\mathsf{s}[\cdot ]$ where $\Phi_j$ acts only on $\sA_j$. Then, there exists a learning algorithm using 
    \begin{equation}
        \mathsf{N_{samp}}:=2^{O(\mathsf{w})}\cdot O\bigg(\eta^{-2}\cdot \log \frac{\mathsf{m\cdot \mathsf{s}}}{\delta}\bigg)
    \end{equation}
    copies of $\vsigma$ which, with probability $\geq 1-\delta$ returns estimates for $\tr[\Phi_j[\vsigma] \vP_i]$ up to additive error $\eta$ for each $i\in [\mathsf{m}], j\in [\mathsf{s}]$. The algorithm runs in $\poly(n, \mathsf{m},\mathsf{N_{samp}} )$ classical and quantum time.
\end{lemma}

The only issue here lies in that any given Pauli operator $\vP_i$ may not be entirely contained within a single subset $\sA_j$. Nevertheless a simple coloring argument reduces this setting to \cref{theorem:pauli-shadows}.

\begin{proof}
    Suppose we sample a bitstring $b_1, \cdots, b_{\mathsf{s}}\in \{0, 1\}$ at random, and apply each $\Phi_j$ to the state $\vsigma$ if $b_j=1$. The mixed state $\valpha$ captures the outcome of this experiment:
    \begin{equation}
        \valpha = 2^{-\mathsf{s}} \sum_{b\in \{0, 1\}^{\mathsf{s}}} \ketbra{b}\otimes \bigg(\bigotimes_{j}^{\mathsf{s}} \Phi^{b_j}\bigg)[\vsigma] 
    \end{equation}
     Moreover, since each $\vP_i$ is at most $\mathsf{w}$ local, for any $j\in [\mathsf{s}]$ the desired estimate $\tr[\Phi_j[\vsigma] \vP_i]$ can be written as the expectation value on $2\mathsf{w}$ qubits of $\valpha$, simply by conditioning on $b_j=1$ and $b_k=0$ for each subset $\mathsf{A}_k\neq\mathsf{A}_j$ which intersects $\mathsf{supp}(\vP_i)$. The resulting expectation value is a linear combination of $2^{O(\mathsf{w})}$ Pauli operators, each of which is $2\mathsf{w}$ local. \cref{theorem:pauli-shadows} then gives the desired bound.
\end{proof}

\section{The stationarity test and the quantum Fisher information}
\label{section:fisher-identifiability}

We dedicate this section to a description of the \textit{stationarity test}, a (quasi)-local and efficiently-checkable criterion which verifies if one has correctly identified or learned the underlying, ``ground truth'' Hamiltonian $\vH$ from samples of its Gibbs state $\vrho_\vH$. Simply put, given a candidate ``guess'' Hamiltonian $\vG$, the test asks whether a prescribed local observable is approximately stationary under the dynamics generated by
\(\vG\). We present a formal account of the test in \cref{subsection:stationarity-test}.

The central virtue of the stationarity test, and the main technical content of this section, is to prove its soundness: namely, if for all local observables, $\vrho_\vH$ is approximately stationary under $\vG$, then we must have $\vH\approx \vG$ (in some local sense). We make this argument precise in \cref{subsection:fixed-lattice-local-faithfulness}, where we rely on the fact that the Fisher information (\cref{def:relative-fisher-information}) between these Gibbs states $\vrho_{\vG}$ and $\vrho_\vH$ happens to be captured by a linear combination of stationarity tests, and furthermore satisfies a certain convexity property.

\subsection{The stationarity test}
\label{subsection:stationarity-test}

We begin with a restatement of the definition. Fix the inverse temperature $\beta>0$ and frequency filter $\sigma:=\beta^{-1}$ in the Lindbladian \cref{def:ckg-lind}.

\begin{definition}[The stationarity test]
\label{def:stationarity-test}
Let $\CA$ be a set of Pauli jumps, and
\(\CL_{\vG,\CA}:=\sum_{a\in \CA}\CL_{\vG,a}\).  For an observable \(\vO\), a
guess Hamiltonian $\vG$, and a state $\vsigma$, define
\begin{equation}
  \mathsf{Test}_{\CA}[\vO,\vG;\sig]
  :=
  \Tr\!\left[
    \sig\,\CL_{\vG,\CA}^{\dagger}[\vO]
  \right].
  \label{eq:stationarity-test}
\end{equation}
\end{definition}

\noindent For a single jump $a\in \CA$, we write \(\mathsf{Test}_{a}\), and in a slight abuse of notation for a site $i\in [n]$ we write \(\mathsf{Test}_{i} = \mathsf{Test}_{\CS_i}\) for the sum over single-site Pauli operators. We refer the reader back to \cref{section:intro-stat-test} for its basic properties, namely the completeness condition (\cref{rmk:completeness-intro}) and the relationship to the Fisher information (\cref{rmk:stat-fisher-intro}). 

As we discuss in \cref{section:gibbs-alg} (\cref{lemma:stationarity-test-single-measurement}), the efficiency of a single measurement under said test can be reduced to simulation of the Lindbladian evolution \cite[Theorem I.2]{chen2023efficient}. Due to Lieb-Robinson bounds, such tests can also be performed quasi-locally:

\begin{lemma}
    [Locality of the stationarity test]\label{lemma:stest-local}
    Let $\vO$ be a local observable of support $\sA\subseteq [n]$, $\vA^a$ be a single-site Pauli operator with support within $\sA$, and $\sB[\sA, \ell]\subseteq [n]$ the region of qubits of radius $\ell>0$ around $\vO$. Then, 
    \begin{equation}
    \bigg| \mathsf{Test}_{a}[\vO,\vG;\sig] - \mathsf{Test}_a[\vO,\vG_{\sB[\sA, \ell]};\sig_{\sB[\sA, \ell]}]\bigg| \leq c_1\cdot e^{-c_2\cdot \ell} \cdot \|\vO\|,
\end{equation}
with $c_1, c_2$ explicit constants as a function of the Hamiltonian parameters. 
\end{lemma}
\begin{proof}
    Follows directly from the Lieb-Robinson bound of \cref{lem:lieb-robinson-commutator}, see \cref{lemma:lipschitz} for explicit constants.
\end{proof}

The main result of this section is the following quantitative \textit{soundness} guarantee. If a given guess $\vG$ approximately satisfies the stationarity test for every choice of local observable $\vO$ (and in particular, the Fisher information of the sample $\vrho_{\vH}$ relative to $\vrho_{\vG}$ is small), then it must also be close to the ground truth $\vH$. 

Henceforth we assume $\beta\cdot \sd\geq 1$ (i.e. low temperatures) and simply rescale later (see \cref{remark:rescale}).

\begin{lemma}[Convexity of the Fisher Information]\label{lemma:fisher-convexity}
    Fix $\beta\geq \sd^{-1}$. Let $\vG=\sum_{\gamma\in \Gamma} g_\gamma \vP_\gamma, \vH=\sum_{\gamma\in \Gamma} h_\gamma \vP_\gamma$ be $(\sd, \sD, \scgr_{\sgr})$ Hamiltonians (\cref{def:polynomial-growth-hamiltonian}) such that $\max_{\gamma\in \Gamma}|h_\gamma-g_\gamma|\leq \eta$. 
    
    Then, for any site $i\in [n]$ and $\Omega>0$:
    \begin{align}
        \max_{\gamma\ni i} |g_\gamma-h_\gamma|\leq c_1\cdot \bigg(e^{\Omega}\cdot \FI_{\vG,i}(\vrho_{\vH}|| \vrho_{\vG})^{1/2} + \scgr_{\sgr}\cdot e^{-c_2\Omega}\cdot \eta  \bigg)
    \end{align}
    for an appropriate set of constants $0<c_1\leq e^{\poly(\beta\sd, \sD)} , c_2\propto (\beta\cdot \sd)^{-1}$. 
\end{lemma}

\subsection{Convexity of the Fisher information}
\label{subsection:fixed-lattice-local-faithfulness}


The first step of the proof of \cref{lemma:fisher-convexity} lies in relating the Fisher information to gradients (i.e. commutators) of the log-likelihood function. In \cref{lem:uniform-truncated-bare-reconstruction} we present a general statement which holds for arbitrary sample states $\vsigma$, which will later be relevant in our generalizations. A relevant special case (discussed in our overview, \cref{section:overview}) is presented in \cref{rmk:fisher-special}.

\begin{lemma}[The Fisher Information and gradients of the log-likelihood]
\label{lem:uniform-truncated-bare-reconstruction}  Fix $\beta\geq \sd^{-1}$.
Let $\vG=\sum_{\gamma\in \Gamma} g_\gamma \vP_\gamma$ be a $(\sd, \sD, \scgr_{\sgr})$ Hamiltonian, $\vsigma$ a full-rank state, and $\vA^a$ a single-site Pauli operator. 
Then, for any $\Omega>0$ and operator $\vX:=\sum_{\gamma\in \Gamma} x_\gamma \vP_\gamma$ s.t. $|x_\gamma|\leq \eta$:
 \begin{align}
        \|[\vA^a,\vX]\|_{\vsigma}^2 \leq c_1 (e^{\Omega}\cdot \QF_{a}(\vX, \vG;\vsigma) + \scgr_{\sgr}^2\cdot   \eta^2 \cdot e^{-c_2\Omega}),
    \end{align}
    where $\QF_{a}$ is the Fisher information functional \eqref{eq:fisher-quadratic}. 
    
    In addition, a similar (but weaker) statement is true in terms of the log-likelihood:
    \begin{align}\label{eq:sigma-fisher-bound}
         \|[\vA^a,\beta \vG+\log \vsigma]\|_{\vsigma}^2 \leq  c_1\cdot \big( e^{\Omega}\cdot \FI_{\vG,a}(\vsigma|| \vrho_{\vG}) + e^{-c_2\Omega}\cdot \|\beta \vG+\log \vsigma\|^2\big).
    \end{align}
    In the above, we pick appropriate constants  $c_1 \leq e^{\poly(\beta\sd, \sD)}, c_2 \propto (\beta \sd)^{-1}$.
\end{lemma}

We pay particular attention to the case $\vsigma$ is the Gibbs state $\vrho_{\vH}$:

\begin{remark}\label{rmk:fisher-special}
    A special case of the above, when $\vX:=\vG-\vH$ and $\vsigma:=\vrho_{\vH}$, is the following bound:
    \begin{align}
     \|[\vA^a,\vG-\vH]\|_{\vrho_{\vH}}^2 \leq c_1 (e^{\Omega}\cdot \FI_{\vG,a}(\vrho_{\vH}|| \vrho_{\vG}) + \scgr_{\sgr}^2\cdot   \eta^2 \cdot e^{-c_2\Omega}).
\end{align}
\end{remark}

To conclude the proof of \cref{lemma:fisher-convexity}, we use the following lemma of  \cite{chen2025learning} which relates the norms of commutators to the local coefficients. We point out that related lower bounds on weighted variances/norms of local observables were leveraged in \cite{anshu2020sample}.

\begin{lemma}
    [{\cite[Lemma~III.6]{chen2025learning}}]\label{lemma:faithfulness} Let $\vB$ be a local operator with support on a region $\sR$. Then, for any  Hamiltonian $\vG$ of bounded degree $\sd$ and $\beta\geq\sd^{-1}$:
    \begin{equation}
        \|\vB\|_{\vtau_\sR}\leq c^{|\sR|\cdot \beta \sd\log 2\beta \sd}\cdot \|\vB\|_{\vrho_{\vG}}.
    \end{equation}
    where $c$ is a universal constant. 
\end{lemma}

We are now in a position to conclude the proof of \cref{lemma:fisher-convexity}. We note this final computation is based on \cite[Lemma III.7]{chen2025learning}.

\begin{proof}
    Let $i\in [n]$ be a single qubit, $\sR\subseteq [n]$ be the constant-sized neighborhood of said qubit and $\vX=\sum_{\gamma\ni i} x_\gamma \vP_\gamma$ an operator with support on $\sR$. Then, the coefficients of $\vX$ are determined by commutators of Paulis on $i$:
    \begin{equation}
  \sum_{\vA^a\in\CS_i}
  \|[\vA^a,\vX]\|_{\vtau_{\sR}}^2
  =
  2^{2\sq+1}
  \sum_{\gamma\ni i}
  x_\gamma^2.
  \label{eq:can-one-site-coefficient-identity}
\end{equation}
\cref{lemma:faithfulness} and \cref{lem:uniform-truncated-bare-reconstruction} (\cref{rmk:fisher-special}) then imply that for any $\Omega>0$
\begin{align}
    \max_{\gamma\ni i} |h_\gamma-g_\gamma| \leq c_1\cdot \bigg(e^{\Omega}\cdot \FI_{\vG,i}(\vrho_{\vH}|| \vrho_{\vG})^{1/2} + \scgr_{\sgr}\cdot e^{-c_2\Omega}\cdot \max_{\gamma} |h_\gamma-g_\gamma|  \bigg)
\end{align}
for an appropriate set of constants $c_1\leq e^{\poly(\beta\sd, \sD)}, c_2\propto (\beta\cdot \sd)^{-1}$.

\end{proof}

\subsubsection{Proof of \cref{lem:uniform-truncated-bare-reconstruction}}

We note to the reader that the computation here is based on the Dirichlet form comparison argument performed in \cite[Theorem VI.1]{BC26}, developed in a study of mixing times for 1D non-commuting systems. However, here we focus our attention to a specific observable $\vX$ (the log-likelihood) and replace the Dirichlet form of \cite{chen2023efficient} by its Fisher information functional. At a high level, the proof proceeds by expanding the local operator $\vA$ as a linear combination of its operator Fourier transforms; the low frequency components can easily be compared to the Fisher information, while the high-frequency components are amenable to locality approximations.

\begin{proof}

We first rewrite a single operator $\vA$ in terms of a linear combination of its OFT $\hat{\vA}_{\vG, \sigma_1}(\omega)$, for a suitably chosen uncertainty $\sigma_1 = 1/(\sqrt{2}\beta)$. Then, we split the integral over $\omega$ into ``high'' and ``low'' frequency components, tuned by a threshold $\Omega>0.$
\begin{align}
   \sqrt{2\sigma_1\sqrt{2\pi}} \vA &= \int_{-\infty}^{\infty} \hat{\vA}_{\vG, \sigma_1}(\omega) \rd \omega  \\
    & = \undersetbrace{\text{low-frequency}}{\sqrt{\frac{\sigma_2\sigma_3}{\sigma_1\sqrt{2\pi}}}\cdot  \int_{\labs{\omega}\le \Omega} \int_{-\infty}^{\infty} \hat{\vA}_{\vG, \sigma_2}(\omega)(t) e^{-i \omega t}f_{\sigma_3}(t) \rd t \rd \omega} + \undersetbrace{\text{high-frequency}}{\int_{\labs{\omega}> \Omega} \hat{\vA}_{\vG, \sigma_1}(\omega) \rd \omega}.\label{eq:A_high_low}
\end{align}
Where we applied the convolution property of OFTs (\cref{lem:OFT_Gaussian_expansion}). The uncertainties are defined to satisfy $1/\sigma_1^2 = 1/\sigma_2^2 +1/\sigma_3^2$, and in particular,  $\sigma_2 = \sigma_3 = 1/\beta:=\sigma$. We proceed by taking the $\vsigma$-weighted norm with the commutator $\norm{[\cdot,\vX]}_{\vsigma}$ on both sides.\\

\noindent \textbf{Low-frequency components.} The low frequency $\labs{\omega} \le \Omega$ part can be bounded by the Fisher information functional of the local Lindbladian $\CL_{\vG, a}$, with respect to the sampled state $\vsigma$. In order to do so, we first observe that over $|s|\leq 4^{-1}$ and $|\omega|\leq \Omega$, the coefficients of the Fisher information functional (\cref{def:relative-fisher-information}) are bounded by that of a Gaussian:
\begin{equation}
   \int_{-1/4}^{1/4} \int \int_{|\omega| \leq \Omega}  \frac{f_{\sigma}^2(t)}{g_s(t)h_s(\omega)} \cdot \rd\omega \rd t \rd s \lesssim e^{\beta\Omega}
\end{equation}
and furthermore, by the AM-GM inequality, the $(\vsigma, s)$-weighted inner products can be related to the KMS inner product \cref{lemma:kms-to-s}:
\begin{align}
   \norm{[\hat{\vA}_{\vG, \sigma}(\omega, t),\vX]}_{\vsigma} \leq 4 \int_{-1/4}^{1/4}  \norm{[\hat{\vA}_{\vG, \sigma}(\omega, t),\vX]}_{\vsigma, s}  \rd s
\end{align}
The Cauchy-Schwarz inequality over $\rd \omega \rd t\rd s$ then produces the desired Fisher information functional: 
\begin{align}
    \int_{\labs{\omega}\le \Omega} \int_{-\infty}^{\infty} \norm{[\hat{\vA}_{\sigma}(\omega)(t),\vX]}_{\vsigma}\cdot  \labs{f_{\sigma}(t)} \rd t \rd \omega 
    &\leq 4\sqrt{\iiint_{\labs{\omega}\le \Omega} \frac{f^2_{\sigma}(t)}{g_s(t)h_s(\omega)}\rd t \rd \omega \rd s} \sqrt{\QF_{a}(\vX, \vG;\vsigma)}\\
    & \lesssim e^{\beta\Omega/2}\cdot \sqrt{\QF_{a}(\vX, \vG;\vsigma)}.
\end{align}
Finally, we note that when $\vX:=\vG-\vH$ we have $\beta^2\cdot \QF_{a}(\vG-\vH, \vG;\vrho_\vH)=  \FI_{\vG,a}(\vrho_{\vH}|| \vrho_{\vG})$.\\

\noindent \textbf{High frequency components.} Next, we turn to the high-frequency components. Here, we leverage a shell decomposition for the OFT which decays in both frequency and space (\cref{lemma:oft-loc}), presented in the appendix for conciseness:
\begin{align}
    \hat{\vA}_{\vG, \sigma_1}(\omega) = \sum_{\ell} \vE_{\ell, \omega}, \quad \|\vE_{\ell, \omega}\|\leq c_1\cdot e^{-|\omega|/8\sd-c_2 \ell}
\end{align}
where the support of $\vE_{\ell, \omega}$ is the ball $\CB(a, \ell)$ of radius $\ell$ around $\vA^a$, and $c_1\leq e^{\poly(\beta\sd)}$. By applying the unconditional bound on the KMS norm as in \cref{lemma:kms-to-s}:
\begin{align}
    \int_{|\omega|\geq \Omega}  \norm{[\hat{\vA}_{\vG, \sigma_1}(\omega),\vX]}_{\vsigma}\rd \omega \leq \sum_\ell  \int_{|\omega|\geq \Omega} \norm{[\vE_{\ell,\omega},\vX]}_{\vsigma} \rd\omega  \leq \sum_\ell  \int_{|\omega|\geq \Omega} \norm{[\vE_{\ell,\omega},\vX]}\rd\omega
\end{align}

The statement in \eqref{eq:sigma-fisher-bound} follows from the statement above with $\vX:=\beta\vG+\log \vsigma$. 

To conclude the proof, we address the case where $\vX:=\sum_{\gamma\in \Gamma} x_\gamma \vP_\gamma$ with $|x_\gamma|\leq \eta$. To proceed, we leverage the fact $\vX$ is a sum over local terms on $\Gamma$ and thus by assumption has the same geometry as a $(\sd, \sD, \scgr_{\sgr})$ Hamiltonian. We then have  
\begin{align}
    \sum_\ell  \int_{|\omega|\geq \Omega} \norm{[\vE_{\ell,\omega},\vX]} \rd\omega &\leq \eta \cdot c_1\cdot \sum_\ell |\CB(a, \ell+1)|\cdot  e^{-c_2\ell}  \int_{|\omega|\geq \Omega}e^{-|\omega|/8\sd} \rd \omega \\
    &\leq c_1'\cdot \eta\cdot \scgr_{\sgr} \cdot  e^{-\Omega/8\sd} \cdot \sum_{\ell}(\ell+1)^\sD\cdot e^{-c_2 \cdot \ell} \\
    &\leq c''_1\cdot \eta\cdot \scgr_{\sgr} \cdot  e^{-\Omega/8\sd}
\end{align}
where again $c''_1\leq e^{\poly(\beta\sd)}$. Appropriately relabeling the constants then gives the desired bound. 
\end{proof}

\newpage

\section{Learning Hamiltonian parameters from their Gibbs states}
\label{section:gibbs-alg}

In this section, we reason that measurements given by the stationarity test can be leveraged to design a learning algorithm which recovers the coefficients of a given lattice Hamiltonian $\vH$ from samples of its Gibbs state $\vrho_\vH$. Subsequently, in \cref{section:structure-learning,section:metastable-learning}, we present modifications to structure learning and to metastable states.
Although the results of the next sections are strict generalizations, we find the presentation here serves an instructive starting point to the analysis. 

The main result of this section is the following theorem, on the sample and time complexity of the learning algorithm. Henceforth we treat the Hamiltonian parameters $\sd, \sD, \scgr_{\sgr}, \sq, \sk$ as fixed constants, and consider the asymptotics in $n, \eta, \delta, \beta$. 

\begin{theorem}[Learning Hamiltonians from their Gibbs states]\label{theorem:gibbs-learning} Let $\vH:=\sum_{\gamma\in \Gamma} h_\gamma \vP_\gamma$ be a $(\sd, \sD, \scgr_{\sgr})$ Hamiltonian on a known set of interactions $\Gamma$. For any $\eta, \delta\in (0, 1)$, $\beta\in \BR^+$, with probability $\geq 1-\delta$, \cref{alg:gibbs-learner,alg:one-round} return a Hamiltonian $\vG=\sum_{\gamma\in \Gamma} g_\gamma \vP_\gamma$ satisfying $\max_{\gamma\in \Gamma}|h_\gamma-g_\gamma|\leq \eta$ using at most 
    \begin{equation}
     \mathsf{N_{samples}}:=  \frac{ e^{\poly(\beta)}}{\beta^2}\cdot O\bigg(\eta^{-2}\log \frac{n}{\delta\eta}\log^\sD \frac{1}{\eta}\bigg)
    \end{equation}
    copies of the Gibbs state $\vrho_{\vH}$ and $\poly(n,  \mathsf{N_{samples}})$ quantum and classical processing time. 
\end{theorem}

Roughly speaking, the Hamiltonian learning algorithm follows an iterative procedure akin to that designed in \cite{chen2025learning}, adapted to our stationarity test. In each (say, the $r$th) iteration a guess $\vG^{(r)}$ for the Hamiltonian is refined to increasing accuracy $\eta_r$ by enumerating over constant-sized $\eta_r-$nets (ensuring time-efficiency), while the quasi-locality of the stationarity test (\cref{lemma:stest-local}) ensures the measurements are parallelizable.

In \cref{section:H-alg}, we give a formal account of the algorithm. In \cref{section:H-correct}, we prove its correctness, and in \cref{section:H-complexity}, we describe the sample and the time complexities.

\subsection{The algorithm}
\label{section:H-alg}

We begin with a broad description of the iterative learning algorithm.

\begin{algorithm}[H]
\caption{Iterative Hamiltonian learner from Gibbs samples}
\label{alg:gibbs-learner}
\begin{algorithmic}[1]
\Require Copies of \(\vrho_{\vH}\), accuracy \(0<\eta\le1\), failure probability
\(0<\delta<1\)
\State Set \(\vG^{(0)}\gets0\) and
  \(T\gets\lceil\log_2(1/\eta)\rceil\)
\For{\(r=0,\ldots,T-1\)}
  \State Set \(\eta_r\gets2^{-r}\)
  \State Run \Cref{alg:one-round} with failure probability
    \(\delta/\max\{1,T\}\), and let its output be \(\vG^{(r+1)}\)
\EndFor
\State \Return \(\vG^{(T)}\)
\end{algorithmic}
\end{algorithm}

\Cref{alg:one-round} will be based on the stationarity test (\cref{def:stationarity-test}), where we enumerate over local perturbations of $\vG^{(r)}$, and accept that which minimizes the Fisher information. To make this description precise, we introduce three algorithm parameters to be specified in the correctness argument (\cref{section:H-correct}). 

\begin{enumerate}
    \item A (sufficiently large) constant \(\ell\), and denote the Hamiltonian terms within distance $\ell$ away from a given site $i\in [n]$ as $\CB[i, \ell]
  :=
  \left\{
    \gamma\in\Gamma:
    \dist(\supp(\vP_\gamma),i)\le \ell
  \right\}$. The intention will be to enumerate over possible Hamiltonian terms in this radius.

  \item A grid spacing \(0<\kappa<1\), and denote the net of local perturbations of $\vG^{(r)}$ as $\vG^{(r, \vU)}
  :=
  \vG^{(r)}+\eta_r\vU$, where $ \vU\in\Grid_i:$
\begin{equation}
  \Grid_i
  :=
  \left\{
    \vU=\sum_{\gamma\in \CB[i, \ell]}u_\gamma\vP_\gamma:
    u_\gamma\in\kappa\BZ\cap[-1,1]
  \right\}.
  \label{eq:fixed-local-grid}
\end{equation}
\item A quasi-locality radius $R^{(r)}=2\ell + c_{\mathrm{rad}}\cdot r$, which scales logarithmically in the target accuracy $\eta_r=2^{-r}$ in said round.
\end{enumerate}

We are now in a position to describe the iteration. 

\begin{algorithm}[H]
\caption{The $(r+1)$st iterative round.}
\label{alg:one-round}
\begin{algorithmic}[1]
\Require Current estimate \(\vG^{(r)}\), current error guarantee \(\eta_r\), fresh copies
of \(\rhoH\)
\For{each site \(i\), in parallel}
  \For{each \(\vU\in\Grid_i\) and   \(\gamma\in\mathcal \CB[i, \ell]\)}
    \State Compute estimates $\hat{\mathsf{T}}_i[\vU, \vP_\gamma]$ for $|\mathsf{Test}_{i}[\vP_\gamma,\vG^{(r)}_{\CB[i, R^{(r)}]} + \eta_r \vU;\vrho_{\vH}]|$  to error
      \(c_{\mathrm{meas}}\cdot \eta_r\)
  \EndFor
  \State Choose \(\vU_i\) minimizing the above test, $ \vU_i:=\arg \min_{\vU\in\Grid_i} \max_{\gamma\in\mathcal \CB[i, \ell]} \hat{\mathsf{T}}_i[\vU, \vP_\gamma].$
  \For{each \(\gamma\in\CB[i, 0]\)}
    \State Set \(g_\gamma^{(r+1)}\gets g_\gamma^{(r)}+\eta_r u_{i,\gamma}\)
  \EndFor
\EndFor
\State \Return \(\vG^{(r+1)}\)
\end{algorithmic}
\end{algorithm}

That is, the current Hamiltonian $\vG^{(r)}$ is perturbed by $\vU$ within the radius $\ell$, and subsequently truncated within the radius $R^{(r)}$ to ensure efficient and quasi-local implementation. We detail how to compute said estimates by short-time Lindbladian evolution in \cref{section:H-complexity}. After the minimizing refinement is found, the Hamiltonian terms incident on the center $i$ are updated. 

\subsection{Correctness}
\label{section:H-correct}

Correctness of the algorithm consists of proving that upon each iteration, the target error $|g_\gamma^{(r)}-h_\gamma|\leq \eta_r$ halves at each round. 

\begin{lemma}[Correctness of a single iteration]
\label{lemma:gibbs-one-round}
Fix $\beta\sd\geq1$. There exists a choice of algorithm parameters $\ell, \kappa, c_{\mathrm{meas}}, c_{\mathrm{rad}}$ satisfying the following guarantee. Suppose at the $r$th iteration, $\max_{\gamma\in\Gamma}|g_\gamma^{(r)}-h_\gamma|
  \le\eta_r.$ 
If every estimate in \Cref{alg:one-round} is within its specified error $c_{\mathrm{meas}}\cdot \eta_r$, then
its output satisfies
\begin{equation}
  \max_{\gamma\in\Gamma}|g_\gamma^{(r+1)}-h_\gamma|
  \leq  \eta_{r+1}=\eta_r/2.
  \label{eq:one-round-conclusion}
\end{equation}
\end{lemma}

It then follows trivially that repeated application of \cref{lemma:gibbs-one-round} guarantees the correctness of \cref{alg:gibbs-learner} + \cref{alg:one-round}.

\begin{proof}
Fix a site $i\in [n]$, an iteration round $r$, an accuracy $\eta_r$, and a grid point \(\vU\in\Grid_i\). We denote as
\begin{equation}
\vK^{(\vU)}
  :=
  \vH+
  \sum_{\gamma\in \CB[i,\ell]}
  (g_\gamma^{(r)}+\eta_r\cdot u_\gamma -h_\gamma)\vP_\gamma.
  \label{eq:completed-candidate}
\end{equation}
the Hamiltonian which agrees with $\vH$ outside of $\CB[i,\ell]$, and agrees with $\vG^{(r, \vU)}$ on $ \CB[i,\ell]$. Then,

\begin{enumerate}
    \item By the Lieb-Robinson bound of \cref{lemma:lipschitz} (\cref{lemma:lipschitz-within}), the test outcomes on $\vG^{(r, \vU)}$ and $\vK^{(\vU)}$ are close, up to an error scaling with both $\eta_r$ and $\ell$:
    \begin{align}
    \max_{\gamma\in\CB[i,\ell]}  |\mathsf{Test}_{i}[\vP_\gamma,\vG^{(r, \vU)};\vrho_{\vH}] - \mathsf{Test}_{i}[\vP_\gamma,\vK^{(\vU)};\vrho_{\vH}]| \leq \poly(\beta\sd)\cdot \eta_r\cdot  e^{-c_1\cdot \ell}\label{eq:LB1}
\end{align}
where $c_1 \propto (\beta \sd)^{-1}$. Here we explicitly suppress the dependence on $\sD, \mathsf{c}_{\sgr}$.

\item There exists a choice of $\vU^*\in \Grid_i$ such that $|g_\gamma^{(r)}+\eta_r\cdot u_\gamma^* -h_\gamma|\leq \kappa\eta_r$ for each $\gamma\in \CB[i, \ell]$. Now by Lipschitz continuity \cref{lemma:lipschitz} (\cref{lemma:lipschitz-outside}), the test outcomes on $\vH$ and this $\vK^{(\vU^*)}$ are close, up to an error scaling with $\eta_r$ and the spacing $\kappa$.
\begin{equation}
     \max_{\gamma\in\CB[i,\ell]}  |\mathsf{Test}_{i}[\vP_\gamma, \vH;\vrho_{\vH}] - \mathsf{Test}_{i}[\vP_\gamma,\vK^{(\vU^*)};\vrho_{\vH}]| \leq \poly(\beta\sd)\cdot \eta_r\cdot |\CB[i,\ell]|\cdot \kappa\label{eq:Lip1}
\end{equation}

\item Let $\vV\in \Grid_i$ be the perturbation chosen by the inner loop of \cref{alg:one-round}, which minimizes the truncated stationarity tests. By the quasi-locality of the stationarity test (\cref{lemma:stest-local}, with parameters given by \cref{lemma:lipschitz-within}) and the estimation error bound, it approximately minimizes the stationarity test:
\begin{align}
   \max_{\gamma\in\CB[i,\ell]}  |\mathsf{Test}_{i}[\vP_\gamma,\vG^{(r, \vV)};\vrho_{\vH}]| \leq&   \max_{\gamma\in\CB[i,\ell]}  |\mathsf{Test}_{i}[\vP_\gamma,\vG^{(r, \vU^*)};\vrho_{\vH}]| \\ &+ 2\eta_r\cdot c_{\mathrm{meas}} + \poly(\beta\sd)\cdot |\CB[i, R^{(r)}]| \cdot e^{-c_1\cdot R^{(r)}}\label{eq:LB2}
\end{align}
\end{enumerate}

\noindent \textbf{Put together,} the completeness condition $\mathsf{Test}_{i}[\vP_\gamma, \vH;\vrho_{\vH}]=0$ and the above estimates enable us to compute the value of the test on the idealized version of the algorithm output, $\vK^{(\vV)}$:
\begin{align}
     &\max_{\gamma\in\CB[i,\ell]}  |\mathsf{Test}_{i}[\vP_\gamma,\vK^{(\vV)};\vrho_{\vH}] |  \\ \leq \quad  &\eta_r \cdot c_{\beta, d}  \cdot \big (|\CB[i,\ell]|\cdot  \kappa +|\CB[i,\ell]|\cdot e^{-c_2\cdot \ell} + c_{\mathrm{meas}}\big) + c_{\beta, d}\cdot |\CB[i, R^{(r)}]| \cdot e^{-c_1\cdot R^{(r)}} \\
     := \quad & \eta_r\cdot \theta
\end{align}
where to compact notation we introduced the scalar $\theta$ (to be specified shortly) and $c_{\beta, d}\leq \poly(\beta \sd)$. This enables us to now bound the Fisher information between $\vK^{(\vV)}$ and $\vH$:
\begin{align}
    0 \leq \FI_{\vK^{(\vV)}, i} (\rho_{\vH}||\vrho_{\vK^{(\vV)}})
  &= -\beta \sum_{\gamma\in\CB[i,\ell]}  (g^{(r)}_\gamma+\eta_r\cdot  v_\gamma -h_\gamma)\cdot \mathsf{Test}_{i}[\vP_\gamma,\vK^{(\vV)};\vrho_{\vH}]| \\
  &\leq 2\beta\cdot \eta_r^2\cdot \theta\cdot  |\CB[i,\ell]|
\end{align}
Finally, applying the approximate convexity statement for the Fisher information \cref{lemma:fisher-convexity}, we conclude the error of the coefficient $g^{(r+1)}_\gamma:=g^{(r)}_\gamma+\eta_r\cdot  v_\gamma$ is bounded by
\begin{equation}
  \max_{\gamma\ni i}
  |g^{(r+1)}_\gamma-h_\gamma|
  \le
  e^{\poly(\beta\sd)}\left[
    e^\Omega\sqrt{\beta |\CB[i,\ell]|\theta}
    +
    2e^{-\Omega/(c_4\cdot \beta\cdot \sd)}
  \right]\eta_r.
  \label{eq:inverse-after-score}
\end{equation}
for any choice of $\Omega>0$, and a universal constant $c_4$.\\

\noindent \textbf{Choosing constants.} It only remains now to choose the relevant parameters $\ell, \kappa, \Omega, c_{\mathrm{meas}}$ and $R^{(r)}$. We first pick $\Omega = \poly(\beta\sd)$ to ensure the second term in the RHS above is $\leq \eta_r / 100.$ Subsequently, we pick $\ell$ in order to ensure 
\begin{equation}
  e^{\poly(\beta\sd)+2\Omega}  |\CB[i, \ell]|^2 e^{-c_2\cdot \ell} \leq \frac{1}{10^2} \Rightarrow \ell = \poly(\beta\sd). 
\end{equation}
Analogously, we select the grid spacing $\kappa = e^{-\poly(\beta\sd)}$, measurement error $c_\mathrm{meas}=e^{-\poly(\beta\sd)}$, and finally $R^{(r)} = 2\ell+ \poly(\beta\sd)\cdot \log \frac{1}{\eta_r}$. Summing the sources of error then gives $\eqref{eq:inverse-after-score}\leq \eta_r/3$ which satisfies the desired bound. 

\end{proof}

\subsection{Sample and time complexity}
\label{section:H-complexity}

To conclude our parameter learning algorithm as described in \cref{alg:one-round}, it remains to describe a mechanism to measure the stationarity test (under perturbations to the Hamiltonian). The main observation (captured in \cref{lemma:stationarity-test-single-measurement}), is that such measurements can be performed by evolving the associated local Lindbladian, truncated to the logarithmic radius, for \textit{constant time}.

\begin{lemma}[Measuring the stationarity test by short-time
Lindbladian evolution]
\label{lemma:stationarity-test-single-measurement}
Fix $\beta\sd\geq1$, $\eta, \delta>0, 1\geq  c_{\mathrm{meas}}> 0$. Let $\vH=\sum_{\gamma\in\Gamma}h_\gamma\vP_\gamma$ and
$\vG=\sum_{\gamma\in\Gamma}g_\gamma\vP_\gamma$ be Hamiltonians
on a common $(\sd,\sD,\scgr_{\sgr})$ interaction graph and
$\max_{\gamma\in\Gamma}|g_\gamma-h_\gamma|\leq\eta$.

For any site $i$, Pauli observable $\vP$, and $\delta\in(0,1)$,
the stationarity test $ \mathsf{Test}_i[\vP,\vG;\vrho_{\vH}]$ can be estimated to additive error $c_{\mathrm{meas}}\eta$,
with probability at least $1-\delta$, using
\begin{equation}
    \mathsf{N_{samp}}=\poly(\beta\sd)\cdot 
    O\!\left( (c_{\mathsf{meas}}^2 \eta)^{-2}\cdot \log\frac{2}{\delta}\right)
\end{equation}
samples of $\vrho_{\vH}$. Furthermore, the procedure evolves the samples of  $\vrho_{\vH}$ under a truncation of the Lindbladian associated to $\vG$ to radius $\poly(\beta  d) \cdot \log (c_{\mathsf{meas}}\cdot \eta)^{-1}$ for time $c_{\mathsf{meas}} / \poly(\beta\sd)$.
\end{lemma}

\begin{proof}
Let $\CL$ denote the Lindbladian associated to the truncation of $\vG_{\CB[i, \mathsf{R}]}$ to radius $\sR$. It follows from the Lieb-Robinson bound (e.g. \cref{lemma:lipschitz-within}) and the continuity lemma  \cref{lemma:lipschitz}, that the derivative of $\vrho_{\vH}$ under $\CL$ is given by:
\begin{align}
    \|\CL[\vrho_{\vH}]\|_1
    &\leq
    \|\CL-\CL_{\vH,i}\|_{1\to1} \leq
    \|\CL-\CL_{\vG,i}\|_{1\to1} +   \|\CL_{\vG,i}-\CL_{\vH,i}\|_{1\to1} \\
&\leq \poly(\beta\sd) \cdot \bigg(e^{-\theta\cdot \sR}+\eta\cdot \sum_{z\geq 0} |\CB[i, z]|\cdot e^{-\theta\cdot z}\bigg) \\ &\leq  \poly(\beta\sd)\cdot \eta \cdot \bigg(1+\sum_{z\geq 0} z^{\sD}\cdot e^{-\theta\cdot z}\bigg) \\
    &\leq \poly(\beta\sd) \cdot \eta,
\end{align}
where we omit factors of $\mathsf{c_{gr}}, \sD$, and make a suitable choice of $\poly(\beta\cdot d) \cdot \log (\min(1, c_{\mathsf{meas}})\cdot \eta)^{-1}$. In other words, $\vrho_{\vH}$ is $\sim \eta$ locally metastable (\cref{def:intro-local-metastability}). The algorithm introduces $t>0$, estimates $\tr[\vP\cdot \vrho_{\vH}]$ and $\tr[\vP\cdot e^{t\CL}[\vrho_{\vH}]]$ to error $ \frac{1}{4}\cdot c_{\mathsf{meas}}\cdot \eta\cdot t$. It then returns:
\begin{align}
    \hat{\mathsf{T}}[\vP]:=\frac{1}{t}\cdot \bigg(\tr[\vP\cdot e^{t\CL}[\vrho_{\vH}]]-\tr[\vP\cdot \vrho_{\vH}]\bigg).
\end{align}
Duhamel's formula entails that the resulting error due to the finite time estimation is:
\begin{align}
   \bigg | \hat{\mathsf{T}}[\vP]-\mathsf{Test}_i[\vP,\vG;\vrho_{\vH}]\bigg| &\leq \|\CL-\CL_{\vG, i}\|_{1-1} +    \bigg | \hat{\mathsf{T}}[\vP]-\tr[\CL^\dagger[\vP]\vrho_{\vH}]\bigg| \\
   &\leq \frac{1}{t}\int_0^t(t-s)\cdot 
    \|e^{s\CL}\CL^2[\vrho_{\vH}]\|_1\rd s + \poly(\beta\cdot d) e^{-\theta\cdot \sR} \\ 
    &\leq \frac{1}{4} c_{\mathsf{meas}}\cdot \eta + t\cdot \eta\cdot  \poly(\beta\cdot \sd) \leq \frac{1}{2} c_{\mathsf{meas}}\cdot \eta \label{eq:duhamel-estimation-efficient}
\end{align}
where we made the explicit choice of $t_\beta:= c_{\mathrm{meas}} / \poly(\beta\sd)$. Finally, a median-of-means estimation returns each desired Pauli expectation value w.p. $1-\delta/2$ using the claimed sample complexity.
\end{proof}

We now are in a position to settle the sample complexity of \cref{alg:gibbs-learner,alg:one-round}. We essentially combine \cref{lemma:stationarity-test-single-measurement} with a coloring argument to ensure parallel measurements of the various perturbed stationarity tests.

\begin{lemma}
    [Sample complexity]\label{lemma:gibbs-sample} Assume $\beta > \sd^{-1}$. The Hamiltonian learning algorithm described in \cref{alg:gibbs-learner,alg:one-round} requires a total number of 
    \begin{equation}
        \mathsf{N_{samp}}=  e^{\poly(\beta\sd)}\cdot  O\bigg(\eta^{-2}\cdot \log \frac{n}{\delta\eta }\cdot \log^\sD \eta^{-1}\bigg)
    \end{equation}
    copies of $\vrho_{\vH}$.
\end{lemma}

Rescaling the desired error by $\beta\cdot \sd$ then gives the claimed bound on the sample complexity. 

\begin{remark}[Rescaling $\beta$]\label{remark:rescale} We assume wlog throughout the proof that  $\alpha:=\beta \sd\geq 1$ i.e. the low temperature regime. If $\alpha<1$ we rescale $(\beta', \vH', \eta') = (\alpha^{-1}\beta, \alpha\vH, \alpha\eta)$. The Gibbs state and Lindbladian are unchanged, and thus it suffices to learn $\vH'$ at inverse temperature $\beta'$
to coefficient error $\eta'$ and divide the resulting estimate by $\alpha$. By appropriately rescaling $\eta_r:=\alpha\cdot 2^{-r}$, the iteration count and truncation radius can be chosen independent of $\alpha$. Thereby this has the effect of introducing multiplicative factors of $\alpha^{-2}$ to the sample complexity of \cref{lemma:stationarity-test-single-measurement}. See \cite[Section IV]{chen2025learning} for related arguments.
\end{remark}

\begin{proof}

    [of \cref{lemma:gibbs-sample}] The total number of iterations is $T=\ceil{\log\eta^{-1}}$. We refer the reader to the proof of \cref{lemma:gibbs-one-round} for the choice of constants.

    The expectation value $\mathsf{Test}_{i}[\vP_\gamma,\vG^{(r)}+\eta_r{\vU};\vrho_{\vH}]$ in \cref{alg:one-round} can be estimated to additive error $\epsilon_\mathsf{test} = c_{\mathrm{meas}}\cdot \eta_r=\eta_r \cdot e^{-\poly(\beta\sd)}$ with probability of failure $\delta_\mathsf{test}= O(\delta\cdot n^{-1}\cdot \kappa^{|\CB[i, \ell]|}\cdot 1/T)$ using $e^{\poly(\beta\sd)}\cdot O(\epsilon_\mathsf{test}^{-2}\cdot \log \frac{1}{\delta_\mathsf{test}}) = e^{\poly(\beta\sd)}\cdot O(\eta_r^{-2}\cdot \log \frac{nT}{\delta})$ samples by \cref{lemma:stationarity-test-single-measurement}. By a union bound, all experiments in \cref{alg:gibbs-learner,alg:one-round} succeed with probability $\geq 1-\delta$.
    
    We assume these Hamiltonians are $(\sd, \sD, \scgr)$ local, and thus any given test in round $r$ overlaps with at most $\mathsf{s}_r=e^{\poly(\beta
    \sd)}O((R^{(r)})^{\sD}) = e^{\poly(\beta
    \sd)}\cdot O(\log^{\sD}\eta_r^{-1})$ other tests due to their finite radius. By a coloring argument, we can partition the set of tests into $\mathsf{s}_r$ disjoint subsets, such that all tests in each subset can be performed in parallel on the same samples. The total sample complexity of such tests, across the $T$ rounds is:
    \begin{align}
      e^{\poly(\beta\sd)}\cdot \sum_{r=1}^T  O\bigg((R^{(r)})^{\sD}\cdot  \eta_r^{-2}\cdot \log \frac{nT}{\delta}\bigg) &\leq  e^{\poly(\beta\sd)}\cdot O(\log \frac{nT}{\delta})\cdot \sum_{r=1}^T (1+r)^{\sD}\cdot 4^r  \\ &\leq  e^{\poly(\beta\sd)} O\bigg(\eta^{-2}\cdot \log \frac{n}{\delta\eta }\cdot \log^\sD \eta^{-1}\bigg)\label{eq:geo-sum}
    \end{align}
    as advertised.
\end{proof}

Next we argue about the time-complexity of the learning algorithm, which ultimately reduces to implementing the truncated Lindbladian. We refer the reader to \cite[Theorem I.2, Section III.C]{chen2023efficient} for details and \cite[Lemma IV.5]{chen2025learning} for related calculations. 

\begin{lemma}
    [Time complexity]\label{lemma:gibbs-time} The quantum and classical circuit complexity of \cref{alg:gibbs-learner,alg:one-round} is $e^{\poly(\beta\sd)}\cdot n \cdot \poly(\eta^{-1}, \log n,  \log \delta^{-1})$.
\end{lemma}

\begin{proof}
    Following \cref{lemma:stationarity-test-single-measurement}, a single sample in \cref{alg:one-round} reduces to simulation of the \cite{chen2023efficient} local Lindbladian on a patch of radius $O(\poly(\beta\sd)\log \frac{1}{\eta})$ of a finite-dimensional lattice Hamiltonian. Following \cite[Theorem I.2, Section III.C]{chen2023efficient}, said simulation can be approximated to error $\eta$ using $O\big(\poly(\beta\sd,\log \frac{1}{\eta})\big) $ gates. Multiplied by the total number of shots $n e^{\poly(\beta\sd)}\cdot \poly(\log n, \frac{1}{\eta}, \log \frac{1}{\delta})$ gives the desired claims. 
\end{proof}

\section{Learning Hamiltonian structure from their Gibbs states}
\label{section:structure-learning}

The parameter learning algorithm of \cref{section:gibbs-alg} assumes that the set of Hamiltonian interactions $\Gamma$ is known, and the task is simply to infer their coefficients. In this section, we remove this assumption, and argue that a related algorithm can infer both the underlying $\sD-$dimensional lattice connectivity and the coefficients of the Hamiltonian $\vH$. 

Let $\CS_{[n]}^{ \sk}$ denote the nonidentity Pauli strings supported on at most $\sk$ sites.  We write
\begin{equation}
  \vH=\sum_{\vP\in\mathcal \CS_{[n]}^{\sk}}h_{\vP}\vP,
  \qquad |h_{\vP}|\leq 1,
  \qquad
  \Gamma_{\vH}:=\{\vP:h_{\vP}\neq 0\}.
  \label{eq:structure-universal-expansion}
\end{equation}

As before we assume $\vH$ is a $(\sd, \sD, \scgr_\sgr)$ Hamiltonian for a fixed set of constants and study the asymptotics as $n, \eta, \delta, \beta$ are allowed to scale.  
\begin{theorem}[Structure learning from Gibbs states]
\label{theorem:structure-learning}
Let $\eta,\delta\in(0,1)$, $\beta>0$ and let $\vH$ be an $n$ qudit, $(\sd, \sD, \scgr_\sgr)$ Hamiltonian. Then there exists an algorithm which, with probability $1-\delta$ outputs a Hamiltonian $\vG=\sum_{\vP\in\CS_{[n]}^{ \sk}}
   g_{\vP}\vP$ such that $\Gamma_{\vG}\subseteq \Gamma_{\vH}$ and $\forall \vP\in\CS_{[n]}^{\sk}: |g_\vP-h_\vP|\leq \eta$. The algorithm uses at most
\begin{equation}
    \mathsf{N_{samp}}:=
  \frac{e^{\poly(\beta)}}{\beta^2}
  O\bigg(\eta^{-2}\cdot \log\frac{n}{\delta\cdot \eta}\cdot \log^\sD \frac{1}{\eta}\bigg),
  \label{eq:structure-complexity}
\end{equation}
samples of $\vrho_{\vH}$, and at most $\poly(\mathsf{N_{samp}})\cdot n^{\poly(\beta)}$ processing time.
\end{theorem}

\subsection{The neighborhood guessing algorithm}
\label{subsection:structure-algorithm}



Much like \cref{section:gibbs-alg}, the learning algorithm will proceed in iterations exactly the same as in \cref{alg:gibbs-learner}. In round $r\in (1, 2, \cdots, \lceil \log \frac{1}{\eta}\rceil)$ a candidate ``guess'' Hamiltonian $\vG^{(r)}$ will be iteratively refined to accuracy $\eta_r:=2^{-r}$, where partial information about the underlying structure will be identified; namely, the interaction graph and associated Hamiltonian terms of coefficient magnitude roughly larger than $\eta_r$. 

The only issue lies in how to specify the next guess. The crux will be to only iterate over perturbations to the Hamiltonian which preserve its $(\sd, \sD, \scgr_\sgr)$ locality properties.

\begin{definition}
    Given a site $i\in [n]$, a subset of clauses $\Lambda\subseteq \CS_{[n]}^{\sk}$ is said to be $t-\mathsf{connected}$ to $i$ if for any $e\in \Lambda$ there exists a path of overlapping clauses within $\Lambda$ incident on $i$, of length $\leq t$. We denote as $\mathsf{CC}_{i, t}$ the set of all such subsets.
\end{definition}

Following the description in \cref{section:gibbs-alg}, we fix constants $\ell>1, \kappa, c_{\mathrm{meas}}>0$ to be specified shortly.

The refinement step will be to enumerate over choices of interaction $\Lambda$ (the neighborhood of $i$), and coefficients on said interactions. Following \eqref{eq:fixed-local-grid} we denote the coefficient grid on a given subset of clauses $\Lambda\in \mathsf{CC}_{i, \ell}$ as:
\begin{equation}
  \Grid(\Lambda)
  :=\left\{
    \vU=\sum_{\vP\in\Lambda}u_{\vP}\vP:
    u_{\vP}\in\kappa\BZ\cap[-1,1]
  \right\}.
  \label{eq:structure-core-grid}
\end{equation}

We note that once the $r$th guess $\vG^{(r)}$, and the perturbation $\Lambda, \vU$ are fixed, there is a well-posed notion of distance on the set of qubits $[n]$ induced by the graph $\Gamma_{\vG^{(r)}, \Lambda}:=\Gamma_{\vG^{(r)}}\cup \Lambda$ of the union of clauses. 
Thus we enumerate over Hamiltonian perturbations which preserve the property that the resulting interaction graph is $(\sd, \sD, \scgr_{\mathsf{gr}})-$local. Imposing this locality constraint will later enable us to import the approach from \cref{section:gibbs-alg}.

We are now in a position to state the iterative step:

\begin{algorithm}[H]
\caption{A single round of the iterative structure learning algorithm}
\label{alg:structure-learner}
\begin{algorithmic}[1]
\Require Current estimate \(\vG^{(r)}\), current error guarantee \(\eta_r\), fresh copies
of \(\rhoH\)
\For{each site \(i\in [n]\) }
  \For{each $\Lambda\in \mathsf{CC}_{i,\ell}$ such that $\Gamma_{\vG^{(r)}, \Lambda}$ is $(\sd, \sD, \scgr_\sgr)$ local, and \(\vU\in\Grid(\Lambda)\)}
    \State Compute $\max_{\vP\in\mathcal \CS^{\sk}_{[n]}} |
    \mathsf{Test}_i
    [\vP,\vG^{(r)}+\eta_r\cdot \vU;\vrho_{\vH}]|$ to error
      \(c_{\mathrm{meas}}\cdot \eta_r\) using \cref{lemma:parallel-stat-meas}.
  \EndFor
  \State Choose \(\Lambda_i, \vU_i\) minimizing the above test. 
  \For{each \(\vP\in \CS^{\sk}_{[n]}\) s.t. $i\in \mathsf{supp}(\vP)$}
    \begin{align}
       g_{\vP}^{(r+1)}\gets \begin{cases}
         g_{\vP}^{(r)}+\eta_r u_{i, \vP} & | g_{\vP}^{(r)}+\eta_r u_{i, \vP}|\geq \frac{1}{4}\cdot \eta_r\\
        0& \text{otherwise}.\label{eq:thresholding}
    \end{cases}
    \end{align}
  \EndFor
\EndFor
\State \Return \(\vG^{(r+1)}\)
\end{algorithmic}
\end{algorithm}

The role of the threshold in \eqref{eq:thresholding} will be to ensure each iteration of the algorithm reveals more of the structure of $\vH$: $\Gamma_{\vG^{(r)}}\subseteq \Gamma_{\vG^{(r+1)}}\subseteq \Gamma_{\vH}$, without issuing false positives. The description above omits precisely how the stationarity tests are measured to the desired error, for every $\vP\in\mathcal \CS^{\sk}_{[n]}$. Although a straightforward approach is possible with polynomial sample complexity in system size $n$ (simply by enumerating and measuring akin to \cref{section:gibbs-alg}), for conciseness we defer a discussion on how to achieve the near-optimal sample complexity to \cref{lemma:parallel-stat-meas} in \cref{section:structure-sample}. 

Next, in \cref{subsection:structure-correctness}, we argue the correctness of the structure learning algorithm.

\subsection{Correctness}
\label{subsection:structure-correctness}

The main result of this subsection is \cref{lemma:structure-core-soundness}, which quantifies the accuracy increase after a single step of the algorithm.  A key observation will be that the invariant $\Gamma_{\vG^{(r)}}\subseteq \Gamma_{\vH}$ enables us to leverage that the guess $\Gamma_{\vG^{(r)}}$ -- and its local perturbations $\Gamma_{\vG^{(r)}, \Lambda}:=\Gamma_{\vG^{(r)}}\cup \Lambda$ -- are always lattice Hamiltonians. This allows us to invoke the previously derived locality and convexity statements. We begin with a simple lemma on the geometry of the perturbed Hamiltonians. 

\begin{lemma}
    \label{lemma:connectivity-perturb}
    Let $\vH$ be a $(\sd, \sD, \scgr_\sgr)$ Hamiltonian, and $\Lambda\subseteq \CS^{\sk}_{[n]}$ be a subset of clauses whose interaction graph is of degree $\leq \Delta$. Then, the union of interaction graphs $\Gamma_{\vH, \Lambda}:=\Gamma_{\vH}\cup \Lambda$ is a $(\hat{\sd}, \sD, \hat{\scgr}_{\mathsf{gr}})$ Hamiltonian, where:
    \begin{align}
        \hat{\sd} \leq 2\sk\cdot  (\sd+\Delta), \quad \hat{\scgr}_{\mathsf{gr}}\leq \mathsf{c}_\sgr\cdot (1+2\sk\cdot |\Lambda|).
    \end{align}
\end{lemma}

We make the combinatorial computations above explicit as later, $\Lambda$ and thus $\hat{\mathsf{c}}_\sgr$ will be allowed to scale with $\ell$, while the degree $\hat{\sd}$ and dimension $\sD$ are $\ell$ independent. 

\begin{proof}
    \textbf{The degree.} By the degree bound, at most $\sd+1$ (resp. $\Delta+1$) clauses of $\Gamma_{\vH}$ (resp. $\Gamma_{\Lambda}$) contain any given site $i\in [n]$. Summing over the $\leq \sk$ qubits per clause, we conclude that any clause of $\Gamma_{\vH}$ gains at most $\sk\cdot (\Delta+1)$ neighbors. Thus
    \begin{align}
        \hat{\sd} \leq \max(\sd+\sk(\Delta+1), \Delta+\sk(\sd+1)) \leq 2(\sd+\Delta)\cdot \sk,
    \end{align}
    assuming $\sd, \Delta, \sk\geq 1$ wlog. \\
    
    \noindent \textbf{The growth constant.} Any path in $\hat{\Gamma}$ is either entirely contained in $\Gamma_{\vH}$, or crosses at least one clause in $\Lambda$. We can thus bound the size of the clause balls from any given starting point $x\in [n]\cup \hat{\Gamma}$ by the clause balls in $\Gamma_{\vH}$ of $\{x\}\cup \Lambda.$ For any $r\geq 1:$
    \begin{align}
        \CB_{\hat{\Gamma}}[x, r] \subseteq \CB_{\Gamma_{\vH}}[x, r] \cup \Lambda\cup\CB_{\Gamma_{\vH}}[\mathsf{supp}(\Lambda), r] \quad \Rightarrow \quad |\CB_{\hat{\Gamma}}[x, r]|\leq  (1+(\sk+1) |\Lambda|)\cdot \mathsf{c_{gr}}\cdot r^{\sD}
    \end{align}
    which concludes the claimed bound. 
\end{proof}

We are now in a position to prove the correctness of the iterative step. Roughly speaking, we assume that up to the $r$th iteration, the current guess coefficients $g_{\vP}^{(r)}\approx h_{\vP}$, and furthermore that all sufficiently non-zero coefficients $g_{\vP}$ correspond to sufficiently large coefficients of $h_{\vP}$.\footnote{We remark that this latter condition is a mild strengthening of the assumption $\Gamma_{\vG^{(r)}}\subseteq \Gamma_{\vH}$.}

\begin{lemma}[Correctness of a single iteration]
\label{lemma:structure-core-soundness}
Fix $\beta\sd\geq 1$. There exists choices of $(\ell, \kappa, c_{\mathsf{meas}})$ as a function of $\beta$ and the Hamiltonian parameters, such that the following holds. Suppose that in the $r$th iteration, the current guess Hamiltonian $\vG^{(r)}=\sum_{\vP} g^{(r)}_{\vP}\vP$ satisfies:
\begin{equation}
  \max_{\vP\in\mathcal P_{\leq \sk}}
  |g_{\vP}^{(r)}-h_{\vP}|\leq\eta_r, \quad \text{and} \quad \forall\vP\in \Gamma_{\vG^{(r)}}: \quad |h_{\vP}|\geq \frac{7}{16}\cdot  \eta_r
  \label{eq:structure-round-invariant}.
\end{equation}
Then, if every estimate in \cref{alg:structure-learner} is within the error tolerance $c_{\mathsf{meas}}\cdot \eta_r$, the output Hamiltonian  $\vG^{(r+1)}=\sum_{\vP} g^{(r+1)}_{\vP}\vP$ satisfies $\Gamma_{\vG^{(r)}}\subseteq \Gamma_{\vG^{(r+1)}} \subseteq \Gamma_{\vH}$, and moreover:
\begin{equation}
  \max_{\vP\in\mathcal P_{\leq \sk}}
  |g_{\vP}^{(r+1)}-h_{\vP}|\leq\eta_{r+1}, \quad \text{and} \quad \forall\vP\in \Gamma_{\vG^{(r+1)}}: \quad |h_{\vP}|\geq \frac{7}{16}\cdot  \eta_{r+1}
 \label{eq:structure-local-contraction}.
\end{equation}
\end{lemma}

The proof is largely based on minor modifications to that of \cref{lemma:gibbs-one-round}, so we refer the reader back to \cref{section:gibbs-alg} for a reminder. The base case of $\eta_0=1, \vG^{(0)}=0$ satisfies the conditions of \cref{lemma:structure-core-soundness}, and thus iteration of \cref{lemma:structure-core-soundness} for $\log \eta^{-1}$ rounds trivially implies the correctness of the structure learning algorithm.

\begin{proof}

[of \cref{lemma:structure-core-soundness}]
Fix the site $i\in [n]$ and for conciseness abbreviate $\eta:=\eta_r$ and $\vG:=\vG^{(r)}$. Let $(\Lambda_i, \vV:=\sum_{\vP\in \Lambda_i} v_{\vP}\vP)$ be the output of the inner loop of \cref{alg:structure-learner}, and let \((\Lambda_i^*, \vU^*:=\sum_{\vP\in \Lambda_i^*} u_{\vP}^*\vP)\) be the ``ideal'' output in that $\Lambda_i^\star:=\CB_{\Gamma_{\vH}}[i,\ell]$ is the neighborhood of $i$ in the interaction graph of $\vH$ and $\vU^*\in \mathsf{Grid}(\Lambda_i^\star)$ is the closest grid point to the ground truth, $\left|
    g_{\vP}+\eta u^\star_{\vP}-h_{\vP}
  \right|
  \leq\kappa\cdot \eta$. \\

\noindent  \textbf{Characterizing the ideal solution $\vU^*$.} Following \cref{lemma:gibbs-one-round} we ``complete'' the ideal candidate $(\Lambda_i^*, \vU^*)$ by defining the Hamiltonian:
\begin{equation}
  \vK^{(\vU^\star)}
  :=\vH+
  \sum_{\vP\in\Lambda_i^\star}
  ( g_{\vP}+\eta u^\star_{\vP}-h_{\vP})\vP,
\end{equation}
which agrees with $\vH$ outside $\Lambda_i^\star$ and with $\vG+\eta \vU^*$ on $\Lambda_i^\star$. We observe that the interaction graph $\Gamma_{\vK^{(\vU^*)}}\subseteq \Gamma_{\vH}$, since $\Gamma_{\vG}\subseteq \Gamma_{\vH}$, and thereby we can invoke, with no modifications, Lipschitz continuity and the Lieb-Robinson bound for $\CL$ as in \cref{lemma:lipschitz} (both corollaries), to conclude
\begin{equation}
    |\mathsf{Test}_i[\vP, \vG+\eta\vU^*; \vrho_{\vH}]| \leq \poly(\beta\sd)\cdot \eta\cdot |\CB[i, \ell]|\cdot (\kappa + e^{-c_1\cdot \ell})
\end{equation}
with $c_1\propto (\beta d)^{-1}$. We note this is simply steps 1, 2 of \cref{lemma:gibbs-one-round} in the ideal case. Here we explicitly suppress the dependence on $\sD, \mathsf{c}_{\sgr}$. \\

\noindent \textbf{Characterizing the minimizing solution $(\Lambda_i, \vV)$.} Next, we ``complete'' the minimizing solution $(\Lambda_i, \vV)$, output by the inner loop of \cref{alg:structure-learner}. In order to do so, we let $\hat{\Gamma}:=\Gamma_{\vH}\cup\Gamma_{\Lambda_i}$ denote the interaction graph containing the terms of $\vH$ and $\Lambda_i$, and denote as $\CB_{\hat{\Gamma}}[i, b]$ the clauses at radius $\leq b$ around $i\in [n]$ in this graph. For a fixed $b\geq \ell$, denote the Hamiltonian:
\begin{equation}
  \vK^{(\vV)}
  :=\vH+
  \sum_{\vP\in\CB_{\hat{\Gamma}}[i, b]}
  ( g_{\vP}+\eta v_{\vP}-h_{\vP})\vP,
\end{equation}
which agrees with $\vH$ outside $\CB_{\hat{\Gamma}}[i, b]$, and with $\vG+\eta \vV$ on $\CB_{\hat{\Gamma}}[i, b]$.

We note that since $\vG+\eta \vV$ is $(\sd, \sD, \mathsf{c_{gr}})$ local, the clause set satisfies the volume bound $|\Lambda_i|\leq \mathsf{c_{gr}}\cdot \ell^\sD$. It then follows from \cref{lemma:connectivity-perturb} that $\hat{\Gamma}$ and thus $\vK^{(\vV)}$ also admit polynomial growth, i.e. is $(\leq 4\sk \sd, \sD, \leq 2\sk\cdot  \mathsf{c_{gr}}^2\cdot \ell^\sD)$ local. This enables us to invoke the Lieb-Robinson bound \cref{lemma:lipschitz} (\cref{lemma:lipschitz-within}), and the fact $(\Lambda_i, \vV)$ is minimal, to write:
\begin{align}
    &\max_{\vP\in\mathcal \CS^{\leq \sk}_{[n]}}
  |\mathsf{Test}_i[\vP,\vK^{(\vV)};\vrho_{\vH}]|
 \\ \leq&\max_{\vP\in\mathcal \CS^{\leq \sk}_{[n]}}
  |\mathsf{Test}_i[\vP,\vG+\eta \vV;\vrho_{\vH}]| + \poly(\beta\cdot \sd)\cdot \ell^\sD\cdot \eta\cdot  e^{-c_1\cdot b} \\
  \leq &\max_{\vP\in\mathcal \CS^{\leq \sk}_{[n]}}
  |\mathsf{Test}_i[\vP,\vG+\eta \vU^*;\vrho_{\vH}]| + \poly(\beta\cdot \sd)\cdot \ell^\sD\eta e^{-c_1\cdot b}+2c_{\mathsf{meas}}\cdot\eta  \\
  \leq & \eta \cdot \poly(\beta\cdot \sd)\cdot (b\ell)^\sD\cdot \big(\kappa+c_{\mathsf{meas}} + e^{-c_1\ell}+ e^{-c_1 b}) := \eta\cdot \theta
\end{align}
where we introduce the notation $\theta$ to suppress notation, and we omit constant pre-factors in $\mathsf{c_{gr}}, \sD$. \\

\noindent \textbf{Computing the Fisher Information.} We now reuse the last part of the proof of
\cref{lemma:gibbs-one-round}.  By the relationship between the stationarity-test and the Fisher information (\cref{rmk:stat-fisher-intro}) and
linearity,
\begin{align}
  \FI_{\vK^{(\vV)},i}(\vrho_{\vH}\Vert\vrho_{\vK^{(\vV)}})
  &=\beta\,
    \mathsf{Test}_i[\vH-\vK^{(\vV)},\vK^{(\vV)};\vrho_{\vH}]\notag\\
  &\leq
  2\beta \eta^2 \cdot |\CB_{\hat{\Gamma}}[i, b]|\cdot 
  \theta.
  \label{eq:structure-fisher-bound}
\end{align}
The convexity of the Fisher information \cref{lemma:fisher-convexity} then applies to $\vK^{(\vV)}$ and $\vH$, since again the union graph is $(\leq 4\sk \sd, \sD, \leq \sk\cdot  \mathsf{c_{gr}}^2\cdot \ell^\sD)$ local by \cref{lemma:connectivity-perturb}. Thus, for every $\Omega>0$,
\begin{equation}
  \max_{\vP:\,i\in\supp(\vP)}
  |g_{\vP}+\eta v_{\vP}-h_{\vP}|
  \leq
  e^{\poly(\beta\cdot d)}\bigg(
    e^\Omega\sqrt{\FI_{\vK^{(\vV)},i}(\vrho_{\vH}\Vert\vrho_{\vK^{(\vV)}})}
    +e^{-c_1\Omega}\cdot \ell^\sD\cdot \eta
  \bigg).
  \label{eq:structure-fisher-inversion}
\end{equation}
where again we omit constant pre-factors in $\mathsf{c_{gr}}, \sD$. \\

\noindent \textit{Choosing constants.} We first pick $\Omega = \poly(\beta\cdot \sd)\log \ell$ to suppress the second term above, then subsequently $b=\ell$, $\kappa=c_{\mathsf{meas}}=e^{-c_1\cdot \ell}$, and finally a sufficiently large $\ell = \poly(\beta\cdot \sd)$ which then ensures $\eqref{eq:structure-fisher-inversion}\leq \eta \cdot (32)^{-1}$.\\

\noindent \textbf{Maintaining the invariant $\Gamma_{\vG^{(r)}}\subseteq \Gamma_{\vG^{(r+1)}}\subseteq \Gamma_{\vH}$.} We consider the cases:
\begin{enumerate}

\item If at the $r$th iteration, $g_{\vP}^{(r)} \neq 0$, then 
\begin{align}
    |g_{\vP}^{(r)}+\eta_r v_{\vP}| &\geq |h_{\vP}| - \eta_r \cdot 32^{-1}  \geq \eta_r\bigg(\frac{7}{16}-32^{-1}\bigg) \geq \eta_r \cdot 4^{-1}.
\end{align}
where we applied \eqref{eq:structure-fisher-inversion}, and the second assumption in \eqref{eq:structure-round-invariant}. By the threshold cutoff \eqref{eq:thresholding}, $g_{\vP}^{(r+1)}\neq 0$ and thus $\Gamma_{\vG^{(r)}}\subseteq \Gamma_{\vG^{(r+1)}}$. 

\item  If after the $r$th round we have $g_{\vP}^{(r+1)} \neq 0$, then 
\begin{equation}
    |h_{\vP}|\geq |g_{\vP}^{(r)}+\eta_r v_{\vP}| - \eta_r \cdot 32^{-1} \geq \eta_r (4^{-1}-32^{-1})\geq \eta_r \frac{7}{32}  = \eta_{r+1} \cdot \frac{7}{16}
\end{equation}
where we applied the guarantee on \eqref{eq:structure-fisher-inversion} and the threshold cutoff \eqref{eq:thresholding}. This implies the second invariant of \eqref{eq:structure-local-contraction}, and in particular $ \Gamma_{\vG^{(r+1)}}\subseteq  \Gamma_{\vH}$.

\item Finally, if $h_{\vP}\neq 0$ but $g^{(r+1)}_{\vP} = 0$, then 
\begin{equation}
        |g_{\vP}^{(r+1)}-h_{\vP}| = |h_{\vP}| \leq |g_{\vP}^{(r)}+\eta_r v_{\vP}| + \eta_r \cdot 32^{-1}  \leq \eta_r(4^{-1}+32^{-1})\leq 2^{-1}\eta_r.
    \end{equation}
    This implies the desired accuracy guarantee in  \eqref{eq:structure-local-contraction}. 
\end{enumerate} 
\end{proof}

\subsection{Sample and time complexity}
\label{section:structure-sample}

It remains to argue the sample complexity of \cref{alg:structure-learner}, which in turn relies on a procedure to measure the outcome of the stationarity test,
\begin{equation}
    |\mathsf{Test}_i
    [\vP,\vG^{(r)}+\eta_r\cdot \vU;\vrho_{\vH}]|,\label{eq:stat-test-struct}
\end{equation}
for each $i\in [n]$, arbitrary local choices of neighborhood $\Lambda\in \mathsf{CC}_{i,\ell}$, any perturbation \(\vU\in\Grid(\Lambda)\) to the coefficients in said neighborhood, and every $\sk$ local Pauli $\vP\in\CS_{[n]}^{\sk}$. We refer the reader to \cref{section:overview} for a brief overview of our approach, and the simpler solutions to this problem with suboptimal sample complexity. The main claim of this subsection is the following lemma:

\begin{lemma}[Parallel measurements of the stationarity test]\label{lemma:parallel-stat-meas}
    Fix $\beta\geq \sd^{-1}$. Recall the context of \cref{alg:structure-learner}, and the choice of algorithm parameters $\ell, \kappa, c_{\mathsf{meas}}$ determined by \cref{lemma:structure-core-soundness}, after the $r$th iteration let $\vG^{(r)}$ be the current guess Hamiltonian accurate to error $\eta_r$.

   Then, there exists an algorithm which with probability $\geq 1-\delta$ outputs estimates for the stationarity test \eqref{eq:stat-test-struct} for every $i\in [n]$, $\Lambda\in \mathsf{CC}_{i,\ell}$, perturbation \(\vU\in\Grid(\Lambda)\), and $\vP\in \CS_{[n]}^{\sk}$, to error $c_{\mathsf{meas}}\eta_r$ using a total of 
   \begin{align}
       \mathsf{N_{samp}} := e^{\poly(\beta\sd)}\cdot O\bigg( \frac{\log^{\sD} \eta_r^{-1}}{\eta_r^2}\cdot \log \frac{n}{\delta} \bigg)
   \end{align}
   samples of $\vrho_{\vH}$. In addition, the classical and quantum time complexity is $n^{\poly(\beta\sd)}\cdot \poly(\mathsf{N_{samp}})$.
\end{lemma}

Before we prove \cref{lemma:parallel-stat-meas}, we present a simple conclusion to the proof of the structure learning result \cref{theorem:structure-learning}.

\begin{proof}

    [Proof of \cref{theorem:structure-learning}] Let $T=\ceil{\log \eta^{-1}}$ denote the iteration count. By invoking \cref{lemma:parallel-stat-meas} on failure probability $\delta / T$, it follows that every estimate in \cref{alg:structure-learner} is accurate to error $\mathsf{c_{meas}}\eta_r$ with probability $1-\delta$. By \cref{lemma:structure-core-soundness}, it then follows that the resulting structure learning algorithm outputs a Hamiltonian $\vG^{(T)}$ which is a subgraph of $\vH$ and coefficient-wise accurate to error $\eta$. 

    The claimed sample complexity over all $T$ rounds then follows from the convergence of the geometric sum in $\eta_r=2^{-r}$ in \eqref{eq:geo-sum} and the bounds in \cref{lemma:parallel-stat-meas}. The resulting time complexity trivially follows from \cref{lemma:parallel-stat-meas} as well. Rescaling $\beta$ and $\vH$ in the high temperature case gives the claimed result (\cref{remark:rescale}).
\end{proof}

Let us now return to \cref{lemma:parallel-stat-meas}.\\

\noindent \textbf{Outline.} The starting point to our approach is the  continuity lemma in \cref{lemma:struct-lip-stat}, which enables us to write \eqref{eq:stat-test-struct} in terms of two components: 

\begin{enumerate}
    \item a ``baseline'' estimate which can be computed just from the \textit{current guess} $\vG^{(r)}$, truncated to a logarithmic radius $\sR^{(r)} = \ell+O(\log \eta_r^{-1})$, akin to \cref{lemma:stationarity-test-single-measurement}:
    \begin{equation}
    |\mathsf{Test}_i
    [\vP,\vG^{(r)}_{\CB[i, \sR^{(r)}]};\vrho_{\vH}]|,\label{eq:stat-test-struct-trunc}
\end{equation}
    \item a perturbative correction (depending on $\vU$) on patches of some constant radius $\mathsf{L}\gg \ell$. This will enable us to compute the expectation value \eqref{eq:stat-test-struct} for each $i\in [n]$, $\Lambda\in \mathsf{CC}_{i,\ell}, \vU\in\Grid(\Lambda)$ simultaneously in parallel, by a simple tomography routine \cite{huang2020predicting} on all (constant-sized) regions of radius (roughly) $\mathsf{L}.$
\end{enumerate}

Let us begin with the relevant continuity lemma. 

\begin{lemma} \label{lemma:struct-lip-stat} Fix $\beta\sd\geq 1$ and $i\in [n]$, $\Lambda\in \mathsf{CC}_{i,\ell}, \vU\in\Grid(\Lambda)$ considered by \cref{alg:structure-learner}. Let $\widehat{\vG}_{\mathsf{L}}:=\vG_{\CB_{\widehat{\Gamma}}[i, \mathsf{L}]}^{(r)}$ denote the truncation of $\vG^{(r)}$ to all the clauses at radius $\leq \mathsf{L}$ in the union $\widehat{\Gamma}:=\Gamma_{\vG^{(r)}}\cup \Lambda$. Then, 
    \begin{align}
        \mathsf{Test}_i
    [\vP,\vG^{(r)}+\eta_r\cdot \vU;\vrho_{\vH}] \approx \mathsf{Test}_i
    [\vP,\vG^{(r)}_{\CB[i, \sR^{(r)}]};\vrho_{\vH}] + \tr\bigg[(\CL_{\widehat{\vG}_{\mathsf{L}}+\eta_r\vU, i}^\dagger-\CL_{\widehat{\vG}_{\mathsf{L}}, i}^\dagger)[\vP]\vrho_{\vH}\bigg] 
    \end{align}
    to error 
    \begin{equation}
        \leq \poly(\beta\cdot\sd, \ell)\cdot \bigg( e^{-c_1\cdot \sR^{(r)}} + \eta_r\cdot e^{-c_1(\mathsf{L}-\ell)}\bigg) \label{eq:struct-lip-stat}
    \end{equation}
    where $c_1\propto (\beta\cdot \sd)^{-1}.$
\end{lemma}

\begin{proof}
We apply \cref{lemma:lipschitz-within} to truncate the stationarity test to radius $\sR^{(r)}$, defined by the distance on the interaction graph of $\vG^{(r)}$. 
\begin{equation}
    \mathsf{Test}_i
    [\vP,\vG^{(r)}_{\CB[i, \sR^{(r)}]};\vrho_{\vH}]\approx \mathsf{Test}_i
    [\vP,\vG^{(r)};\vrho_{\vH}]
\end{equation}
Subsequently, \cref{lemma:lipschitz-within-shell} to truncate the refinement to radius $\mathsf{L}$, defined by the distance on the union of interaction graphs of $\widehat{\Gamma}$, up to error scaling with the perturbation strength $\eta_r\cdot |\CB_{\hat{\Gamma}}[i, \ell]|$ and with $\mathsf{L}$: 
\begin{align}
     &\mathsf{Test}_i
    [\vP,\vG^{(r)}+\eta_r\vU;\vrho_{\vH}]-\mathsf{Test}_i
    [\vP,\vG^{(r)};\vrho_{\vH}] \\ 
    \approx &  \mathsf{Test}_i
    [\vP,\widehat{\vG}_{\mathsf{L}}+\eta_r\vU;\vrho_{\vH}]-\mathsf{Test}_i
    [\vP,\widehat{\vG}_{\mathsf{L}};\vrho_{\vH}]
\end{align}

Finally, we note \cref{alg:structure-learner} only considers perturbations $(\Lambda, \vU)$ where $\widehat{\Gamma}:=\Gamma_{\vG^{(r)}}\cup \Lambda$ is still a $(\sd, \sD, \mathsf{c})$ local Hamiltonian, and thereby admits polynomial growth $|\CB_{\hat{\Gamma}}[i, \ell]|\leq \poly(\ell).$
\end{proof}

Next we discuss how to implement \cref{lemma:struct-lip-stat} algorithmically, and thus prove \cref{lemma:parallel-stat-meas}. We focus on how to compute the baseline estimates, as the refinement step is a simple application of Pauli shadow tomography \cref{theorem:pauli-shadows}. In \cref{alg:struct-meas-baseline} we describe a procedure which combines the coloring argument of \cref{section:H-complexity} with a Pauli shadows step \cite{huang2020predicting}.

\begin{proof}

    [of \cref{lemma:parallel-stat-meas}]
     From \cref{lemma:struct-lip-stat}, given the choices of $\ell=\poly(\beta\cdot \sd), c_{\mathsf{meas}}\leq e^{-\poly(\beta\cdot \sd)}$ as defined in the proof of \cref{lemma:structure-core-soundness}, we make appropriate choices of $\mathsf{L}=\poly(\beta\cdot d)$, $\sR^{(r)} = \poly(\beta\cdot \sd)(1+\log \eta_r^{-1})$ ensuring the error in the continuity lemma does not exceed:
     \begin{align}
         \eqref{eq:struct-lip-stat} \leq \eta_r \cdot e^{-\poly(\beta\cdot \sd)}\leq \frac{1}{100} \cdot \eta_r \cdot c_{\mathsf{meas}}
     \end{align}
     Next we turn our attention to computing said quantities. \\
     
     \noindent \textbf{The refinement step.} Following the outline above, this step considers all the quantities of the form:
     \begin{align}
    \mathsf{Test}_i
    [\vP,\vG^{(r)}_{\CB_{\hat{\Gamma}}[i, \mathsf{L}]}+\eta_r\vU;\vrho_{\vH}] - \mathsf{Test}_i
    [\vP,\vG^{(r)}_{\CB_{\hat{\Gamma}}[i, \mathsf{L}]};\vrho_{\vH}], 
\end{align}
     We note that \cref{alg:structure-learner} considers only neighborhoods $\Lambda\in \mathsf{CC}_{i, \ell}$ such that the union $\widehat{\Gamma}:=\Gamma_{\vG^{(r)}}\cup \Lambda$ admits polynomial growth, namely the same geometry as a $(\sd, \sD, \mathsf{c}_{\mathsf{sgr}})$ Hamiltonian. Consequently, the resulting collection of observables is of total count:
     \begin{align}
         n\cdot |\CS_{[n]}^{\sk}|\cdot  \max_i|\mathsf{CC}_{i, \ell}|\cdot \max_{\Lambda}|\mathsf{Grid}[\Lambda]| \leq n^{\sk+1}\cdot n^{\poly(\ell)}\cdot e^{\poly(\beta\sd, \ell)} = n^{\poly(\beta\sd)},
     \end{align}
     and each acts on $\leq \sk+ \max_i\CB_{\widehat{\Gamma}}[i, \mathsf{L}]\leq \poly(\mathsf{L})$ qubits at a time. It then follows verbatim by \cref{theorem:pauli-shadows} that with probability $1-\delta/2$, said expectation values can be estimated to error $\frac{1}{100} \cdot \eta_r \cdot c_{\mathsf{meas}}$ using the asserted sample complexity. \\

     \noindent \textbf{The baseline estimates.} We refer the reader to \cref{alg:struct-meas-baseline} for a description of the estimation algorithm. The correctness of the parallel short-time estimation follows akin to the coloring argument of \cref{section:gibbs-alg}. Indeed, it follows from \cref{lemma:stationarity-test-single-measurement} (in particular, the application of Duhamel's formula in \eqref{eq:duhamel-estimation-efficient}) that the finite-time difference estimates are accurate to error:
     \begin{align}
        \bigg| \frac{1}{t}\bigg(\tr[\vP\cdot  \Phi_i[\vrho_{\vH}]]-\tr[\vP\cdot  \vrho_{\vH}]\bigg) - \mathsf{Test}_i
    [\vP,\vG^{(r)};\vrho_{\vH}] \bigg| \leq \frac{1}{16}c_{\mathsf{meas}}\eta_r + \eta_r\cdot t\cdot \poly(\beta\sd) \label{eq:alg-baseline-error}
     \end{align}
     so long as we pick $\sR \geq \poly(\beta\sd) \log (c_{\mathsf{meas}}\eta_r)^{-1}$. Consequently, with a suitable choice of short-time evolution $t = e^{-\poly(\beta\sd)}$ it follows that the error in $\eqref{eq:alg-baseline-error} \leq \frac{1}{8} c_{\mathsf{meas}}\eta_r$. 
     
     Next we address the sample complexity of this step. \cref{theorem:pauli-shadows-after-channels} implies that we can estimate every desired Pauli expectation value, within one of the $\mathsf{s} = \min(n, \poly(\beta\sd)\log^\sD\eta^{-1})$ colors to the intended accuracy of $\frac{1}{8}t c_{\mathsf{meas}}\cdot \eta_r$ with probability of failure $\delta$ with:
     \begin{align}
         e^{\poly(\beta\sd)}\cdot O\bigg( \frac{1}{\eta_r^2}\cdot \log \frac{n \mathsf{s}}{\delta} \bigg) =  e^{\poly(\beta\sd)}\cdot O\bigg( \frac{1}{\eta_r^2}\cdot \log \frac{n}{\delta} \bigg)
     \end{align}
     appropriately rescaling $\delta$ to union bound over the $\mathsf{s}$ colors, and summing over their samples gives the advertised bound.\\
     
     \noindent \textbf{Time complexity.} This follows simply from the implementation costs described in \cref{lemma:stationarity-test-single-measurement}, that of enumerating all the neighborhoods, and that of  computing all the measurements, resulting in $n^{\poly(\beta)}\cdot \poly(\mathsf{N_{samp}})$.
\end{proof}

\begin{algorithm}[H]
\caption{Parallel measurements of the ``baseline'' tests.}
\label{alg:struct-meas-baseline}
\begin{algorithmic}[1]
\Require Current estimate \(\vG\), error \(\eta\), samples of \(\rhoH\), integer $\mathsf{R}$, and time $t>0$.

\State Compute estimates $\hat{\mathsf{T}}[\vP; \vrho_{\vH}]$ for  $\tr[\vP \vrho_{\vH}]$ of every $\sk$ local Pauli operator $\vP\in \CS_{[n]}^{\sk}$ to error $8^{-1}\cdot t\cdot c_{\mathsf{meas}}\cdot \eta$ using \cref{theorem:pauli-shadows}.

\State Partition $[n]$ into $\mathsf{s}:=O(\sR^{\sD})$ disjoint subsets of sites $\mathsf{C}_1, \cdots, \mathsf{C_{s}}\subseteq [n]$, such that the balls $\sB_{\vG}[i_1, \sR]$, $\sB_{\vG}[i_2, \sR]\subseteq [n]$ for any two $i_1, i_2$ of the same color do not intersect.

\For{Each color $\mathsf{C}$ in the partition}

\For{Each $i\in \mathsf{C}$} define the map $\Phi_i[\cdot ]$, truncating the evolution $\CL_{\vG_{\CB[i, \sR]}, i}$  to radius $\sR$:
\begin{equation}
    \Phi_i[\cdot ]:=\exp\bigg[t\cdot \CL_{\vG_{\CB[i, \sR]}, i}\bigg][\cdot ],
\end{equation}
\EndFor

\State  Compute estimates for $\tr[\vP\cdot  \Phi_i[\vrho_{\vH}]]$,  for each $ \vP\in\CS^{\sk}_{[n]}, i\in \mathsf{C}$ to error $8^{-1}\cdot t\cdot c_{\mathsf{meas}}\cdot \eta$ using \cref{theorem:pauli-shadows-after-channels}. Denote the outcome as  $\hat{\mathsf{T}}[\vP; \Phi_i[\vrho_{\vH}]]$.
\EndFor

\State For each $i\in [n]$, $\vP\in \CS_{[n]}^{\sk}$ output
\begin{equation}
    \frac{1}{t}\cdot \bigg(\hat{\mathsf{T}}[\vP; \Phi_i[\vrho_{\vH}]]- \hat{\mathsf{T}}[\vP; \vrho_{\vH}]\bigg)
\end{equation}
as an estimate for $\tr[\CL_{\vG_{\CB[i, \sR]}, i}^\dagger[\vP]\vrho_{\vH}]$.
\end{algorithmic}
\end{algorithm}

\newpage

\section{Learning Hamiltonian parameters from their metastable states}
\label{section:metastable-learning}

In this section, we prove that the parameter learning algorithm devised in \cref{section:gibbs-alg} succeeds even when input samples from metastable states (as opposed to the Gibbs state). We operate under the local formulation of metastability presented in \cref{def:intro-local-metastability}, and restated here.

\begin{definition}[Metastability as approximate stationarity {\cite{BCV25}}]
\label{def:local-metastability}
A state \(\vsigma\) is said to be \(\epsilon\)-locally metastable with respect to a Hamiltonian
\(\vH\) if, for every site \(i\in[n]\),
\begin{equation}
    \left\|\CL_{\vH,i}[\vsigma]\right\|_1\leq\epsilon.
    \label{eq:local-metastability}
\end{equation}
\end{definition}

\begin{remark}\label{rmk:metastab-complete}
\cref{def:local-metastability} entails an approximate version of the completeness condition for the stationarity test. For any observable $\vO$,
    \begin{equation}
        |\mathsf{Test}_{i}[\vO,\vH;\sig]| \leq \epsilon\cdot \|\vO\|. \tag{Completeness}
    \end{equation}
\end{remark}

In other words, the ground truth Hamiltonian $\vH$ approximately passes the stationarity test. The following \cref{thm:end-to-end-metastable} is the main result of this section, which learns the coefficients of the Hamiltonian up to an accuracy threshold ``plateau'' which depends on $\epsilon$. As before, we consider $(\sd, \sD, \scgr)$, $\sk$ local qubit Hamiltonians and treat said parameters as fixed constants, in order to study the asymptotics in $n, \eta, \beta, \delta$.

\begin{theorem}[Learning from metastable states]
\label{thm:end-to-end-metastable} 
    Fix $\beta\in \BR^+, \epsilon\in (0, 1)$, an $n$ site $(\sd, \sD, \scgr_{\sgr})$ Hamiltonian $\vH:=\sum_{\gamma\in \Gamma} h_\gamma \vP_\gamma$  on a fixed set of interactions $\Gamma$ with $|h_\gamma|\leq 1$. Let $\vsigma$ be an $\epsilon-$locally metastable state w.r.t. $\vH$. Fix $\delta\in (0, 1)$ and $\eta >\eta_{\mathsf{thr}}:=\scgr_\beta\cdot \epsilon$ where $\scgr_\beta\leq \beta^{-1}\cdot e^{\poly(\beta)}$. 
    
    Then, with probability $\geq 1-\delta$, the learning algorithm described in \cref{section:gibbs-alg} returns a Hamiltonian $\vG=\sum_{\gamma\in \Gamma} g_\gamma \vP_\gamma$ satisfying $\max_{\gamma\in \Gamma}|h_\gamma-g_\gamma|\leq \eta$ using at most 
    \begin{equation}
      \mathsf{N_{samp}}:= \frac{ e^{\poly(\beta)}}{\beta^2}\cdot O\bigg(\eta^{-2}\log \frac{n}{\delta\eta}\cdot \log^{\sD}\frac{1}{\eta}\bigg)
    \end{equation}
    copies of $\vsigma$ and $\poly(n,  \mathsf{N_{samp}})$ quantum and classical processing time. 
\end{theorem}

We use the exact same algorithm as described in \cref{alg:gibbs-learner,alg:one-round}; however, with minor modifications to the choice of parameters as discussed in \cref{section:metastab-correct} which do not change the asymptotic sample or time complexity. As we discuss shortly, the correctness argument (\cref{lemma:metastab-correctness}) is based on revisiting the approximate convexity property for the quantum Fisher information, and deriving an analogous statement directly for the stationarity test (evaluated on the metastable state). The resulting error will depend on the extent with which $\vsigma$ violates a certain detailed-balance condition \cite{BCV25}, which we introduce in the sequence (\cref{section:adb}). Subsequently, we argue the correctness of our iterative algorithm in \cref{section:metastab-correct}, and conclude with the sample and time complexity in \cref{section:metastab-sample-and-time}

\subsection{The approximate detailed-balance condition}
\label{section:adb}

We dedicate this section to collecting the relevant preliminaries on metastable states, later required to prove correctness. We will make extensive use of the following definition \cite{BCV25} of a robust version of KMS detailed-balance condition, and here we derive several additional implications. Throughout this subsection we recall that the Gaussian energy filter width is fixed to
\(\sigma=\beta^{-1}\).

\begin{definition}[Approximate detailed balance
{\cite[Definition~C.1]{BCV25}}]
\label{def:metastable-adb}
For a full-rank state \(\vsigma\), a Hamiltonian \(\vG\), and a
Hermitian jump \(\vA\), define
\begin{equation}
 \ADB_{\vG,\vA}[\vsigma]
 :=\int_{\BR}\!\rd\omega\int_{\BR}\!\rd t\,
 \gamma_{\mathsf M}(\omega)g(t)
 \left\|
  \widehat{\vA}_{\vG}(\omega,t)\sqrt{\vsigma}
  -\sqrt{\vsigma}\,\vrho_{\vG}^{-1/2}
   \widehat{\vA}_{\vG}(\omega,t)\vrho_{\vG}^{1/2}
 \right\|_2^2 .
 \label{eq:metastable-adb-definition}
\end{equation}
\end{definition}

 A central innovation in \cite{BCV25} was to relate this condition to the gradients of the log likelihood.

\begin{lemma}[ADB and entropy gradients
{\cite[Lemma~C.2]{BCV25}}]
\label{lem:metastable-adb-entropy-gradient}
In the context of \cref{def:metastable-adb}:
\begin{align}
 \ADB_{\vG,\vA}[\vsigma]
 &=\int_{-1/2}^{1/2}\!\rd s
   \int_{\BR}\!\rd\omega\int_{\BR}\!\rd t\,
   h_s(\omega)g_s^{\ADB}(t)
   \left\|[\widehat{\vA}_{\vG}(\omega,t),\log\vsigma-\log\vrho_{\vG}]
   \right\|_{\vsigma,s}^{2},
 \label{eq:metastable-adb-entropy-gradient}\\
 g_s^{\ADB}(t)
 &:=\frac12\int_{-(1/2-|s|)}^{1/2-|s|}g_{s'}(t)\,\rd s'
 =\frac{1}{2\pi\beta}
 \log\!\left(
 \frac{\cosh(2\pi t/\beta)+\cos(\pi s)}
      {\cosh(2\pi t/\beta)-\cos(\pi s)}
 \right),
 \label{eq:metastable-adb-kernel}
\end{align}
where \(g_s, h_s\) are the time and frequency filter functions from the Fisher information
\eqref{eq:gs-explicit}.
\end{lemma}

Under this formulation, the approximate detailed balance condition or `error' can be related to the Fisher information:

\begin{lemma}[ADB and the Fisher Information
{\cite[Theorem~C.1]{BCV25}}]
\label{lem:metastable-FI-to-ADB}
In the context of \cref{def:metastable-adb}, let $\vA^a:=\vA$ satisfy $\|\vA^a\|\leq 1$. Then, there exists a universal constant $c>0$ s.t. 
\begin{equation}
  \ADB_{\vG,\vA^a}[\vsigma]
  \leq
c\cdot  \FI_{\vG,a}(\vsigma\Vert\vrho_{\vG})
  \left(
    1+
    \log_+
    \frac{c\cdot \|\log \vsigma - \log \vrho_{\vG}\|^2}
         {\FI_{\vG,a}(\vsigma\Vert\vrho_{\vG})}
  \right),
  \label{eq:metastable-FI-to-ADB}
\end{equation}
where \(\log_+x=\max\{0,\log x\}\), and the right-hand side is zero
when the Fisher information vanishes.
\end{lemma}

We dedicate the remainder of this subsection to describing two key ideas. First, in \cref{section:multi-to-single-adb} we describe how the approximate detailed-balance condition under single-site updates can be compared to the same condition under few-site (e.g. 2-site) updates, up to additive corrections. We remark on the passing resemblance to the Dirichlet form comparison inequalities in \cite{chen2025GibbsMarkov, BC26, BCV25, B26all}. Then, in \cref{section:kms-to-F} we describe how to relate the KMS norm of local observables, weighted under the metastable state, to the associated Frobenius norm, using the ADB condition. This is in analogue to the ``faithfulness" condition for Gibbs states described in \cref{lemma:faithfulness}.

\subsubsection{Multi-site to single-site approximate detailed balance}
\label{section:multi-to-single-adb}

The next ingredient we require are two brief lemmas on how to relate multi-site variants of the approximate detailed-balance condition. Namely, suppose $\vsigma$ is approximately detailed-balanced, w.r.t single-site jump operators with certain error. Does this entail $\vsigma$ is also approximately detailed-balanced under slightly larger jump operators?

\begin{lemma}[Multi-site to single-site ADB]
\label{lem:metastable-multisite-to-single-site-adb}
Let $\beta\sd\geq 1.$ Let $\vG$ be a local Hamiltonian on a bounded, degree $\sd$ graph, $\vsigma$ a full rank state, and $\vP = \otimes_i^{|\sR|}\vA_i$ a Pauli operator supported on a region $\sR$ of qubits. Then, there exists a constant $c\leq e^{\poly(\beta\sd)}$ s.t. for every $\Omega\geq 1$,
\begin{equation}
\ADB_{\vG,\vP}[\vsigma]
 \leq c^{|\sR|}\cdot \bigg(e^{4\beta \Omega}\max_{1\leq i\leq |\sR|} \ADB_{\vG,\vA_i}[\vsigma]
      +e^{-\Omega/8\sd}\bigg).
 \label{eq:differentiated-gradient-multisite-adb}
\end{equation}

\end{lemma}

We remark to the interested reader that the proof is akin to the computation in \cite[Section D.4]{BCV25}, in the context of Markov properties for metastable states.

\begin{proof}
To curtail notation we denote, for an arbitrary operator $\vB$,
\begin{equation}
 \Delta(\vB)
 :=\vB\sqrt{\vsigma}-\sqrt{\vsigma}\,\vrho_{\vG}^{-1/2}\vB\vrho_{\vG}^{1/2},
 \qquad
 \|\vB\|_{\ADB}:=\sqrt{\ADB_{\vG,\vB}[\vsigma]}.
 \label{eq:differentiated-gradient-adb-defect}
\end{equation}
and suppress the $\vG$ subscripts. We note two apriori bounds:
\begin{enumerate}
    \item a scaled triangle inequality $\|\vA+\vB\|_{\ADB}\leq \sqrt{2}(\|\vA\|_{\ADB}+\|\vB\|_{\ADB})$, by expansion and the fact $\ADB$ depends quadratically in the jump operator. 
    \item an upper bound by the operator norm: $\|\vB\|_{\ADB}\leq 4\|\vB\|.$
\end{enumerate}
The latter bound follows since:
\begin{align}
    \|\vB\|_{\ADB}^2 &\leq 2\cdot \iint g(t)\gamma(\omega)\bigg( \|\widehat{\vB}(\omega, t)\sqrt{\vsigma}\|_2^2+\|\sqrt{\vsigma}\vrho_{\vG}^{-1/2}\widehat{\vB}(\omega, t)\vrho_{\vG}^{1/2}\|_2^2\bigg) \rd \omega \rd t
\end{align}
We note the second factor above admits the upper bound:
\begin{align}
     &\iint g(t)\gamma(\omega)e^{\beta\omega+\beta^2\sigma^2/2}\cdot  \tr\bigg[\vsigma\widehat{\vB}(\omega+\beta\sigma^2,t)
 \widehat{\vB}(\omega+\beta\sigma^2,t)^\dagger\bigg] \rd\omega \rd t \leq \|\vB\|^2
\end{align}
where we used $\gamma_{\mathsf M}(\nu)
e^{\beta\nu+\beta^2\sigma^2/2}\leq1$, the frequency shift identity under modular conjugation \cref{lem:bounds_imaginary_conjugation}, and Parseval's \cref{lem:OParseval}. The first factor is analogous. 

To proceed with the proof, we follow the approach in \cite[Section D.4]{BCV25}, wherein the individual single-site Pauli operators in $\vP$ are each separately expanded into their OFTs using \cref{lem:sumoverenergies}, and subsequently jointly truncated down to a simplex parametrized by $\Omega>0$:
\begin{equation}
\vP_{\leq\Omega}
 :=c_{\mathrm{rec}}^{|\sR|}
   \int_{\mathcal U_\Omega}
   \prod_{j=1}^{|\sR|}
   \widehat{\vA_j}(\omega_j)
   \rd\boldsymbol\omega,
 \qquad
 \mathcal U_\Omega
 :=\left\{\boldsymbol\omega\in\BR^{|\sR|}:
          \sum_{j=1}^{|\sR|}|\omega_j|\leq\Omega\right\},
 \label{eq:differentiated-gradient-total-frequency-simplex}
\end{equation}
and where we define $ c_{\mathrm{rec}}:=(2\sigma\sqrt{2\pi})^{-1/2}=(2\beta^{-1}\sqrt{2\pi})^{-1/2}$. Applying \cref{lem:bounds_oft,lem:norm_bound_pauli} to each
one-site factor gives (exactly vertabim to \cite[Lemma D.11]{BCV25})
\begin{equation}
 \|\vP-\vP_{\leq\Omega}\|
 \leq e^{c_{\beta, d}'|\sR|-\Omega/8d}\quad \Rightarrow \quad \|\vP-\vP_{\leq\Omega}\|_{\ADB} \leq e^{c_{\beta,\sd}|\sR|-\Omega/8\sd}
 \label{eq:differentiated-gradient-word-high-frequency}
\end{equation}
where $c_{\beta, \sd}\leq \poly(\beta\sd)$.\\

\noindent \textbf{Low-frequency components.} It remains to understand $\|\vP_{\leq\Omega}\|_{\ADB}$. For this purpose, we begin by re-writing the ADB error $ \|\vB\|_{\ADB}$ in terms of time evolutions of $\vB$:
\begin{align}
     \|\vB\|_{\ADB}^2&= \frac{1}{2\pi}\iint g(t) \gamma(\omega)\bigg\| \int_{\BR}e^{-\mathrm{i}\omega r}f_{\beta^{-1}}(r)
  \Delta\!\left(\vB_{\vG}(t+r)\right)\rd r\bigg\|^2_2\rd t\rd \omega \\
  &\leq \iint g(t)|f_{\beta^{-1}}(r)|^2 \cdot \bigg\|
  \Delta\!\left(\vB_{\vG}(t+r)\right)\bigg\|^2_2\rd t\rd r\\
  &\leq \int  \bigg\|\Delta\!\left(\vB_{\vG}(u)\right)\bigg\|^2_2 \cdot \int g(t)f_{\beta^{-1}}(u-t)^2 \rd t \rd u  \\ &\leq 2e^{\pi^2/2} \int g(u)  \big\|\Delta\!\left(\vB_{\vG}(u)\right)\big\|^2_2\rd u\label{eq:adb-to-time}
\end{align}
where we computed the convolution of the explicit filter functions $g, f_{\beta^{-1}}$.

Our goal next is to understand the above with $\vB:=\vP_{\leq\Omega}$ into terms which depend only on the single-site Pauli's. For this purpose, we nest the following product identity:
\begin{equation}
 \Delta(\vB_1\vB_2)
 =\vB_1\Delta(\vB_2)+\Delta(\vB_1)\vrho_{\vG}^{-1/2}\vB_2\vrho_{\vG}^{1/2}.
 \label{eq:differentiated-gradient-adb-product-rule}
\end{equation}
And further (to compress notation) we introduce for $j\in \sR$ the function:
\begin{equation}
 E_j(\omega_j)^2
 :=\int_{\BR}g(u)
 \left\|\Delta\!\left(
   \widehat{\vA_j}(\omega_j,u)
 \right)\right\|_2^2\rd u \quad \Rightarrow \quad \int_{\mathcal U_\Omega}E_j(\omega_j)\rd\boldsymbol\omega
 \leq e^{c_{\beta, d}''|\sR|+\Omega\beta}
 \sqrt{\ADB_{\vG,\vA_j}[\vsigma]},
 \label{eq:differentiated-gradient-word-cauchy-schwarz}
 \end{equation}
 where we use the bound on the simplex volume $\mathsf{vol}(\CU_{\Omega})\leq (2\Omega)^{|\sR|}/|\sR|!\leq (2\beta^{-1})^{|\sR|}e^{\Omega\beta}$. By applying the product rule \eqref{eq:differentiated-gradient-adb-product-rule} repeatedly to the time-integral expression in \eqref{eq:adb-to-time}, followed by the Cauchy-Schwarz inequality, and the norm bounds on the OFT and its conjugation \cref{lem:bounds_oft,lem:bounds_imaginary_conjugation} gives:
\begin{equation}
 \|\vP_{\leq\Omega}\|_{\ADB}
 \leq e^{d_{\beta, d}|\sR|+2\Omega\beta}
 \sum_{j=1}^{|\sR|}\sqrt{\ADB_{\vG,\vA_j}[\vsigma]}.
 \label{eq:differentiated-gradient-word-minkowski}
\end{equation}
The triangle inequality in $\|\cdot \|_{\ADB}$ from (1) above concludes the proof. 
\end{proof}

\begin{claim}
\label{claim:multi-site-sint}
Let $\beta\sd\geq 1.$
Let $\vG$ be a local Hamiltonian on a degree $d$ graph, $\vsigma$ a full rank state, and $\vP$ a Pauli operator supported on a region $\sR$ of qubits. Then, for any $c\geq 1$:
\begin{align}
\int_{-1/2}^{1/2}
   \|[\vP,\log \vsigma+\beta\vG]\|_{\vsigma,s}^2\rd s \leq  e^{c\cdot \poly(\beta\sd)\cdot |\sR|}
 \ADB_{\vG,\vP}[\vsigma]
 +e^{-c\cdot |\sR|}\cdot \|\log \vsigma+\beta\vG\|^2.
\end{align}
\end{claim}

\begin{proof}

The challenge in relating the above to the ADB filter function is a divergence near $s\in \pm \frac{1}{2}$. Accordingly, fix $0<\delta\leq1/4$ and denote $I_\delta:=(-\frac{1}{2}+\delta, \frac{1}{2}-\delta)$. Near the endpoints we use \cref{lemma:kms-to-s}, and compute the apriori bound:
\begin{equation}
 \int_{[-1/2,1/2]\setminus I_\delta}
 \|[\vP,\log \vsigma+\beta\vG]\|_{\vsigma,s}^2\rd s
 \leq 8\delta\cdot  \|\log \vsigma+\beta\vG\|^2.
 \label{eq:differentiated-gradient-endpoint-bound}
\end{equation}
In turn, far from the endpoints ($s\in I_\delta$), we use the Gaussian convolution identity
\cref{lem:OFT_Gaussian_expansion}, followed by
\cref{lem:sumoverenergies}, to write
\begin{equation}
 \vP
 =\frac{1}{\sqrt{2\pi}}\iint_{\BR^2}
 e^{-\mathrm{i}\nu q}f_{\sigma}(q)
 \widehat{\vP}_{\vG}(\nu,q)\rd q\rd\nu.
 \label{eq:differentiated-gradient-gaussian-reconstruction}
\end{equation}
To proceed we introduce a frequency filter parameter $\Omega>0$.   \cref{lem:bounds_oft} and  \cref{lem:norm_bound_pauli} imply the above can be truncated to low-frequency components:
\begin{equation}
 \int_{|\nu|>\Omega}\int_{\BR}
 |f_\sigma(q)|\,\|\widehat{\vP}_{\vG}(\nu,q)\|
 \rd q\rd\nu
 \leq e^{c_{\beta, d}|\sR|-\Omega/8\sd}.
 \label{eq:differentiated-gradient-pauli-tail-prebound}
\end{equation}
 In turn, the low-frequency components can be expressed in terms of a filter-function akin to that in the approximate detailed-balanced condition. To suppress notation, here we denote $\vX:=\log \vsigma+\beta\vG$.
\begin{align}
 \|[\vP,\vX]\|_{\vsigma,s}^2
 &\leq
2\bigg| \iint_{|\nu|\leq \Omega} 
f_{\sigma}(q)
 \|[\widehat{\vP}_{\vG}(\nu,q),\vX]
 \|_{\vsigma,s}\rd q\rd\nu\bigg|^2
+e^{c_{\beta, \sd}|\sR|-\Omega/8\sd}\cdot \|\vX\|^2
 \label{eq:differentiated-gradient-bare-to-filtered} \\
& \leq 200\cdot e^{\beta\Omega} \iint_{\BR^2}
 h_s(\nu)g(q)
 \|[\widehat{\vP}_{\vG}(\nu,q),\vX]
 \|_{\vsigma,s}^2\rd q\rd\nu+e^{c_{\beta, \sd}|\sR|-\Omega/8\sd}\cdot \|\vX\|^2
\end{align}
Finally, the inequality \(\log((1+x)/(1-x))\geq2x\) implies that for $s\in I_\delta$, the filter function (\cref{lem:metastable-adb-entropy-gradient}) satisfies the bound $g_s^{\ADB}(q)
 \geq\frac{2\delta}{\pi}g(q)$ and therefore 
\begin{equation}
 \int_{I_\delta}
 \|[\vP,\log \vsigma+\beta\vG]\|_{\vsigma,s}^2\rd s
 \leq
 10^3\delta^{-1}e^{\beta\Omega}\ADB_{\vG,\vP}[\vsigma]
 +e^{c_{\beta, d}|\sR|-\Omega/8\sd}\cdot \|\log \vsigma+\beta\vG\|^2
 \label{eq:differentiated-gradient-pauli-adb-interior}
\end{equation}
Combining this with \eqref{eq:differentiated-gradient-endpoint-bound}
and choosing
$\delta=e^{-\Theta_{\beta, \sd}(|\sR|)}, \Omega = \Theta_{\beta, \sd}(|\sR|)$ gives the advertised bound.
\end{proof}

\subsubsection{Locally relating the KMS and Frobenius norms}
\label{section:kms-to-F}

Our next goal is to be able to relate the KMS and Frobenius norms of local operators, when defined over the KMS norm of a metastable state $\vsigma$. These are analogs of the ``faithfulness'' condition in \cref{lemma:faithfulness} from \cite{chen2025learning}. For this purpose, we rely on the following statement \cite[Lemma~D.5]{BCV25} that the norms of local operators $\vX$ are ``protected'' even if rotated by local Pauli's:

\begin{lemma}[Exposing the KMS norm
{\cite[Lemma~D.5]{BCV25}}]
\label{lem:metastable-exposing-kms} Let $\beta\sd\geq 1$. Let $\vG$ be a Hamiltonian of degree $\sd$, $\vsigma$ a full rank state, $\|\vX\|\leq 1$ an arbitrary operator and $\vP, \vQ$ Pauli operators supported on a region $\sR$ of qubits.

Then, there exists constants $1\geq \alpha_1\geq 1/\poly(\beta\sd)$ and $\alpha_2\leq e^{\poly(\beta\sd)}$ such that:
\begin{equation}
  \left|\Tr[\vP\vsigma\vQ\vX]\right|
  \le
  \alpha_2^{|\sR|}
  \left(
    \|\vX\|_{\vsigma}
    +\max_{a\in \CS_{\sR}^1}\sqrt{\mathsf{ADB}_{\vG, a}[\vsigma]}
  \right)^{\alpha_1}.
  \label{eq:metastable-exposing-kms}
\end{equation}
\end{lemma}

We are now in a position to state the last new idea of this subsection. 

\begin{lemma}[Locally relating the $\vrho$ and Frobenius norms]
\label{lem:metastab-faithfulness} In the context of \cref{lem:metastable-exposing-kms}, assume an operator $\vY$ is supported on the patch $\sR$ of qubits. Then, there exists a constant $c\leq e^{\poly(\beta\sd)}$ such that:
\begin{align}
   c^{-\sq|\sR|}\|\vY\|_{\vtau_{\sR}}
  \le
  \|\vY\|_{\vsigma}
  +\max_{a\in \CS_{\sR}^1}\sqrt{\mathsf{ADB}_{\vG, a}[\vsigma]}\cdot \|\vY\|.
\end{align}
where $\vtau_\sR$ is the maximally mixed state on $\sR$.
\end{lemma}

\begin{proof}
The observation is that by the Pauli twirl, 
\begin{equation}
  \sum_{\vP,\vQ}
  \left|
    \Tr[(\vP\vsigma_{\sR}\vQ)^\dagger\vY]
  \right|^2
  =
  2^{3\sq\cdot |\sR|}\cdot \Tr(\vsigma_{\sR}^2)\cdot \|\vY\|_{\vtau_\sR}^2.
\end{equation}
Since the purity of the RDM satisfies $\Tr(\vsigma_{\sR}^2)\geq 2^{-\sq|\sR|}$, there exists a pair of Pauli's $\vP, \vQ$ on $\sR$ such that:
\begin{equation}
   2^{-\sq|\sR|}\|\vY\|_{\vtau_\sR}\leq   |\Tr[\vP\vsigma_{\sR}\vQ\vY]|
  \leq \|\vY\|\cdot  \alpha_2^{|\sR|}
  \left(
    \frac{\|\vY\|_{\vsigma}}{\|\vY\|}
    +\max_{a\in \CS_{\sR}^1}\sqrt{\mathsf{ADB}_{\vG, a}[\vsigma]}
  \right)^{\alpha_1}.
  \end{equation}
  where we simply applied \cref{lem:metastable-exposing-kms} to
\(\vX=\vY/\|\vY\|\). Taking the $\alpha_1$ root on both sides and the chain $\|\vY\|_{\vtau_\sR}\leq \|\vY\|\leq 2^{\sq|\sR|}\cdot \|\vY\|_{\vtau_\sR}$ then gives the advertised statement. 
\end{proof}

\subsection{Correctness}
\label{section:metastab-correct}

We next proceed with the correctness argument. We refer the reader back to \cref{section:gibbs-alg} for a description of the parameters in \cref{alg:gibbs-learner} and \cref{alg:one-round}.

\begin{lemma}[Correctness of a single iteration]
\label{lemma:metastab-correctness}
Fix $\beta>\sd^{-1}$. There exists choices of
\(\ell,\kappa,c_{\mathrm{meas}}, c_{\mathrm{noise}} >0\), and $R^{(r)}:=2\ell+\poly(\beta\sd)\cdot \log \frac{1}{\eta_r}$ as a function of $\beta$ (and the Hamiltonian parameters) s.t. if $\vsigma$ is $\epsilon$ locally metastable w.r.t. $\vH$, every estimate in \cref{alg:one-round} is within the error tolerance $c_{\mathrm{meas}}\cdot \eta_r$, and the current guess Hamiltonian $\vG^{(r)}:=\sum g_\gamma^{(r)}\vP_\gamma$ satisfies:
\begin{equation}
  \max_{\gamma\in\Gamma}|g_\gamma^{(r)}-h_\gamma|
  \leq\eta_r
  \qquad
  \epsilon\leq c_{\mathrm{noise}}\cdot \eta_r.
  \label{eq:metastable-round-invariant}
\end{equation}
Then the output $\vG^{(r+1)}:=\sum g_\gamma^{(r+1)}\vP_\gamma$ of the $r$th iteration satisfies:
\begin{equation}
  \max_{\gamma\in\Gamma}|g_\gamma^{(r+1)}-h_\gamma|
  \leq\eta_r/2.
  \label{eq:metastable-one-round-conclusion}
\end{equation}
\end{lemma}

The crux of the correctness argument will be to establish a convexity property for the stationarity test on metastable states. We do so by understanding its behavior when locally perturbing the underlying Hamiltonian around $\vH$. The new ingredient (relative to the Gibbs state proof in \cref{section:H-correct}) is the following lemma:

\begin{lemma}[Convexity of the stationarity test]
\label{lem:metastable-convexity}
In the context of \cref{lemma:metastab-correctness}, fix $\ell\in \BZ^+$ and assume $\epsilon\leq e^{-\poly(\ell, \beta\cdot \sd)}$. Then, for any $\vU = \sum_{\gamma\in \CB[i, \ell]} u_\gamma\cdot \vP_\gamma$ where $|u_\gamma|\leq \eta \leq 2$:
\begin{align}
  \bigg|\mathsf{Test}_{i}[\vU,\vH+\vU;\vsigma] -  \mathsf{Test}_{i}[\vU,\vH;\vsigma]\bigg|
    \geq\,
  e^{-\poly(\beta\sd)}
  \left(
    \max_{\gamma\ni i}|u_\gamma|^2-\frac{1}{32}\cdot \eta^2
  \right).
  \label{eq:metastable-centered-coercivity}
\end{align}
\end{lemma}

We defer the proof of this lemma to the bottom of this subsection (\cref{section:pdf-convexity-stest}). We are now in a position to conclude the correctness argument. 

\begin{proof}

[Proof of \Cref{lemma:metastab-correctness}] Let us fix our attention to a single $i\in [n]$. Following the proof of \cref{lemma:gibbs-one-round}, consider a grid point \(\vU=\sum_{\gamma\in\CB[i,\ell]}u_\gamma\vP_\gamma\in\Grid_i\) and define
\begin{align}
  \vK^{(\vU)}
  &:= \vH+
  \sum_{\gamma\in\CB[i,\ell]}
  \bigl(g_\gamma^{(r)}+\eta_ru_\gamma-h_\gamma\bigr)\vP_\gamma,
  \label{eq:metastable-completed-error}
\end{align}
to be the Hamiltonian which agrees with $\vH$ outside $\CB[i,\ell]$ and with $\vG^{(r)}+\eta_r\vU$ on $\CB[i,\ell]$. \\

\noindent \textbf{Characterizing the minimizing solution $\vV$.} By comparing the grid point $\vV$ output by the algorithm, to the ground truth as in \cref{lemma:gibbs-one-round}, it follows from the Lieb-Robinson bounds in \cref{lemma:lipschitz} that  $\vV$ satisfies:
\begin{align}
 &\max_{\gamma\in\CB[i,\ell]}
  \left|
    \mathsf{Test}_{i}[\vP_\gamma,\vK^{(\vV)};\sig]
  \right|\\ \leq& \max_{\gamma\in\CB[i,\ell]}
  \left|
    \mathsf{Test}_{i}[\vP_\gamma,\vG^{(r)} + \eta_r\vV;\sig]
  \right| +\eta_r\cdot  e^{\poly(\beta\sd)-c_1\cdot \ell} \\
  \leq & \eta_r e^{\poly(\beta\sd)} |\CB[i, \ell]| (c_{\mathrm{meas}}+e^{-c_1\cdot \ell}+\kappa) + \max_{\gamma\in\CB[i,\ell]}\left|
    \mathsf{Test}_{i}[\vP_\gamma,\vH;\vsigma] 
  \right| + \poly(\beta\sd)e^{-c_1 R^{(r)}/2}  \\
  \leq & \eta_r e^{\poly(\beta\sd)} |\CB[i, \ell]| (c_{\mathrm{meas}}+e^{-c_1\cdot \ell}+\kappa)+\epsilon+ \poly(\beta\sd)e^{-c_1 R^{(r)}/2} \label{eq:stat-test-error-meta}
\end{align}
At this point we can make the explicit choice of constants $c_{\mathrm{meas}}=\kappa$ as a function of $\ell$:
\begin{equation}
    c_{\mathrm{meas}}=\kappa=e^{\poly(\beta\sd)-c_1\cdot \ell}.
\end{equation}
as well as $R^{(r)}=2\ell+\poly(\beta\sd) \log \frac{1}{\eta}$. Together with the assumption on the metastability parameter, $\epsilon \leq \eta_r\cdot e^{-\poly(\beta\sd)} $ this imposes the error in $\eqref{eq:stat-test-error-meta}\leq \eta_r e^{\poly(\beta\sd) - c_1\ell/2}$.\\

\noindent \textbf{Convexity of the stationarity test.} With $\mathbf{\Delta}:= \sum_{\gamma\in\CB[i,\ell]}
  \bigl(g_\gamma^{(r)}+\eta_r v_\gamma-h_\gamma\bigr)\vP_\gamma$ such that $\vK^{(\vV)}=\vH+\mathbf{\Delta}$ we then have
\begin{align}
  &\left|
    \mathsf{Test}_{i}[\mathbf{\Delta},\vH+\mathbf{\Delta};\sig]- \mathsf{Test}_{i}[\mathbf{\Delta},\vH;\sig] 
  \right| \\ \leq \,& \eta^2_r \cdot e^{\poly(\beta\sd)-c_1\cdot \ell/2}+ 2\eta_r |\CB[i, \ell]|\cdot \epsilon \\  \leq \, & \frac{e^{-\poly(\beta\sd)}}{32}\cdot \eta_r^2,
\end{align}
where we leveraged an appropriate choice of $\ell = \poly(\beta\sd)$, and again the assumption on $\epsilon $. This enables us to apply the convexity statement in \cref{lem:metastable-convexity}:
\begin{align}
    \max_{\gamma\ni i}
  |g_\gamma^{(r)}+\eta_r v_\gamma-h_\gamma|^2 \leq \frac{1}{4}\cdot \eta_r^2,
\end{align}
as advertised.

\end{proof}

\subsubsection{Convexity of the stationarity test (proof of \cref{lem:metastable-convexity}).}
\label{section:pdf-convexity-stest}

We dedicate this subsection to the proof of \cref{lem:metastable-convexity}. The proof essentially relies on understanding the continuity of the test around the point $\vG:=\vH+\vU$ with $\vU:=\sum_{\gamma\in \CB[i, \ell]} u_\gamma \vP_\gamma$. Henceforth, we will restrict our attention  to a single Pauli jump $\vA^a$, and temporarily suppress the labels $\widehat{\vA}_{\vG}:=\widehat{\vA}_{\vG}^{a}(\omega,t)$ and $\rd z:= g(t)h(\omega)\rd\omega \rd t\rd s$ for notational conciseness. 

The starting point to our argument is to consider a canonical formula for the stationarity test, in the form of a gradient flow equation \cite[Equation B7]{BCV25}:
\begin{align}
\mathsf{Test}_{a}[\vU,\vG;\sig] = \Tr(\sig\CL_{\vG,a}^{\dagger}[\vU])=
-\int
 \left\langle
 [\widehat{\vA}_{\vG}^{a},\vU],
 [\vrho_{\vG}^{-s}\widehat{\vA}_{\vG}^{a}
                  \vrho_{\vG}^{s},
  \log\vsigma+\beta\vG]
 \right\rangle_{\vsigma,s}\rd z .
 \label{eq:test-expansion}
\end{align}
which is a simple expansion of the Dirichlet form (\cref{lem:Dicirchlet}) with the following exact identity. 
\begin{lemma}
	[{\cite[Lemma B.4]{BCV25}}]\label{lemma:expose-log} For any full rank states $\vsigma, \vrho$ and any operator $\vA$:
	\begin{equation}
	\sqrt{\vrho} [\vA, \vrho^{-1/2}\vsigma \vrho^{-1/2}]\sqrt{\vrho} = \int_{-1/2}^{1/2} \vsigma^{\frac{1}{2}+s}[\vrho^{-s}\vA\vrho^{s},\log  \vsigma - \log  \vrho]\vsigma^{\frac{1}{2}-s} \rd s
	\end{equation}
\end{lemma}

The following lemma enables us to relate the outcome of the stationarity test on ``adjacent'' (local perturbations of the) Hamiltonians to the Fisher information functional $\QF$ of the Hamiltonian perturbation; as defined in \cref{def:relative-fisher-information}. This will later enable us to import the convexity results for the Fisher functional as derived in \cref{lemma:fisher-convexity}. 

\begin{lemma}
    [Continuity of the stationarity test]
\label{lem:metastable-differentiated-entropy-gradient}
Let $\vsigma$ be a full rank state, $\vA^a$ a single-qubit jump operator on site $i\in [n]$, and $\vK:=\sum_\Gamma v_\gamma \vP_\gamma$ a $(\sd, \sD, \scgr_{\sgr})$ Hamiltonian. For any $\beta\geq \sd^{-1}$, integers $b\geq 2\ell$ and Hamiltonian perturbation $\vU:=\sum_{\gamma\in \CB[i, \ell]}u_\gamma \vP_\gamma$:
\begin{align}
    &|\mathsf{Test}_{a}[\vU,\vK+\vU;\sig]
 -\mathsf{Test}_{a}[\vU,\vK;\sig]
 +\beta\QF_a(\vU,\vK+\vU;\sig)| \\
 \leq\,& e^{\poly(\beta\sd)}\cdot \|\vU\|^2\bigg( \|\log \vsigma+\beta\vK\|\cdot e^{-\mu\cdot (b-\ell)}  +e^{\poly(\beta\sd)\cdot |\sB[i,b]|}\cdot  \max_{\vP\in\CS_{\sB[i,b]}}\sqrt{\ADB_{\vK,\vP}[\vsigma]}\bigg)
\end{align}
\end{lemma}

\begin{remark}
    We phrase the above in terms of the Hamiltonian $\vK$ (as opposed to $\vH$ or $\vG$) as it will later represent a truncation of $\vH$.
\end{remark}

\begin{proof}

We begin with an expansion of the difference of stationarity tests \eqref{eq:test-expansion} at $\vG=\vK$ from the same
identity at $\vG=\vK+\vU$:
\begin{align}
\mathsf{Test}_{a}[\vU,\vK+\vU;\sig]
 -&\mathsf{Test}_{a}[\vU,\vK;\sig]
 =-\beta\QF_a(\vU,\vK+\vU;\sig)\\
&-\int
 \left\langle[\widehat{\vA}_{\vK+\vU}-\widehat{\vA}_{\vK},\vU],
 [\vrho_{\vK}^{-s}\widehat{\vA}_{\vK}\vrho_{\vK}^{s},\log\vsigma+\beta\vK]\right\rangle_{\vsigma,s}\rd z
\\
&
 -\int
 \left\langle[\widehat{\vA}_{\vK+\vU},\vU],
 [\vrho_{\vK+\vU}^{-s}\widehat{\vA}_{\vK+\vU}\vrho_{\vK+\vU}^{s}-\vrho_{\vK}^{-s}\widehat{\vA}_{\vK}\vrho_{\vK}^{s},\log\vsigma+\beta\vK]\right\rangle_{\vsigma,s}\rd z .
 \label{eq:differentiated-gradient-exact-difference}
\end{align}
We proceed by addressing each of these error terms separately. The second line in \eqref{eq:differentiated-gradient-exact-difference}  can be directly related to the approximate detailed-balance condition. Indeed, for any operator $\vB$, applying \cref{lemma:expose-log} in the reverse direction gives
\begin{align}
\int_{-1/2}^{1/2}
 \langle\vB,[\vrho_{\vK}^{-s}\widehat{\vA}_{\vK}
                    \vrho_{\vK}^{s},\log\vsigma+\beta\vK] \rangle_{\vsigma,s}\rd s &=
 \Tr\!\left[\vB^\dagger\sqrt{\vrho_{\vK}}
 \left[
  \widehat{\vA}_{\vK},
  \vrho_{\vK}^{-1/2}\vsigma\vrho_{\vK}^{-1/2}
 \right]
 \sqrt{\vrho_{\vK}}
 \right]. 
\end{align}
We note the second term above, after redistributing the factors of $\vrho_{\vK}$, correctly induces the ADB expression \cref{def:metastable-adb}. Thus the Cauchy-Schwarz inequality applied to the above with $\vB=[\widehat{\vA}_{\vK+\vU}-\widehat{\vA}_{\vK},\vU]$, gives
\begin{align}
&\bigg|\int
 \left\langle[\widehat{\vA}_{\vK+\vU}-\widehat{\vA}_{\vK},\vU],
 [\vrho_{\vK}^{-s}\widehat{\vA}_{\vK}\vrho_{\vK}^{s},\log\vsigma+\beta\vK]\right\rangle_{\vsigma,s}\rd z\bigg| \\\leq\quad & 2\cdot \sqrt{\ADB_{\vK,a}[\vsigma]}\cdot   \bigg|\int g(t)\bigg(\|\vB\sqrt{\vsigma}\|^2_2+\|\sqrt{\vsigma}\vB\|^2_2\bigg)\rd t\rd \omega\bigg|^{1/2}\\ \leq \quad &\|\vU\|^2\cdot \sqrt{\ADB_{\vK,a}[\vsigma]}\cdot 8\beta,
\end{align}
where we used the trivial bound on the frequency filters $h^2(\omega)/\gamma(\pm \omega)\leq e^{1/4}$, and in the second to third line above we leveraged \cref{lem:filtered-duhamel-estimates} (below this proof). 

It remains to bound the last integral in
\eqref{eq:differentiated-gradient-exact-difference}. To proceed, we use \cref{lem:bounds_imaginary_conjugation} to express the modular conjugation of the OFT, as an OFT with shifted frequency, and then subsequently apply the shell decomposition for perturbations of OFTs (\cref{lemma:oft-lipschitz}):
\begin{align}
 \vrho_{\vK+\vU}^{-s}\widehat{\vA}_{\vK+\vU}\vrho_{\vK+\vU}^{s}-\vrho_{\vK}^{-s}\widehat{\vA}_{\vK}\vrho_{\vK}^{s}
   &=\sum_{\chi\geq0}\vX_\chi(s,\omega,t),\\
\text{where}\quad \|\vX_\chi(s,\omega,t)\|
 &\leq  \|\vU\|\cdot e^{\poly(\beta\sd)+\theta(v|t|-\mu\cdot (\chi-\ell))}\cdot e^{s\beta \omega}
 \label{eq:differentiated-gradient-shells}
\end{align}
for any $\theta \in (0, 1)$. We make the explicit choice of $\theta = \min(1, \pi (10 v\beta)^{-1})$, which we note ensures the convergence of:
\begin{equation}
    \int_{\BR} g(t)e^{2\theta v|t|} \rd t\leq 4.
\end{equation}
Fix a length-scale cutoff $b>\ell$ (to be specified shortly). What follows is a case division over the two length-scale regimes. \\

\noindent \textbf{Case I: Distant $\chi\geq b$ shells.} First we consider the coefficients $\vX_\chi$ in the shell decomposition which are at radius $\chi\geq b$ away. As we shall see, these terms can be controlled trivially by an application of Schatten-Holder:
\begin{align}
     \bigg|\int \langle [\widehat{\vA}_{\vK+\vU}^{a}, \vU],  [\vX_\chi, \log \vsigma +\beta\vK]\rangle_{\vsigma, s}  \rd z\bigg| &\leq \|\vU\|\cdot  \|\log \vsigma+\beta\vK\|\cdot  \int \|\widehat{\vA}_{\vK+\vU}^{a}\|\cdot \|\vX_{\chi}\| \rd z \\ &\leq \|\vU\|^2\cdot  \|\log \vsigma+\beta\vK\|\cdot e^{\poly(\beta\sd) -\theta\mu\cdot (\chi-\ell)},
\end{align}
where we used the formula in \eqref{eq:differentiated-gradient-shells}, the exponentially decaying bound for $\|\widehat{\vA}_{\vK+\vU}^{a}\|$ in \cref{lem:bounds_oft}, and observed that in the integral over $\rd z:= g(t)h(\omega)\rd\omega \rd t\rd s$ the $\omega$ dependence $h(\omega)e^{\beta s\omega}\leq 1$. Summing the above over the shells $\chi\geq b$ converges with error decaying with $\|\vU\|^2$ and exponential in $b$.\\

\noindent\textbf{Case II: Close $\chi\leq b$ shells.} For the closeby shells, we denote $\vX_{\leq b}:=\sum_{\chi\leq b}\vX_\chi$ and expand pointwise in $(s,\omega,t)$ in the Pauli basis over $\sB[i,b]$:
\begin{equation}
 \vX_{\leq b}(s,\omega,t)
 =\sum_{\vP\in\CS_{\sB[i,b]}}
 c_{\vP}(s,\omega,t)\vP \Rightarrow \sum_{\vP}|c_{\vP}(s,\omega,t)|
 \leq 2^{4\sq |\sB[i,b]|}\|\vU\|\cdot e^{\poly(\beta\sd)+\theta v|t|}\cdot e^{s\beta \omega}
 \label{eq:differentiated-gradient-pauli-expansion}
\end{equation}
The  Cauchy--Schwarz inequality (twice), and integrating over the $\omega, t$ filter functions allows us to write the resulting error in terms of the commutator $[\vP,\log\vsigma+\beta\vK]$ of a worst-case choice of multi-qubit Pauli $\vP\in\CS_{\sB[i,b]}$:
\begin{align}
&\left|\int
 \left\langle[\widehat{\vA}_{\vK+\vU},\vU],
 [\vX_{\leq b},\log\vsigma+\beta\vK]
 \right\rangle_{\vsigma,s}\rd z\right| \\ \leq&
 e^{\poly(\beta\sd)\cdot |\sB[i,b]|}\|\vU\|^2
 \max_{\vP\in\CS_{\sB[i,b]}}
 \left(
   \int_{-1/2}^{1/2}
   \|[\vP,\log\vsigma+\beta\vK]\|_{\vsigma,s}^2\rd s
 \right)^{1/2}
 \label{eq:differentiated-gradient-close-reduction}\\
 \leq& \|\vU\|^2\bigg(e^{\poly(\beta\sd)\cdot |\sB[i,b]|}
 \max_{\vP\in\CS_{\sB[i,b]}}\sqrt{\ADB_{\vK,\vP}[\vsigma]}
 +e^{-\mu' |\sB[i,b]|} \|\log \vsigma+\beta\vK\|.\bigg)
\end{align}
\noindent where in the last line we applied the inequality from \cref{claim:multi-site-sint} under a judicious choice of constants (and we assume WLOG $|\CB[i, b]|\geq b$, otherwise the exponents above can be replaced by $\max(b,|\CB[i, b]| )$), and leverage the fact the integrals over $\omega, t$ converge to a constant which can be absorbed into the prefactor.\\

\noindent \textbf{Put together.} Combining the sources of error gives the bound:
\begin{align}
    &|\mathsf{Test}_{a}[\vU,\vK+\vU;\sig]
 -\mathsf{Test}_{a}[\vU,\vK;\sig]
 +\beta\QF_a(\vU,\vK+\vU;\sig)| \\
 \leq& e^{\poly(\beta\sd)}\cdot \|\vU\|^2\bigg( \|\log \vsigma+\beta\vK\|\cdot e^{-\theta \mu\cdot b} 
  +e^{\poly(\beta\sd)\cdot |\sB[i,b]|}\cdot  \max_{\vP\in\CS_{\sB[i,b]}}\sqrt{\ADB_{\vK,\vP}[\vsigma]}\bigg)
\end{align}
Summing over the single-qubit Pauli jump operators $a\in \CS_i^1$ gives the desired statement. 
\end{proof}

We are now in a position to conclude the proof of \cref{lem:metastable-convexity}, on the convexity of the stationarity test.

\begin{proof}[Proof of \Cref{lem:metastable-convexity}] We begin by introducing the length-scales $\mathsf{L}\gg b\gg \ell$ to be specified shortly. To shorten notation, we denote as $\mathsf{A}:=\mathsf{supp}(\CB[i, \mathsf{L}])\subseteq [n]$ the ball of radius $\mathsf{L}$ around $i$. We denote $\vsigma^{\lambda}_{\sA}:=(1-\lambda)\vsigma_{\sA}+\lambda\cdot \vrho_{\vH_{\CB[i, \mathsf{L}]}}$ as the reduced density matrix of $\vsigma$ on the patch $\sA$, smoothed by an interpolation with the Gibbs state of the induced Hamiltonian on the terms $\CB[i, \mathsf{L}]$, parametrized by $\lambda\in (0, 1)$. \\

\noindent \textbf{Step 1: Truncating and smoothing the stationarity test.} The locality and Lipschitz statements in \cref{lemma:lipschitz-within-shell} entail that we can truncate the stationarity test to the radius $\mathsf{L}$
\begin{align}
    &\mathsf{Test}_{i}[\vU,\vH+\vU;\vsigma] -  \mathsf{Test}_{i}[\vU,\vH;\vsigma] \\\approx &\mathsf{Test}_{i}[\vU,\vH_{\CB[i, \mathsf{L}]}+\vU;\vsigma] -  \mathsf{Test}_{i}[\vU,\vH_{\CB[i, \mathsf{L}]};\vsigma]\label{eq:truncated-st-1}
\end{align}
to error $\poly(\beta\cdot \sd)\cdot \|\vU\|^2\cdot e^{-c_1(\mathsf{L}-\ell)}$. 

We note that the terms on the RHS of \eqref{eq:truncated-st-1} depend only on the subregion $\sA$, and we can thus proceed by applying the continuity statement in \cref{lemma:lipschitz} to replace:
\begin{align}
    &\mathsf{Test}_{i}[\vU,\vH_{\CB[i, \mathsf{L}]}+\vU;\vsigma] -  \mathsf{Test}_{i}[\vU,\vH_{\CB[i, \mathsf{L}]};\vsigma]  \\ \approx &  \mathsf{Test}_{i}[\vU,\vH_{\CB[i, \mathsf{L}]}+\vU; \vsigma^{\lambda}_{\sA}] -  \mathsf{Test}_{i}[\vU,\vH_{\CB[i, \mathsf{L}]};\vsigma^{\lambda}_{\sA}], \label{eq:truncated-st-2}
\end{align}
now to error $\lambda\cdot \|\vU\|^2\cdot \poly(\beta\cdot \sd).$ To simplify further computations, here we make the explicit choice of $\lambda = e^{-c_1(\mathsf{L}-\ell)}$. \\

\noindent \textbf{Step 2: From the (smoothed) stationarity test to the Fisher functional.} Our intention is now to apply the continuity statement for the stationarity test derived in \cref{lem:metastable-differentiated-entropy-gradient} with the choice of Hamiltonian $\vK:=\vH_{\CB[i, \mathsf{L}]}$ and state $\vsigma_{\sA}^\lambda$, to relate \eqref{eq:truncated-st-2} to the Fisher information functional $\QF_i$ (\cref{def:relative-fisher-information}). For this purpose, we first compute the relevant $\mathsf{ADB}$ errors. \\

\noindent \textit{Computing the $\mathsf{ADB}$ errors.} Consider any site $j\in \sB[i, b]=\mathsf{supp}(\CB[i, b])$ in the inner radius $b$, and let $\vA_j\in \CS_{j}^1$ be a Pauli operator on $j$. The $\mathsf{ADB}$  (\cref{def:metastable-adb}) error of $\vsigma_{\sA}^\lambda$ under $\vH_{\CB[i, \mathsf{L}]}$ and the jump $\vA_j$:
\begin{align}
    \|\CL_{\vH_{\CB[i, \mathsf{L}]}, j}[\vsigma_{\sA}^\lambda]\|_1 &\leq \epsilon+\poly(\beta)\cdot e^{-c_1\cdot (\sL-b)}\\ \Rightarrow \quad \mathsf{ADB}_{\vH_{\CB[i, \mathsf{L}]},\vA_j}[\vsigma_{\sA}^\lambda] &\leq  e^{\poly(\beta\sd)}\cdot \bigg(|\sA|+\log \frac{1}{\lambda}\bigg)^2 \cdot \bigg(\epsilon + e^{-c_1\cdot (\sL-b)} + \lambda\bigg) \\ 
    &\leq e^{\poly(\beta\sd)}\cdot |\CB[i, \mathsf{L}]|^2\cdot e^{-c_1\cdot (\sL-b)}\label{eq:adb-lr-stc}
\end{align}
where we further assume $\epsilon\leq \epsilon^* = e^{-c_1(\sL-b)}$.
This now enables us to compute the multi-site \textsf{ADB} error within the region of radius $b$. From \cref{lem:metastable-multisite-to-single-site-adb}, for any $\Omega_2\geq 1$ and Pauli $\vP=\otimes_j^{\sB[i, b]}\vA_j\in \CS_{\sB[i, b]}$:
\begin{align}
\mathsf{ADB}_{\vH_{\CB[i, \mathsf{L}]},\vP}[\vsigma_{\sA}^\lambda]
& \leq e^{\poly(\beta\sd)\cdot |\sB[i, b]|}\cdot \bigg(e^{4\beta \Omega_2}\max_{j\in\sB[i, b]} \ADB_{\vH_{\CB[i, \mathsf{L}]},\vA_j}[\vsigma_{\sA}^\lambda]
      +e^{-\Omega_2/8\sd}\bigg)\\
      &\leq |\CB[i, \mathsf{L}]|^2\cdot  \exp\big[\poly(\beta\sd)\cdot |\sB[i, b]|-c_2\cdot \mathsf{L}\big]
\end{align}
By making the explicit choice of $\Omega_2 := \mathsf{L}/\poly(\beta\sd)$ which ensures $c_2\geq \poly(\beta\sd)^{-1}$, and imposing the assumption $\mathsf{L}\geq 2b$. \\

\noindent \textit{Returning to \eqref{eq:truncated-st-2}.} Now placed in the continuity statement \cref{lem:metastable-differentiated-entropy-gradient}, we have:
\begin{align}
    \mathsf{Test}_{i}[\vU,\vH_{\CB[i, \mathsf{L}]}+\vU; \vsigma^{\lambda}_{\sA}] -  \mathsf{Test}_{i}[\vU,\vH_{\CB[i, \mathsf{L}]};\vsigma^{\lambda}_{\sA}] \approx -\beta\QF_i(\vU,\vH_{\CB[i, \mathsf{L}]}+\vU;\vsigma^{\lambda}_{\sA})\label{eq:smoothing-qf}
\end{align}
to error $ \|\vU\|^2 \cdot  |\CB[i, \mathsf{L}]|^2\cdot e^{\poly(\beta\cdot \sd)-c_3\cdot b}\leq \|\vU\|^2 \cdot   e^{\poly(\beta\cdot \sd)-c_4\cdot b}$ so long as we chose a suitable $\sL = \poly(\beta\cdot \sd) \cdot b^\sD \gg |\sB[i, b]|$ and $b\geq 2\ell$, leveraging the fact the underlying graph is a $\sD$-dimensional lattice. \\

\noindent \textbf{Step 3: From the Fisher Information to the coefficients.} It follows from \cref{lem:uniform-truncated-bare-reconstruction} that the Fisher information functional $\QF_i$ satisfies for $c_1\propto (\beta \sd)^{-1}$ and   $\forall \Omega_1>0$:
\begin{align}
 \sum_{a\in \CS_i}\|[\vA^a, \vU]\|_{\vsigma^{\lambda}_{\sA}}^2\leq e^{\poly(\beta\sd)}\bigg(e^{\Omega_1} \QF_i(\vU,\vH_{\CB[i, \mathsf{L}]}+\vU;\vsigma^{\lambda}_{\sA}) + e^{-c_1\Omega_1}\cdot  \max_{\gamma} |u_\gamma|^2\bigg)
  \label{eq:metastable-fixed-core-reconstruction}
\end{align}

Since $[\vA^a, \vU]$ lies on the (local) neighborhood $\sB[i, 1]$ of site $i$, we proceed by invoking \cref{lem:metastab-faithfulness} to relate the KMS and Frobenius norms:
\begin{equation}
  \|[\vA^a, \vU]\|_{\vsigma^{\lambda}_{\sA}}^2
  \geq e^{-\poly(\beta\sd)}\cdot 
  \|[\vA^a,\vU]\|_{\vtau_{\sB[i, 1]}}^2
  -\|[\vA^a, \vU]\|^2 \cdot \max_{c\in \CS_{\sB[i, 1]}^1 }\sqrt{\mathsf{ADB}_{\vH_{\CB[i, \mathsf{L}]}, c}[\vsigma^\lambda_{\sA}]}
  \label{eq:metastable-fixed-core-kms-use}
\end{equation}
To combine the two identities above, we note 
\begin{equation}
   \sum_{a\in \CS_i}\|[\vA^a,\vU]\|_{\vtau_{\sB[i, 1]}}^2 \geq \max_{\gamma\ni i} |u_\gamma|^2, \quad  \sum_{a\in \CS_i}\|[\vA^a,\vU]\|_{\vtau_{\sB[i, 1]}}^2 \leq 2^{2\sq+1}\sd^2\cdot  \max_{\gamma\ni i} |u_\gamma|^2
  \label{eq:metastable-one-site-identity}
\end{equation}
Which then put together with the bound on the $\mathsf{ADB}$ error gives:
\begin{align}
    &e^{\Omega_1} \QF_i(\vU,\vH_{\CB[i, \sL]}+\vU;\vsigma^{\lambda}_{\sA}) + e^{-c_1\Omega_1}\cdot  \max_{\gamma} |u_\gamma|^2 \\\geq &e^{-\poly(\beta\sd)}\cdot \max_{\gamma \ni i}|u_\gamma|^2 \bigg(1-e^{\poly(\beta\sd)-c_1\cdot \sL/2}\bigg) \\
    \geq &e^{-\poly(\beta\sd)}\cdot \max_{\gamma \ni i}|u_\gamma|^2.
\end{align}
Where in the last line we make explicit choices of sufficiently large $b = \poly(\beta\cdot \sd)+\ell$ and thereby $\sL$. We conclude this step with the choice of an $\Omega_1=\poly(\beta\sd)$ which ensures:
\begin{align}
    \QF_i(\vU,\vH_{\CB[i, \sL]}+\vU;\vsigma^{\lambda}_{\sA})   \geq  e^{-\poly(\beta\sd)}\cdot \bigg(\max_{\gamma \ni i}|u_\gamma|^2 - \frac{1}{64} \max_{\gamma} |u_\gamma|^2\bigg).\label{eq:qf-to-coef}
\end{align}

\noindent \textbf{Put together.} Combining step 3 \eqref{eq:qf-to-coef}, step 2 \eqref{eq:smoothing-qf}, and the errors in step 1 \eqref{eq:truncated-st-1}, \eqref{eq:truncated-st-2}, gives the advertised bound.
\end{proof}

Finally, we prove the deferred statement on integrals over norms of the OFT and their differences: 

\begin{lemma} 
\label{lem:filtered-duhamel-estimates} For any pair of Hamiltonians $\vK, \vU$, and state $\vsigma$:
\begin{align}
    \iint g(t)\cdot \bigg\|[\widehat{\vA}_{\vK+\vU}(\omega, t)-\widehat{\vA}_{\vK}(\omega, t),\vU]\sqrt{\vsigma}\bigg \|^2_2\rd t\rd \omega \leq 16 \beta^2 \cdot \|\vU\|^4
\end{align}
where $g(t)$ is the filter function from the Dirichlet form \cref{lem:Dicirchlet}.
\end{lemma}

\begin{proof}
By Parseval's, we can peel off the frequency dependence in the integral:
\begin{align}
   &\iint g(t)\cdot \bigg\|[\widehat{\vA}_{\vK+\vU}(\omega, t)-\widehat{\vA}_{\vK}(\omega, t),\vU]\sqrt{\vsigma}\bigg \|^2_2\rd t\rd \omega \\ = & \iint g(t)\cdot  |f_{\beta^{-1}}(q)|^2\cdot \bigg\|[\vA_{\vK+\vU}( t+q)-\vA_{\vK}(t+q),\vU]\sqrt{\vsigma}\bigg \|^2_2\rd t\rd q \\
    \leq &  4\|\vU\|^2\cdot  \iint g(t)\cdot  |f_{\beta^{-1}}(q)|^2\cdot \bigg\|\vA_{\vK+\vU}( t+q)-\vA_{\vK}(t+q)\bigg \|^2\rd t\rd q\label{eq:meta-after-pars}
\end{align}
To proceed we apply Duhamel's formula to bound
\begin{equation}
    \big\|\vA_{\vK+\vU}(z)-\vA_{\vK}(z)\big \| \leq 2|z|\cdot \|\vU\|.
\end{equation}
And consequently, returning to \eqref{eq:meta-after-pars}:
\begin{align}
    \eqref{eq:meta-after-pars}&\leq 16\|\vU\|^4\cdot  \iint g(t)\cdot  |f_{\beta^{-1}}(q)|^2 (t+q)^2\rd t\rd q \\
    & = 16\|\vU\|^4\cdot \int g(t)\cdot (t^2+\beta^2/4) \rd t\\
    & \leq 16 \beta^2 \cdot \|\vU\|^4
\end{align}
from simple Gaussian moment bounds. 
\end{proof}

\subsection{Sample and time complexity}
\label{section:metastab-sample-and-time}

To conclude the proof of \cref{thm:end-to-end-metastable}, it only remains to argue about the time and sample complexity of the algorithm, which in turn is missing only a means to estimate the stationarity test on metastable states (an analog of \cref{lemma:stationarity-test-single-measurement}).

\begin{proof}

    [Proof of \cref{thm:end-to-end-metastable}] Here we simply highlight the differences to the statement of \cref{section:H-complexity}. We begin with an implementation of the stationarity test from short-time evolution in analog to \cref{lemma:stationarity-test-single-measurement}.
    
    During the $r$th round of the execution of the algorithm, suppose $\CL$ denotes the Lindbladian associated to the truncation of a guess Hamiltonian $\vG_{\CB[i, \mathsf{R}]}$ to radius $\sR$, where $\vG$ is coefficient-wise accuracy to $\vH$ to error $\eta$. It again follows from the Lieb-Robinson bound (e.g. \cref{lemma:lipschitz-within}), the continuity lemma  \cref{lemma:lipschitz}, and the definition of local metastability \cref{def:local-metastability}, that the derivative of $\vsigma$ under $\CL$ is given by:
\begin{align}
    \|\CL[\vsigma]\|_1
    &\leq
    \|\CL-\CL_{\vH,i}\|_{1\to1}  + \|\CL_{\vH,i}[\vsigma]\|_1\\
    &\leq \|\CL-\CL_{\vG,i}\|_{1\to1} +   \|\CL_{\vG,i}-\CL_{\vH,i}\|_{1\to1} + \epsilon \\
&\leq \poly(\beta\sd) \cdot \bigg(e^{-\theta\cdot \sR}+\eta\cdot \sum_{z\geq 0} |\CB[i, z]|\cdot e^{-\theta\cdot z}\bigg) + \epsilon \\ &\leq  \poly(\beta\sd)\cdot \eta \cdot \bigg(1+\sum_{z\geq 0} z^{\sD}\cdot e^{-\theta\cdot z}\bigg) +\epsilon \\
    &\leq \poly(\beta\sd) \cdot \eta,
\end{align}
so long as we appropriately choose the truncation radius $\sR = \poly(\beta\sd)(1+\log_+\frac{1}{\eta})$, and we recall the implicit assumption $\eta\geq e^{\poly(\beta\sd)}\epsilon$. The application of Duhamel's in equation \eqref{eq:duhamel-estimation-efficient} then implies we can estimate the desired stationarity test on $\CL_{\vG, i}$ by $e^{-\poly(\beta\cdot\sd)}$ time evolution of $\CL$. 

Correctness of the iterative algorithm is guaranteed by \cref{lemma:metastab-correctness}. Since asymptotically the algorithmic parameters $\ell, \kappa, c_{\mathsf{meas}}, R^{(r)}$ are unchanged from the case of Gibbs states \cref{lemma:gibbs-one-round}, the resulting sample and time complexity are inherited. Finally, following \cref{remark:rescale}, we rescale $\beta$ to the high temperature case if $\beta<\sd^{-1}$, which results in the claimed sample complexity and the decrease of the accuracy threshold $\eta_{\mathsf{thr}}\rightarrow \beta^{-1}\eta_{\mathsf{thr}}$.
\end{proof}

\printbibliography

\appendix

\section{Locality of \texorpdfstring{$\CL$}{L} \texorpdfstring{\cite{chen2023efficient}}{[CKG23]}}
\label{section:locality}

We dedicate this section to a collection of locality properties for the operator Fourier transform and the Lindbladian \cite{chen2023efficient}.

\subsection{Locality in time, space and frequency}

We begin by importing a simple statement on the real-time evolution of local observables under perturbations to the Hamiltonian.

\begin{lemma}[Time evolution under Hamiltonian perturbations
{\cite[Lemma~B.2]{chen2025learning}}]
\label{lem:can-real-time-perturbation}
Let $\vH:=\sum h_\gamma \vP_\gamma, \vG:=\sum g_\gamma \vP_\gamma$ be Hamiltonians s.t. $|h_\gamma|, |g_\gamma|\leq 1$, defined on the same degree $\leq \sd$
interaction graph. If \(\vA\) is a single-site Pauli operator, then 
\begin{equation}
  \left\|
    e^{\mathrm{i}t\vH}\vA e^{-\mathrm{i}t\vH}
    -
    e^{\mathrm{i}t\vG}\vA
      e^{-\mathrm{i}t\vG}
  \right\|
  \le
  \sum_\gamma
  |h_\gamma-g_\gamma|
  \min\!\left\{
    4\sd|t|,
    \frac{(4\sd|t|)^{\mathsf{dist}(\gamma,\vA)+1}}
         {(\mathsf{dist}(\gamma,\vA)+1)!}
  \right\}.
  \label{eq:can-real-time-perturbation}
\end{equation}
Here \(\mathsf{dist}(\gamma,\vA)\) is the interaction-graph distance from
\(\gamma\) to the support of \(\vA\).
\end{lemma}

Next, we leverage Lieb-Robinson bounds and the convergence of the operator norm of the OFT, to establish locality properties for the OFT which decay in both space and in frequency. We note this is a mild generalization of \cite[Lemma VI.2]{BC26} in the 1D case to higher-dimensional systems.

\begin{lemma}
[Locality in Space and Frequency]\label{lemma:oft-loc} In the context of \cref{def:oft}, let $\vA$ be a single site Pauli operator on a site $i\in [n]$ and $\vG$ a degree $\sd$ Hamiltonian. For $\sigma>0$ denote the $\ell\geq 0$ truncation to the OFT $\hat{\vA}_{\vG, \sigma}^{(\ell)}(\omega):=\hat{\vA}_{\vG_{\CB[i, \ell]}, \sigma}(\omega)$ with $\hat{\vA}_{\vG, \sigma}^{(-1)}(\omega):=0.$ Then,
\begin{align}
     \hat{\vA}_{\vG, \sigma}(\omega)
    &=
    \sum_{\ell\geq 0}
        \hat{\vA}_{\vG, \sigma}^{(\ell)}(\omega)- \hat{\vA}_{\vG, \sigma}^{(\ell-1)}(\omega) \quad \text{where} \\ \| \hat{\vA}_{\vG, \sigma}^{(\ell)}(\omega)- \hat{\vA}_{\vG, \sigma}^{(\ell-1)}(\omega)\|&\leq \frac{c_{\mathrm{loc-oft}}}{\sqrt{\sigma}}\cdot e^{-|\omega|/8\sd-\mu \ell/2}
\end{align}
where $c_{\mathrm{loc-oft}}:=\mathsf{c_{LR}}^{1/2}2^{5/4}\exp[\frac{\sigma^2}{32\sd^2}+\frac{v^2}{8\sigma^2}]$ is a constant which depends only on $\sigma, \sd$ and the Lieb-Robinson constants $\mathsf{c_{LR}}, \mu, v$. 
\end{lemma}

\begin{proof}
    It follows by \cref{lem:bounds_oft,lem:norm_bound_pauli} that the $\ell$th term above satisfies the apriori norm bound
    \begin{equation}
    	\|\hat{\vA}_{\vG, \sigma}^{(\ell)}(\omega)\|\leq \frac{2}{\sqrt{\sigma \sqrt{2\pi}}} \exp\bigg(\frac{\sigma^2}{16\sd^2}-\frac{|\omega|}{4\sd}\bigg).
    \end{equation}
    Furthermore, by a LR bound we can establish locality in space. For conciseness, let $\vG^{(\ell)}:=\vG_{\CB[i, \ell]}$. From the telescoping formulation of the LR bound in \cref{lem:lieb-robinson-commutator}, it follows that
    \begin{align}
    \|\hat{\vA}_{\vG, \sigma}^{(\ell)}(\omega)- \hat{\vA}_{\vG, \sigma}^{(\ell-1)}(\omega)\|&\leq \frac{1}{\sqrt{2\pi}} \int |f_{\sigma}(t)|\cdot \| e^{i\vG^{(\ell)} t}\vA e^{-i\vG^{(\ell)} t}-e^{i\vG^{(\ell-1)} t}\vA e^{-i\vG^{(\ell-1)} t}\| \rd t  \\ &\leq \frac{c_{LR} e^{-\mu \ell}}{\sqrt{2\pi}}\cdot   \int f_{\sigma}(t) e^{v|t|} \rd t \\
    & \leq c_{LR} e^{-\mu \ell}\cdot \frac{2^{3/4}}{\pi^{1/4}\sqrt{\sigma }}\cdot e^{v^2/4\sigma^2}.
    \end{align}
    If $0\leq x\leq \min(a, b)$ then $x\leq \sqrt{ab}$ which then gives the desired bound. Finally, the LR constants depend only on the degree of the Hamiltonian.
\end{proof}

\subsection{Lipschitz continuity properties}

Building on these quasi-locality properties, next we combine the two lemmas above in order to establish Lipschitz continuity of the shell decomposition of the OFT, under perturbations to the Hamiltonian. 

\begin{lemma}
    [Lipschitz continuity of the OFT]\label{lemma:oft-lipschitz} Let $\vH:=\sum h_\gamma \vP_\gamma, \vG:=\sum g_\gamma \vP_\gamma$ be $(\sd,\sD, \scgr)$ Hamiltonians defined on the same interaction graph. Suppose $\vU:=\vG-\vH$ has coefficients supported in $ \CB[i, \ell]$ for some $i\in [n]$ and $\ell\geq 1$. Then, for $\sigma>0$ the OFT of a single-site operator $\vA^a, \|\vA^a\|\leq 1$ supported on site $i$ admits the telescoping decomposition:
    \begin{align}
      \forall t, \omega\in \BR,\quad \theta\in (0, 1]:\quad  \widehat{\vA}_{\vG, \sigma}^{a}(\nu,t) -
 \widehat{\vA}_{\vH, \sigma}^{a}(\nu,t)
& =
 \sum_{b\geq\ell}\vZ_b, \\ \|\vZ_b\|&\leq \|\vU\|\cdot \mathsf{c_{oft-Lip}}\cdot  e^{\theta(2v|t|-\mu \cdot (b-\ell))}
    \end{align}
    where $ \supp(\vZ_b)\subseteq\sB[i,b]$, and $ \mathsf{c_{LR}}, \mu, v$ are the Lieb-Robinson constants \cref{lem:lieb-robinson-commutator}. Furthermore, $\mathsf{c_{oft-Lip}}:= \mathsf{c}\cdot  \sigma^{-1/2}(\sigma^{-1} + (\theta v)^{-1})e^{v^2/\sigma^2}$, for a constant $\mathsf{c}$ which depends only on $\sd, \mathsf{c_{LR}}$.
\end{lemma}

\begin{proof}
    Consider the interpolation $\vK_x:=\vH + x(\vG-\vH)$ between $\vG$ and $\vH$. A telescoping decomposition gives:
    \begin{align}
      \vZ_b(\nu,t) = \int_0^1
 \left(
   \partial_x
 \widehat{\vA}_{\vK_x^{(b)}}^{a}(\nu,t)
   -
    \partial_x
 \widehat{\vA}_{\vK_x^{(b-1)}}^{a}(\nu,t)
 \right)\rd x\label{eq:z-def-oft}
    \end{align}
    where $\vK_x^{(b)}:=(\vK_x)_{\CB[i, b]}$ and thus $ \supp(\vZ_b)\subseteq\sB[i,b]$ by construction. The base case $b=\ell$ simply omits the second summand. To understand the derivative, we use Duhamel's formula on the time-evolution:
   \begin{equation}
     \partial_x\vA_{\vK_x}(q)
    =
    \mathrm{i}
    \left[
        \int_0^q(\vG-\vH)_{\vK_x}(r)\,\rd r,
       \vA_{\vK_x}(q)
    \right],
\end{equation}
whose quasi-locality follows from a Lieb-Robinson bound (\cref{lem:lieb-robinson-commutator}) applied to $\vA_{\vK_x}(q)$ and $\vU_{\vK_x}(r)$. The assumption $b\geq \ell$ and that $\vU$ is local further enables one to remove factors of volume\footnote{We remark this removal is not strictly necessary for our applications but makes the statements somewhat cleaner.} in $\ell$, by expanding $\vU=\sum_{\gamma}u_\gamma \vP_\gamma$ and leveraging the fact $\sum_{\gamma}|u_\gamma|\leq \poly(\sd)\cdot \|\vU\|$ see e.g. \cite[Theorem 7 and Section 7]{HPT23}. Then, 
\begin{align}
    \bigg\| \partial_x
 \vA_{\vK_x^{(b)}}(q)
   -
    \partial_x
 \vA_{\vK_x^{(b-1)}}(q)
\bigg\| \leq |q|\cdot e^{\theta(v|q|-\mu (b-\ell))} \cdot \|\vU\| \cdot c_1\label{eq:derivs-error}
\end{align}
for a suitable constant $c_1$ as a function of $\mathsf{c_{LR}}, \sd$. Finally, we integrate over $q$ in \eqref{eq:z-def-oft} to arrive at the desired bound: 
\begin{align}
 \|\vZ_b(\nu,t)\|
 &\leq \frac{c_1\|\vU\|e^{-\theta\mu(b-\ell)}}{\sqrt{2\pi}}
       \int_{\BR}|f_\sigma(s)|\,|t+s|e^{\theta v|t+s|}\,\rd s
       \label{eq:a3-oft-integral}\\
 &\leq \frac{c_1\|\vU\|}{\sqrt{\sigma}}
       \left[\left(\frac2\pi\right)^{1/4}|t|
             +\left(\frac2\pi\right)^{3/4}\frac1\sigma\right]
       e^{(\theta v)^2/(2\sigma^2)+\theta v|t|-\theta\mu(b-\ell)}
       \label{eq:a3-gaussian-moments}\\
 &\leq \|\vU\|\,\mathsf{c_{oft-Lip}}\,
       e^{\theta(2v|t|-\mu(b-\ell))}.
       \label{eq:a3-oft-integrated}
\end{align}
where in the last line we used $|x|e^{-a|x|}\leq(ea)^{-1}$.
\end{proof}

We can now finally arrive at the desired Lipschitz continuity properties for the Lindbladian generator $\CL$ from \cite{chen2023efficient}. 
\begin{lemma}
	[Lipschitz Continuity of $\CL$]\label{lemma:lipschitz} Let $\vH:=\sum_{\gamma\in \Gamma} h_\gamma \vP_\gamma, \vG:=\sum_{\gamma\in \Gamma} g_\gamma \vP_\gamma$ be Hamiltonians defined on the same degree $\leq \sd$ interaction graph. Fix a site $i\in [n]$, a single-site Pauli operator $\vA^a$ on $i$, and $\sigma=\beta^{-1}$. Then, 
    \begin{align}
        \|\CL^\dagger_{\vH, a}-\CL^\dagger_{\vG, a}\|_{\infty-\infty} \leq 32\beta\sd\cdot \sum_{\gamma\in \Gamma} |g_\gamma-h_\gamma|\cdot e^{-\theta \cdot \mathsf{dist}(i, \gamma)}
    \end{align}
    where $\theta:=\min\{1,(8e\beta\sd)^{-1}\}$.
\end{lemma}

 The proof is based on the computation in \cref{lem:can-real-time-perturbation} from \cite{chen2025learning} and bounds on the filter function integrals. We highlight certain special cases in the sequence. 

\begin{proof} 

[of \cref{lemma:lipschitz}]
Let us fix $\sigma=\beta^{-1}$ and
$\theta:=\min\{1,(8e\beta\sd)^{-1}\}$. It follows from the Lipschitz continuity of the time evolution
Lemma~\ref{lem:can-real-time-perturbation}, using
$\min\{1,z^m/(m+1)!\}\leq e^{\theta(ez-m)}$, that
\begin{equation}\label{eq:a4-time-short}
\begin{aligned}
 \|\vA_{\vG}(t)-\vA_{\vH}(t)\|
 &\leq 4\sd|t|e^{|t|/(2\beta)}
       \sum_{\gamma\in\Gamma}|h_\gamma-g_\gamma|
                       e^{-\theta\mathsf{dist}(i,\gamma)}=: |t|e^{|t|/(2\beta)} \cdot \mathsf{\Delta}
\end{aligned}
\end{equation}
where for conciseness we introduced the error parameter $\mathsf{\Delta}$. We next treat the dissipative and coherent components of the Lindbladian \cref{def:ckg-lind} separately.  \\

\noindent \textbf{Dissipative component.} The dissipative contribution ends up as the easier of the two cases. By the operator Cauchy-Schwarz inequality, followed by the operator Parseval's identity (\cref{lem:OParseval}): 
\begin{equation}\label{eq:a4-dissipative-short}
\begin{aligned}
 \|\mathcal D_{\vG, a}^\dagger-\mathcal D_{\vH, a}^\dagger\|_{\infty\to\infty}
 &\leq4\left\|\int_{\mathbb R}
       (\hat{\vA}_{\vG}(\omega)-\hat{\vA}_{\vH}(\omega))^\dagger
        (\hat{\vA}_{\vG}(\omega)-\hat{\vA}_{\vH}(\omega))\,\rd\omega\right\|^{1/2}\\
 &\leq 4\left(\int_{\mathbb R}|f_{\beta^{-1}}(t)|^2\cdot 
                        \|\vA_{\vG}(t)-\vA_{\vH}(t)\|^2\,\rd t\right)^{1/2}\\
 &\leq 4\mathsf{\Delta}\cdot 
       \left(\int_{\mathbb R}|f_{\beta^{-1}}(t)|^2t^2
                                 e^{|t|/\beta}\,\rd t\right)^{1/2}
 \leq 4\mathsf{\Delta}\cdot\beta.
\end{aligned}
\end{equation}
where, in the above, we leveraged $0\leq\gamma_M\leq1$
and $\|f_\sigma\|_2=1$.\\

\noindent \textbf{Coherent component.} The coherent component is more delicate, due to the divergence in the filter function $c(t)\propto 1/\sinh at$. Accordingly, we symmetrize $\pm t$ the perturbations to the OFT in their contributions to the perturbations to $\vC_{\vG}-\vC_{\vH}$. Let us first introduce some notation. For fixed $t>0$ and any function $g:\BR\rightarrow\BR$, let
\begin{equation}
    f_{\pm}(u):=f_{\beta^{-1}}(u\mp t), \quad \delta f:=f_+-f_-, \quad \hat{\vA}_{\vK, g}(\omega):= \frac{1}{\sqrt{2\pi}}\int g(u)\cdot  e^{-i\omega u} \cdot \vA_{\vK}(u) \rd u.
\end{equation}
By Gaussian integration we note the filter function difference satisfies $\|\delta f\|_2^2\leq 4\min(1, t^2/\beta^2)$. We next turn our attention to the quantity $\vQ_{\vK}(t):=\int_{\BR}\gamma_M(\omega)
\widehat{\vA}_{\vK}(\omega,t)^\dagger
\widehat{\vA}_{\vK}(\omega,t)\,\rd\omega$, present in the coherent term $\vC_{\vK}$. By expansion (and temporarily suppressing the $\omega, t$ dependence),
\begin{align}
    &\bigg\|(\vQ_{\vG}(t)- \vQ_{\vH}(t)) -   (\vQ_{\vG}(-t)- \vQ_{\vH}(-t))\bigg\| \label{eq:222-qdelta}\\
    =&\bigg\|\int \rd \omega \cdot  \gamma(\omega)\bigg((\hat{\vA}_{\vG, \delta f}-\hat{\vA}_{\vH, \delta f})^\dagger\cdot \hat{\vA}_{\vG, f_+} + \hat{\vA}_{\vH, \delta f}^\dagger(\hat{\vA}_{\vG, f_+}- \hat{\vA}_{\vH, f_+}) \\
    &\qquad\qquad+(\hat{\vA}_{\vG, f_-}-\hat{\vA}_{\vH,  f_-})^\dagger\cdot \hat{\vA}_{\vG, \delta f} + \hat{\vA}_{\vH, f_-}^\dagger(\hat{\vA}_{\vG, \delta f}- \hat{\vA}_{\vH,  \delta f})\bigg)\bigg\| \\
    \leq & 2 \mathsf{\Delta} \bigg|\int \rd u \cdot  \delta f(u)^2\cdot |u|^2 e^{|u|/\beta}\bigg|^{1/2} + 2 \mathsf{\Delta} \|\delta f\|_2\cdot \bigg|\int \rd u \cdot  f_-(u)^2\cdot |u|^2 e^{|u|/\beta}\bigg|^{1/2} \\
    \leq & 12\mathsf{\Delta}\,t(1+t/\beta)^2 e^{t/(2\beta)}
\end{align}
where in the second to third line we similarly apply operator Cauchy-Schwarz and Parseval's \cref{lem:OParseval} and subsequently evaluate the Gaussian moments.

Finally, we arrive at the bound on the coherent term:
\begin{align}
     \|\vC_{\vG}-\vC_{\vH}\|
 &\leq\int_0^\infty c(t)\cdot 
           \bigg\|(\vQ_{\vG}(t)- \vQ_{\vH}(t)) -   (\vQ_{\vG}(-t)- \vQ_{\vH}(-t))\bigg\|\rd t\\ &\leq 12\mathsf{\Delta}\int_0^\infty
          \frac{t(1+t/\beta)^2e^{t/(2\beta)}}{\beta\sinh(2\pi t/\beta)}\,\rd t\\
 &=12\beta\mathsf{\Delta}\int_0^\infty
          \frac{s(1+s)^2e^{s/2}}{\sinh(2\pi s)}\,\rd s
 \leq 2\beta\mathsf{\Delta}.
\end{align}
where the final inequality is evaluated numerically. 
\end{proof}

We next record two special cases of the Lipschitz continuity lemma above. The first follows trivially by ignoring the decay outside of a radius $\ell$:

\begin{corollary}
    [$\vH, \vG$ agree outside a neighborhood] \label{lemma:lipschitz-outside}
    In the context of \cref{lemma:lipschitz}, let $\ell\geq 1$, $\eta\in [0, 2]$ and suppose $\max_{\gamma\in \Gamma} |h_\gamma-g_\gamma| \leq \eta$ and that $\forall\gamma\notin \CB[i, \ell], h_\gamma=g_\gamma$. Then, 
    \begin{equation}
        \|\CL^\dagger_{\vH, a}-\CL^\dagger_{\vG, a}\|_{\infty-\infty}\leq 32\beta\sd \cdot \eta\cdot |\CB[i, \ell]|.
    \end{equation}
\end{corollary}

Next, we consider the effect of distant perturbations on the Hamiltonian to the Lindbladian; when we further suppose the underlying Hamiltonian lies on a lattice (or, at least admits polynomial growth).

\begin{corollary}
    [$\vH, \vG$ agree within a neighborhood]\label{lemma:lipschitz-within} Fix $\beta>\sd^{-1}$.
    In the context of \cref{lemma:lipschitz}, let $\ell\geq 1$, $\eta\in [0, 2]$ and suppose $\max_{\gamma\in \Gamma} |h_\gamma-g_\gamma| \leq \eta$ and that $\forall\gamma\in \CB[i, \ell], h_\gamma=g_\gamma$. Furthermore, assume $\vH, \vG$ are $(\sd, \sD, \scgr_{\sgr})$ Hamiltonians as in \cref{def:polynomial-growth-hamiltonian}. Then, 
    \begin{equation}
        \|\CL^\dagger_{\vH, a}-\CL^\dagger_{\vG, a}\|_{\infty-\infty}\leq \scgr_\sgr\cdot \mathsf{c_{D}}\cdot (\beta\sd)^{\sD+1}\cdot  \eta\cdot  e^{-\ell\cdot \theta/2},
    \end{equation}
    where $\scgr_{\sD}$ is a function only of $\sD$, and $\theta$ is the constant from \cref{lemma:lipschitz}.
\end{corollary}

\begin{proof}
    Follows by evaluating the sum in \cref{lemma:lipschitz}:
    \begin{align}
 \sum_{\gamma: \mathsf{dist}(i,\gamma)>\ell}e^{-\theta\cdot  \mathsf{dist}(i,\gamma)}
 &\leq e^{-\theta\ell/2}\frac{\theta}{2}\cdot 
       \int_\ell^\infty e^{-\theta r/2}\cdot 
       |\CB[i,r]|\cdot \rd r\\
 &\leq \scgr_\sgr e^{-\theta\ell/2}\frac{\theta}{2} \int_0^\infty e^{-\theta r/2}\cdot 
      r^\sD\cdot \rd r \leq \scgr_\sgr\cdot  \sD!\cdot (2/\theta)^\sD e^{-\theta\ell/2}.
\end{align}
\end{proof}

Finally, we consider the effect of truncating a local perturbation to the Hamiltonian to a given radius, which in a sense combines the two above bounds. 
\begin{lemma}
    [A shell decomposition for \cref{lemma:lipschitz-within}]\label{lemma:lipschitz-within-shell}
    In the context of \cref{lemma:lipschitz}, let $\beta\sd\geq 1,$ $\ell\geq 1$ and suppose $\vH, \vG=\vH+\vU$ are degree $\sd$ Hamiltonians on the same interaction graph such that $\forall\gamma\notin \CB[i, \ell], h_\gamma=g_\gamma$. For $b\geq \ell$  let $\vH_b:=\vH_{\CB[i, b]}, \vG_b:=\vG_{\CB[i, b]}$ denote the truncation of the Hamiltonians to radius $b$.
    Then, 
    \begin{align}
        \bigg\| (\CL^\dagger_{\vG_{b+1}, a}-\CL^\dagger_{\vH_{b+1}, a}) - (\CL^\dagger_{\vG_{b}, a}-\CL^\dagger_{\vH_{b}, a})\bigg\|_{\infty-\infty} \leq \poly(\beta\cdot \sd)\cdot  \|\vU\|\cdot e^{-c_1(b-\ell)}
    \end{align}
    where $c_1\propto (\beta\sd)^{-1}$.
\end{lemma}

\begin{proof}
    We let $x\in [0, 1]:\vK_{x, b}:=\vH_b + x(\vG_b-\vH_b)$ be the interpolation from \cref{lemma:oft-lipschitz}, and expose a telescoping decomposition for the time evolution via Duhamel's formula:
    \begin{align}
        \bigg\|(\vA_{\vG_{b+1}}(t)-\vA_{\vH_{b+1}}(t)) -  (\vA_{\vG_{b}}(t)-\vA_{\vH_{b}}(t))\bigg\| &\leq \int_0^1 \|\partial_x \vA_{\vK_{x, b+1}}(t)-\partial_x \vA_{\vK_{x, b}}(t)\| \rd x \\ &\leq c_2|t|e^{\theta(v|t|-\mu(b-\ell))} \|\vU\|
    \end{align}
    where the last inequality was computed in \eqref{eq:derivs-error}, for $c_2$ a polynomial function of $\mathsf{c_{LR}}, \sd$ and we can choose $\theta:=\min(1, (8e\beta v)^{-1})$. Crucially, the above decays in both distance and $\|\vU\|$.   
    The rest of the proof is just careful integration of filter functions akin to that of \cref{lemma:lipschitz}. 
\end{proof}

\newpage

\section{Refinements to the metastability area law}
\label{section:area-law}

In this section we prove \cref{thm:intro-metastability-area-law}, on an area law for metastable states which holds in the thermodynamic limit. Our approach will be to combine ingredients from \cite{BCV25} with the quasi-locality of the Fisher information described in \cref{section:fisher-identifiability}. 

\subsection{Local decoders imply an area law}

As input, we require two building blocks from \cite{BCV25}, the first of which is the following local Markov property for metastable states.

\begin{lemma}[A local Markov property for metastable states {\cite[Theorem D.1]{BCV25}}]
\label{lem:bcv-area-inputs}
Let $\sA\subseteq [n]$ and let $\vsigma$ be a full rank state which is $\epsilon-$locally metastable w.r.t. $\vH$ as in \cref{def:intro-local-metastability}. Then, there exists a recovery channel $\CR_\sA$ which fixes $\vrho_{\vH}$ s.t. for any noise channel $\CN_\sA$ acting on $\sA$,
\begin{equation}
 \|\vsigma-\CR_{\sA}\circ \mathcal N_{\sA}[\vsigma]\|_1
 \leq e^{\mu|\sA|} \cdot \max_{a\in \CS_{\sA}^1} \mathsf{ADB}_{\vH, a}[\vsigma]^{\lambda},
 \label{eq:bcv-D1}
\end{equation}
Where $\mu, \lambda^{-1}\leq \poly(\beta, \sd)$.

\end{lemma}

We further require a minor modification to the implication that local decoding implies an area law (see e.g. \cite[Lemma IV.1]{BCV25}), in the form: 

\begin{lemma}
[An Area law from local recovery, with side information]
\label{lem:bcv-area-law-passive-reference}
Let \(\sA\subseteq\sB\), and let \(\vG\) be
supported on \(\sB\). Suppose that a channel \(\CR\) acting only on \(\sB\) fixes $\vrho_{\vG}\otimes\vsigma_{\bar{\sB}}$ and approximately recovers 
\begin{equation}
    \|(\CR\otimes \mathsf{id}_{\bar{\sB}})[\vsigma_{\sA}\otimes\vsigma_{\sA^c}]-\vsigma\|_1\leq\eta\leq\frac12.
\end{equation}
Then, $\vsigma$ satisfies an approximate area law:
\begin{equation}
    I(\sA:\sA^c)_{\vsigma}
    \leq
    2\beta\|\partial_{\sA:\sB}\vG\|
    +4\eta\max\left\{
        \log\frac1\eta,\,
        \beta\|\vG\|,\,
        |\sB| \sq
    \right\}.
    \label{eq:bcv-IV1-passive-reference}
\end{equation}
\end{lemma}

\begin{proof}

[of \cref{lem:bcv-area-law-passive-reference}]
Since $\sA\subseteq \sB$, the states $\vsigma$ and $\vsigma_{\sA}\otimes\vsigma_{\bar{\sA}}$ admit the same marginal on $\bar{\sB}$. Thus we can directly expand the relative entropies. For any $\valpha$ such that $\valpha_{\bar{\sB}}=\vsigma_{\bar{\sB}}$:
\begin{equation}
    D(\valpha\|\vrho_{\vG}\otimes\vsigma_{\sB^c}) = -S(\sB\mid\sB^c)_{\valpha}
    +\beta\Tr[\vG\valpha_{\sB}]
    +\log\tr[e^{-\beta\vG}].\label{eq:rel-exp}
\end{equation}
and moreover by subtracting the two,
\begin{equation}
\begin{split}
    D(\vsigma\|\vrho_{\vG}\otimes\vsigma_{\bar{\sB}})-D(\vsigma_{\sA}\otimes\vsigma_{\bar{\sA}}\|\vrho_{\vG}\otimes\vsigma_{\bar{\sB}})
    &=
    I(\sA:\bar{\sA})_{\vsigma}
      +\beta\Tr\!\left[\vG(\vsigma-\vsigma_{\sA}\otimes\vsigma_{\bar{\sA}}\right]\\
      &\geq  I(\sA:\bar{\sA})_{\vsigma}-2\beta\cdot \|\partial_{\sA:\sB}\vG\|.
\end{split}
\label{eq:passive-reference-free-energy}
\end{equation}
where \(\partial_{\sA:\sB}\vG\) is the sum of interactions meeting both
\(\sA\) and \(\sB\setminus\sA\). To proceed, we note that the state $\CR[\vsigma_{\sA}\otimes\vsigma_{\bar{\sA}}]$ similarly admits the same RDM on $\vsigma_{\bar{\sB}}$. Therefore, we can apply continuity of the conditional entropy (\cref{fact:alicki}) in \eqref{eq:rel-exp}:
\begin{align}
    &D(\vsigma\|\vrho_{\vG}\otimes\vsigma_{\bar{\sB}})-D(\CR[\vsigma_{\sA}\otimes\vsigma_{\bar{\sA}}]\|\vrho_{\vG}\otimes\vsigma_{\bar{\sB}}) \\ \leq  &|S(\sB\mid\bar{\sB})_{\vsigma}-S(\sB\mid\bar{\sB})_{\CR[\vsigma_{\sA}\otimes\vsigma_{\bar{\sA}}]}| 
    +\beta \|\vG\|\cdot \|\vsigma_{\sB}-\CR[\vsigma_{\sA}\otimes\vsigma_{\bar{\sA}\cap \sB}]\|_1 \\
    \leq &4\eta\cdot\max\bigg\{\log \frac{1}{\eta},\beta \|\vG\|, |\sB|\cdot  \sq \bigg\} \label{eq:area-law-rec-error}
\end{align}
Since the recovery map $\CR$ fixes $\vrho_{\vG}\otimes\vsigma_{\bar{\sB}}$, application of \eqref{eq:area-law-rec-error} proceeded by the data-processing inequality gives:
\begin{equation}
    D(\vsigma||\vrho_{\vG}\otimes\vsigma_{\bar{\sB}})
    \approx
    D\!\left(
      (\CR_{\sB}\otimes\mathsf{id}_{\bar{\sB}})[\vsigma_{\sA}\otimes\vsigma_{\bar{\sA}}]
      \|
      \vrho_{\vG}\otimes\vsigma_{\bar{\sB}}
    \right)
    \leq D(\vsigma_{\sA}\otimes\vsigma_{\bar{\sA}}\|\vrho_{\vG}\otimes\vsigma_{\bar{\sB}}).
    \label{eq:passive-reference-dpi}
\end{equation}
Rearranging these bounds into \eqref{eq:passive-reference-free-energy} gives the statement. 

\end{proof}

In the above we leveraged the following continuity statement for the quantum conditional entropy. 

\begin{fact}
    [Alicki--Fannes inequality]\label{fact:alicki} For any two states $\vrho_1, \vrho_2$ defined on a bipartite Hilbert space $\calH_{\sB}\otimes \calH_{\bar{\sB}}$, with $u:=\frac12\|\vrho_1-\vrho_2\|_1$.
\begin{equation}
    \left|
      S(\sB\mid\bar{\sB})_{\vrho_1}
      -S(\sB\mid\bar{\sB})_{\vrho_2}
    \right|
    \leq
    2u\log\mathsf{dim}(\calH_\sB)
    +(1+u)\cdot h_2\!\left(\frac{u}{1+u}\right).
\end{equation}
with $h_2$ the binary entropy function.
\end{fact}

\subsection{Proof \texorpdfstring{of \cref{thm:intro-metastability-area-law}}{}}

We are now in a position to conclude the proof of the area law in \cref{thm:intro-metastability-area-law}.

\begin{proof}
Assume \(\epsilon\leq1/4\), since otherwise the claim follows after
increasing \(\mu\). Fix a set \(\sA\), and let us fix a radius $r:=1+\left\lceil c_0^{-1}\log_+(C_0/\epsilon)\right\rceil$ and $\sB:=\sB[\sA, r]$. Then, by a Lieb-Robinson bound \cref{lemma:lindblad-LR} and the definition of locally metastable, $\vsigma$ is metastable w.r.t. the Lindbladian truncated to the region $\sB$ up to error:
\begin{equation}
 \|\CL_{\sA}^{(\sB)}[\vsigma]\|_1
 \leq
 |\sA|\bigl(\epsilon+C_0e^{-c_0r}\bigr)
 \leq
 2|\sA|\epsilon
 \label{eq:local-area-truncation}
\end{equation}
We proceed by smoothing $ \vsigma^\delta:=(1-\delta)\vsigma+\delta\vrho_{\vH_\sB}\otimes\vsigma_{\bar{\sB}}$. Then in PSD order,
\begin{equation}
    \vsigma
 \leq 2^{\sq|\sB|}\vI_{\sB}\otimes\vsigma_{\bar{\sB}}
 \leq 2^{2\sq|\sB|}e^{2\beta\|\vH_{\sB}\|}\vrho_{\vH_\sB}\otimes\vsigma_{\bar{\sB}},
\end{equation}
 which consequently entails a bound on the operator norm of the log-likelihood ratio, by operator monotonicity of the matrix logarithm:
\begin{align}
     \|\log\vsigma^\delta-\log(\vrho_{\vH_\sB}\otimes\vsigma_{\sB^c}^\delta)\|
 =O\bigg(
 |\sB|\sq+\beta \|\vH_{\sB}\|+\log(1/\delta)\bigg)
 =:\Delta_\delta.
 \label{eq:local-area-log}
\end{align}
By H\"older's inequality, the Fisher Information of \(\vsigma^\delta\)
under \(\CL_{\sA}^{(\sB)}\) is at most
\( 2|\sA|\epsilon\Delta_\delta\). \cref{lem:metastable-FI-to-ADB} then yields
\begin{equation}
 \max_{a\in\CS^1_{\sA}}
 \ADB_{\vH_\sB,a}[\vsigma^\delta]
 \leq
 2|\sA|\epsilon\Delta_\delta
\cdot\bigg (1+ 
   \log_+\frac{ \Delta_\delta}{ 2|\sA|\epsilon}
 \label{eq:local-area-adb}\bigg)
\end{equation}
To proceed, we make the assumption the Hamiltonian lies on a lattice / admits polynomial growth and thereby $|\sB|\leq |\sA| \cdot \poly \log \frac{1}{\epsilon}$. We further make the explicit choice of $\delta = \epsilon$ which ensures:
\begin{equation}
 \max_{a\in\CS^1_{\sA}}
 \ADB_{\vH_\sB,a}[\vsigma^\delta]
 \leq \epsilon^{\lambda}\poly(|\sA|)
 \label{eq:local-area-adb2}
\end{equation}
for suitable choice of $\lambda$. Applying \cref{lem:bcv-area-inputs}, and the triangle inequality over $\delta$,  gives a channel \(\CR\), acting only on \(\sB\) and fixing
$\vrho_{\vH_\sB}\otimes\vsigma_{\sB^c}$, which recovers $\vsigma$ from erasures up to additive error
\begin{equation}
 \|\vsigma-\CR\circ\mathcal N_{\sA}[\vsigma]\|_1
 \leq 4\delta+
 e^{\mu_2|\sA|}\cdot \max_{a\in\CS^1_{\sA}}
 \ADB_{\vH_\sB,a}[\vsigma^\delta]^{\lambda_2} \leq e^{\mu_3|\sA|}\cdot \epsilon^{\lambda_3}
 \label{eq:local-area-recovery}
\end{equation}
for an adequate choice of constants \(\mu_3,\lambda_3>0\). Placed into \cref{lem:bcv-area-law-passive-reference} with $\vG = \vH_\sB$ gives the desired bound. 
\end{proof}

\end{document}